\documentclass[twoside]{article}

\usepackage{amsmath}
\usepackage{amsfonts}
\usepackage{amssymb}
\usepackage{amsthm}
\usepackage{bm}
\usepackage{booktabs}
\usepackage{graphicx}
\usepackage{subcaption}
\usepackage[dvipsnames]{xcolor}

\usepackage{algorithm}
\usepackage{algpseudocode}

\usepackage{hyperref}

\usepackage[capitalise]{cleveref}

\hypersetup{
    pdftitle={SLICE: Sequential Likelihood-Free Inference via Conformal E-Values},
    pdfauthor={Stefano Cortinovis, Laura Battaglia, Fran\c{c}ois Caron},
    linkcolor=blue,
    filecolor=magenta,
    urlcolor=cyan,
}

\newtheorem{theorem}{Theorem}[section]
\newtheorem{proposition}[theorem]{Proposition}

\newtheorem{corollary}[theorem]{Corollary}

\theoremstyle{definition}

\theoremstyle{remark}

\newcommand{\PP}{\mathbb P}
\newcommand{\EE}{\mathbb E}
\newcommand{\dd}{\mathrm d}

\newcommand{\calB}{{\mathcal{B}}}

\newcommand{\calF}{{\mathcal{F}}}

\newcommand{\calY}{{\mathcal{Y}}}

\newcommand{\thetatrue}{\theta_\ast}

\DeclareMathOperator*{\argmin}{arg\,min}

\DeclareMathOperator{\KL}{KL}

\def\[#1\]{\begin{align}#1\end{align}}

\usepackage[inline]{enumitem}
\newlist{enuminline}{enumerate*}{1}
\setlist[enuminline]{label=\textbf{\arabic*})}

\usepackage[preprint]{aistats2027}

\usepackage[numbers, compress]{natbib}

\begin{document}

\runningauthor{Cortinovis, Battaglia, Caron}

\twocolumn[

\aistatstitle{SLICE: Sequential Likelihood-Free Inference via Conformal E-Values}

\aistatsauthor{ Stefano Cortinovis$^{*}$ \And Laura Battaglia$^{*}$ \And Fran\c{c}ois Caron }

\aistatsaddress{} ]

{\renewcommand{\thefootnote}{}\footnotetext{$^{*}$Equal contribution.
All authors: Department of Statistics, University of Oxford.
Correspondence to: Stefano Cortinovis \texttt{<cortinovis@stats.ox.ac.uk>}.}}

\begin{abstract}
    Many scientific models can be simulated but have intractable likelihoods, making frequentist uncertainty quantification difficult. We construct exact anytime-valid confidence sequences for the parameters of such models. At each observation, fresh simulations from each candidate null conformalize a positive predictable score into a conditional e-value; their product is a test martingale, which we invert over the parameter. Validity holds for any such score, whatever its accuracy. We show that the oracle likelihood-ratio score maximizes expected log growth, which motivates approximating it with an amortized neural ratio estimator. On synthetic benchmarks, an audit of model routing in an LLM API, and microlensing data, the resulting confidence sequences stay valid under continuous monitoring, where fixed-sample methods do not, and neural scores remain competitive with exact likelihood ratios and application-specific scores.
\end{abstract}

\section{INTRODUCTION}

Many scientific models are easiest to express as simulators.
Given a parameter $\theta$, a computer program can generate synthetic observations $Y\sim P_\theta$, but evaluating the likelihood $f_\theta(y)$ may be impossible or prohibitively expensive.
This setting, often called likelihood-free or simulation-based inference, occurs in population genetics, ecology, epidemiology, cosmology or climate science \citep{Cranmer2020}.
Inference replaces likelihood evaluation with repeated simulation, using methods such as neural density estimation or classification-based likelihood-ratio estimation. These methods can be effective, but their uncertainty quantification is usually approximate or designed for a prespecified sample size
\citep{Hermans2020,Thomas2022,Miller2022,Dalmasso2020,Dalmasso2024}.

We study \emph{anytime-valid frequentist uncertainty quantification} in likelihood-free models. For observations $Y_1,Y_2,\ldots \overset{\mathrm{i.i.d.}}{\sim} P_{\thetatrue}$ arriving sequentially, we seek a sequence of random sets $(C_{n,\alpha})_{n\ge 1}$ forming a $(1-\alpha)$ confidence sequence (CS)~\citep{Lai1976,Howard2021}, i.e., satisfying
\begin{equation}
    \PP_{\thetatrue}\!\left(\thetatrue\in C_{n,\alpha}\text{ for all }n\ge 1\right)
    \ge 1-\alpha
    \label{eq:cs-definition}
\end{equation}
for every data generating parameter $\thetatrue$. Unlike a
collection of fixed-sample confidence sets, a CS remains valid under
continuous monitoring and at
data-dependent stopping times. Existing frequentist approaches to
likelihood-free inference primarily construct and calibrate
fixed-sample tests or confidence regions
\citep{Dalmasso2020,Dalmasso2024}. Obtaining the time-uniform guarantee
in \eqref{eq:cs-definition} requires a sequentially valid construction.
When likelihoods are available, confidence sequences can be obtained, via Ville's inequality \citep{Ville1939},
by inverting likelihood-ratio test martingales~\citep{Ramdas2023,Ramdas2025}.
In implicit models the ratios cannot be evaluated, and simply substituting a neural estimate
generally destroys the martingale property on which the guarantee relies.

We develop \textbf{SLICE}: Sequential Likelihood-Free Inference
via Conformal E-Values. SLICE turns amortized neural ratio estimates into anytime-valid confidence sequences. It builds on normalized-score conformal e-values
\citep{Gammerman1998,Balinsky2024,Vovk2025ConformalE}
and their sequential or simulation-based extensions
\citep{Gauthier2025,dombowsky2026besag}.

For each candidate parameter, we evaluate the incoming observation
and fresh simulations from the candidate model using the same
positive score function, chosen using only past information.
Dividing the observed score by the average score across the
observation and its references gives a \emph{normalized-score
e-value}. Multiplying these factors over time provides a test
martingale, whose inversion yields a confidence sequence. Coverage holds for any positive predictable score and any number of reference observations.

The score nevertheless matters for efficiency. We assess it
through expected log growth, a measure of how quickly evidence
accumulates against false candidates. Specializing the log-optimality
results of \citet{Koning2026}, we identify the likelihood ratio as
an optimal score within the normalized-score construction.
A classification-based proof connects this result to contrastive
learning \citep{Oord2018,Durkan2020} and bounds the efficiency lost to finitely many references. In practice, we approximate the optimal score with an amortized neural ratio estimator \citep{Hermans2020,Miller2022}.

We evaluate SLICE on Gaussian and $g$-and-$k$ location models and the SLCP benchmark, where neural scores yield confidence sequences close to conformalized likelihood ratios. We then apply it to a semi-synthetic audit of model routing in an LLM API, using real model replies, and to inference of the stellar initial mass function from OGLE-IV microlensing events. In both, SLICE stays valid under continuous monitoring, unlike fixed-sample intervals, with neural score confidence sequences performing comparably to those using application-specific scores. 

\section{RELATED WORK}
\label{sec:related}

\textbf{Simulation-based conformal e-values.}
Normalizing scores over an exchangeable batch yields conformal
e-values \citep{Vovk2025ConformalE}.
Closely related, \citet{dombowsky2026besag} develop
simulation-based normalized-score e-values, confidence-set
inversion, and sequential products, emphasizing
Besag--Clifford MCMC while also treating independent null
simulation. Our oracle-score result specializes the
log-optimality of normalized likelihood ratios under
permutation invariance \citep{Koning2026}, and our
finite-reference bound specializes the log-growth comparison
of \citet{dombowsky2026besag}.
SLICE builds on these by combining direct null
simulation with predictable alternatives and amortized neural
ratio scores for sequential likelihood-free parameter
inference.

\textbf{Frequentist simulation-based inference.}
ACORE and LF2I construct fixed-sample tests and confidence
regions using simulation-based calibration
\citep{Dalmasso2020,Dalmasso2024}.
WALDO similarly calibrates a Wald-type statistic built from
prediction algorithms or posterior estimators
\citep{Masserano2023}.
TRUST and TRUST\texttt{++} use conformal calibration of statistical
confidence sets \citep{Cabezas2026}.
These approaches address fixed-sample inference.
SLICE instead uses fresh simulations from each candidate null
to normalize the score of every incoming observation,
then inverts the resulting test martingales.
Under the stated sampling conditions, this gives time-uniform
coverage for any positive predictable score and reference-batch size, without estimating critical values.

\textbf{Neural ratio estimation.}
Our scores build on classification-based amortized ratio
estimation and contrastive learning
\citep{Hermans2020,Durkan2020,Miller2022}.
We use a likelihood-to-evidence network to construct ratios
between a predictable alternative and each candidate null.
The resulting neural score need not be an exact likelihood
ratio or itself a valid e-value: conditional exchangeability
and score normalization supply validity, while score quality
affects efficiency. This permits approximate neural models to
support exact sequential inference.
Further connections to posterior calibration, density-ratio
estimation, and sequential testing are discussed in
\cref{app:related}.

\textbf{Concurrent work.}
After the submission of an earlier version of this manuscript, an independent concurrent work by \citet{Forre2026} appeared on arXiv, studying likelihood-free anytime-valid hypothesis testing.
Both works use the normalized-score construction, building on earlier literature \cite[cf.~Lemma 3.1 and Eqs.~(56)--(58)]{Forre2026}, and establish closely related oracle optimality and finite-reference log-growth bounds (our \cref{prop:optimal-score}; \cite[Proposition 3.2]{Forre2026}).
\citet{Forre2026} additionally provide theoretical guarantees for approximate growth optimality of empirically trained scores.
Our emphasis is on parameter confidence sequences obtained by test inversion, parameter-amortized neural ratio scores with predictable adaptive alternatives, and extensive empirical evaluation on synthetic benchmarks and applications.

\section{PRELIMINARIES}

\paragraph{Notation and statistical setting.}

Let $\{P_\theta:\theta\in\Theta\}$ be a family of
probability distributions on $(\calY,\calB)$, dominated
by a common $\sigma$-finite measure $\mu$, with densities
$f_\theta=\dd P_\theta/\dd\mu$.
The observations are $Y_1,Y_2,\ldots\overset{\mathrm{i.i.d.}}{\sim}P_{\thetatrue}$, where $\thetatrue\in\Theta$ is fixed and unknown.
For each $\theta\in\Theta$, write $\PP_\theta$ for the
joint law of the observations and all auxiliary randomness
under $H_\theta:\thetatrue=\theta$, and $\EE_\theta$
for expectation under this law.

Let $(\calF_n)_{n\geq0}$ be the filtration generated by
the observations and auxiliary information revealed
through time $n$ which, in the context of this paper, will include simulator draws.
The initial sigma-field $\calF_0$ contains pretraining
and initialization information. We assume that $Y_i$ is
independent of $\calF_{i-1}$ for every $i\geq1$.

Finally, let $r_{\theta_1,\theta_0}(y)=f_{\theta_1}(y)/f_{\theta_0}(y)$ be the likelihood ratio, where $f_{\theta_0}(y)>0$.

\subsection{Test martingales and confidence sequences}
\label{sec:backgroundCS}

For a candidate parameter $\theta$, a nonnegative, adapted, integrable process
$(M_n^{(\theta)})_{n\geq0}$, initialized at $M_0^{(\theta)}=1$,
is a \emph{test supermartingale} under $H_\theta$ if
\begin{equation*}
\EE_\theta
\!\left[M_n^{(\theta)}\mid\mathcal F_{n-1}\right]
\leq M_{n-1}^{(\theta)}
\quad \PP_\theta\text{-a.s.},\qquad n\geq1.
\end{equation*}
It is a \emph{test martingale} when equality holds \citep{Shafer2011}.

A useful construction multiplies \emph{conditional e-values}:
nonnegative, $\calF_n$-measurable factors $E_n^{(\theta)}$
satisfying
\begin{equation}
\EE_\theta\!\left[
    E_n^{(\theta)}\mid\calF_{n-1}
\right]\leq1.\label{eq:conditionalevalue}
\end{equation}
Their product $M_n^{(\theta)}=\prod_{i=1}^n E_i^{(\theta)}$ is a test supermartingale, and is a test martingale if each
factor has conditional expectation one.

For every $\alpha\in(0,1)$, Ville's inequality
\citep{Ville1939} gives
\begin{equation}
\PP_\theta
\!\left(
\sup_{n\geq0}M_n^{(\theta)}\geq\alpha^{-1}
\right)
\leq\alpha.
\label{eq:ville}
\end{equation}
Consequently, if $(M_n^{(\theta)})_{n\geq0}$ is a test
supermartingale under $\PP_\theta$ for every $\theta\in\Theta$,
then $C_{n,\alpha}=\{\theta\in\Theta:M_n^{(\theta)}<\alpha^{-1}\}$, $n\geq1$, forms a $(1-\alpha)$ confidence sequence \citep{Lai1976,Howard2021}, that is, satisfies \eqref{eq:cs-definition} for every $\thetatrue\in\Theta$.
See \cref{app:coverage} for the coverage argument.

\paragraph{Likelihood-ratio test martingale.}
For a fixed alternative $\vartheta\in\Theta$, taking
$E_i^{(\theta)}=r_{\vartheta,\theta}(Y_i)$ gives
\begin{equation}
M_n^{(\theta)}
=
\prod_{i=1}^n r_{\vartheta,\theta}(Y_i),
\qquad M_0^{(\theta)}=1.
\label{eq:lr-martingale}
\end{equation}
Under $H_\theta$, each factor has conditional
expectation one given $\calF_{i-1}$, so this process
is a test martingale \citep{Ramdas2025}.

The choice $\vartheta=\thetatrue$ gives the
\emph{oracle likelihood-ratio martingale}, whose
expected log growth is
\begin{equation*}
\frac1n\EE_{\thetatrue}\!
\left[\log M_n^{(\theta)}\right]
=
\KL(P_{\thetatrue}\Vert P_\theta),
\qquad n\geq1,
\end{equation*}
provided this divergence is finite. This oracle choice maximizes expected
log growth under $\PP_{\thetatrue}$ among all
$(\calF_n)$-adapted test martingales under $\PP_\theta$
initialized at one \citep[Theorem~7.11]{Ramdas2025}. This benchmark requires the unknown parameter
$\thetatrue$ and likelihood ratios that are unavailable
in our setting.
Substituting learned ratios generally loses the
conditional expectation property \eqref{eq:conditionalevalue}, motivating the
normalization developed in
\cref{sec:likelihood-free-martingale}.

\subsection{Neural ratio estimation}
\label{sec:backgroundNRE}

Neural ratio estimation (NRE) learns likelihood ratios from
simulator draws by classification, without evaluating the
underlying densities \citep{Cranmer2015,Hermans2020}.
 In amortized NRE, a single
network $s_\phi(\theta,y)$ takes a simulator parameter value
$\theta$ and an observation $y$ as inputs; $\phi$ denotes
its trainable weights.
Given a training density $q$ over $\Theta$, define the marginal density $m(y)=\int_\Theta f_\vartheta(y)q(\vartheta)\,\dd\vartheta$.

We use NRE-C \citep{Miller2022}, which trains a multiclass
classifier to identify which candidate parameter generated
a simulated observation, with an additional class in which
the observation is independent of all candidates.
Its ideal population target is the log
likelihood-to-evidence ratio
\begin{equation}
s^\star(\theta,y)
=
\log\frac{f_\theta(y)}{m(y)}.
\label{eq:lte-target}
\end{equation}
Consequently, differences of network outputs provide
likelihood-ratio estimates:
\begin{equation*}
\exp\!\left\{
s_\phi(\theta_1,y)-s_\phi(\theta_0,y)
\right\}
\approx
r_{\theta_1,\theta_0}(y)
=
\frac{f_{\theta_1}(y)}{f_{\theta_0}(y)},
\end{equation*}
since the marginal density cancels.
This is the construction used for the neural scores in
\cref{sec:neural-scores}.
The training objectives and alternative estimators,
including NRE-A and NRE-B, are discussed in
\cref{app:nre-losses}; implementation details are given
in \cref{app:training}. 
\section{SEQUENTIAL UNCERTAINTY QUANTIFICATION FOR LIKELIHOOD-FREE MODELS}
\label{sec:sequential-uq}

We first construct a likelihood-free test martingale by
normalizing a predictable score using fresh null simulations,
and invert these martingales to obtain confidence sequences.
We then characterize the score that maximizes expected log
growth and use this characterization to motivate a practical
implementation based on amortized neural ratio estimation.

\subsection{A likelihood-free test martingale}
\label{sec:likelihood-free-martingale}

Fix a candidate $\theta\in\Theta$ and an integer $K\geq1$.
At time $i$, choose a positive measurable score
\[
T_i^{(\theta)}:\calY\to(0,\infty)
\]
using only $\calF_{i-1}$, before observing $Y_i$ or its reference simulations. We call such a score \emph{predictable}.

Generate a reference batch
\begin{equation}
\widetilde Y_{i1}^{(\theta)},\ldots,
\widetilde Y_{iK}^{(\theta)}
\overset{\mathrm{i.i.d.}}{\sim}P_\theta,
\label{eq:referencesimulations}
\end{equation}
independently of $(\calF_{i-1},Y_i)$. All reference batches generated through time $i$, across all
candidate parameters, are included in the auxiliary randomness
recorded in $\calF_i$. Under $H_\theta$, conditional on $\calF_{i-1}$, the observation
and its reference batch are therefore exchangeable draws from
$P_\theta$.

Following the simulation-based normalized-score construction
\citep{Vovk2025ConformalE,dombowsky2026besag}, define
\begin{align}
E_i^{(\theta)}
&=
\frac{(K+1)T_i^{(\theta)}(Y_i)}
{T_i^{(\theta)}(Y_i)
+\sum_{j=1}^K
T_i^{(\theta)}(\widetilde Y_{ij}^{(\theta)})},
\label{eq:main-evalue}\\
M_0^{(\theta)}&=1,
\qquad
M_n^{(\theta)}=\prod_{i=1}^n E_i^{(\theta)}.
\nonumber
\end{align}
Each factor compares the observed score with the average
score across the observation and its references.
Conditional exchangeability makes this factor a conditional
e-value, as established below.

\begin{proposition}[Exact validity]
\label{prop:exact-martingale}
For every $\theta\in\Theta$ and every sequence of positive
predictable scores $(T_i^{(\theta)})_{i\geq1}$, the process
$(M_n^{(\theta)})_{n\geq0}$ defined by
\eqref{eq:referencesimulations} and \eqref{eq:main-evalue} is a test martingale under $H_\theta$.
\end{proposition}

\begin{proof}
Fix $\theta$ and condition on $\mathcal F_{i-1}$ under
$H_\theta$. Set $Z_{i0}=Y_i$ and
$Z_{ij}=\widetilde Y_{ij}^{(\theta)}$ for $j=1,\ldots,K$.
The batch is conditionally exchangeable, and
$T_i^{(\theta)}$ is fixed given the past. Hence the
normalized scores
\begin{equation*}
W_{ij}
=
\frac{T_i^{(\theta)}(Z_{ij})}
{\sum_{\ell=0}^K T_i^{(\theta)}(Z_{i\ell})},
\qquad j=0,\ldots,K,
\end{equation*}
have equal conditional expectations and sum to one.
Therefore,
\begin{equation*}
\EE_\theta[E_i^{(\theta)}\mid\mathcal F_{i-1}]
=
(K+1)\EE_\theta[W_{i0}\mid\mathcal F_{i-1}]
=1.
\end{equation*}
The product construction in \cref{sec:backgroundCS}
gives the test-martingale property.
\end{proof}

Consequently, for every $\alpha\in(0,1)$,
\begin{equation}
C_{n,\alpha}
=
\{\theta\in\Theta:M_n^{(\theta)}<\alpha^{-1}\}
\label{eq:likelihood-free-cs}
\end{equation}
satisfies the time-uniform coverage guarantee
in \eqref{eq:cs-definition}. The reference simulations
may share randomness across candidate parameters:
the proof requires conditional exchangeability for
each candidate under its own null, rather than
independence across candidates.

Validity holds for every positive predictable score.
The score determines expected log growth against false
candidates, which we study next.

\subsection{The optimal score}
\label{sec:optimal-score}

We assess a score through the expected log growth of
its test martingale under $P_{\thetatrue}$. For a fixed positive measurable
score $T$, define
\begin{equation}
\Gamma_K^{(\theta)}(T)
=
\EE\!\left[
\log
\frac{(K+1)T(X_0)}
{\sum_{j=0}^K T(X_j)}
\right],
\label{eq:finite-k-growth}
\end{equation}
where $X_0\sim P_{\thetatrue}$ and
$X_1,\ldots,X_K\overset{\mathrm{i.i.d.}}{\sim}P_\theta$
are mutually independent. If the same score $T$ is used
at every step, then
\begin{equation*}
\frac1n\EE_{\thetatrue}[\log M_n^{(\theta)}]
=
\Gamma_K^{(\theta)}(T).
\end{equation*}
Within the simulation-based normalized score construction \eqref{eq:main-evalue}, an \emph{oracle score} maximizes this quantity for testing
$H_\theta$ against $P_{\thetatrue}$.

The oracle characterization below specializes the
log-optimality of normalized likelihood ratios for
permutation invariance \citep{Koning2026}.
We give its classification proof to make the connection
with InfoNCE explicit \citep{Oord2018,Poole2019}.
The finite-$K$ bound also follows by specializing the
log-growth comparison of \citet{dombowsky2026besag}.

\begin{proposition}[Oracle score and finite-$K$ gap]
\label{prop:optimal-score}
Assume
$\KL(P_{\thetatrue}\Vert P_\theta)<\infty$.
For every positive measurable score $T$ with
$\Gamma_K^{(\theta)}(T)>-\infty$,
\begin{align*}
\Gamma_K^{(\theta)}(T)
&\leq
\Gamma_K^{(\theta)}(r_{\thetatrue,\theta})\leq
\KL(P_{\thetatrue}\Vert P_\theta).
\end{align*}
The first inequality is an equality whenever
$T=c\,r_{\thetatrue,\theta}$ for some $c>0$. Define
\begin{equation*}
\Delta_K^{(\theta)}
=
\KL(P_{\thetatrue}\Vert P_\theta)
-
\Gamma_K^{(\theta)}(r_{\thetatrue,\theta}).
\end{equation*}
 $\Delta_K^{(\theta)}$ measures the loss in expected
log growth relative to the oracle likelihood-ratio
martingale of \cref{sec:backgroundCS}
If
\begin{equation*}
\chi^2(P_{\thetatrue}\Vert P_\theta)
:=
\int
\{r_{\thetatrue,\theta}(y)-1\}^2
\,P_\theta(\dd y)
<\infty,
\end{equation*}
then
\begin{equation}
\begin{aligned}
0\leq\Delta_K^{(\theta)}
&\leq
\log\!\left(
1+\frac{\chi^2(P_{\thetatrue}\Vert P_\theta)}{K+1}
\right)\\
&\leq
\frac{\chi^2(P_{\thetatrue}\Vert P_\theta)}{K+1}.
\end{aligned}
\label{eq:finite-k-chi-square-bound}
\end{equation}
In particular, under this moment condition,
$\Gamma_K^{(\theta)}(r_{\thetatrue,\theta})
\to\KL(P_{\thetatrue}\Vert P_\theta)$ as $K\to\infty$.
\end{proposition}

\begin{proof}
Write $r=r_{\thetatrue,\theta}$.
Consider a classification problem with
$J$ uniform on $\{0,\ldots,K\}$. Conditional on $J$,
draw $X_J$ from $P_{\thetatrue}$ and all remaining
points independently from $P_\theta$. For a score $T$,
define
\begin{equation*}
\pi_T(j\mid X_{0:K})
=
\frac{T(X_j)}{\sum_{\ell=0}^K T(X_\ell)}.
\end{equation*}
By symmetry,
\begin{equation*}
\Gamma_K^{(\theta)}(T)
=
\log(K+1)+\EE[\log\pi_T(J\mid X_{0:K})].
\end{equation*}
Bayes' rule gives $\pi_r$ as the true conditional
distribution of $J$. Consequently,
\begin{align*}
&\Gamma_K^{(\theta)}(r)-\Gamma_K^{(\theta)}(T)\\
&\quad=
\EE\!\left[
\KL\!\left(
\pi_r(\cdot\mid X_{0:K})
\Vert
\pi_T(\cdot\mid X_{0:K})
\right)
\right]\geq0.
\end{align*}
Multiplying $r$ by a positive constant leaves $\pi_r$
unchanged. The oracle objective is finite because
the conditional entropy of $J$ lies between zero
and $\log(K+1)$.

For the gap, return to the sampling scheme in
\eqref{eq:finite-k-growth} and set
\begin{equation*}
S_K=\frac1{K+1}\sum_{j=0}^K r(X_j).
\end{equation*}
Write $\EE_{\mathrm{null}}$ for expectation when all
$K+1$ points are instead independent draws from
$P_\theta$. A change of measure and symmetry give
\begin{align*}
\Delta_K^{(\theta)}
&=\EE[\log S_K]\\
&=\EE_{\mathrm{null}}[r(X_0)\log S_K]\\
&=\EE_{\mathrm{null}}[S_K\log S_K]\geq0.
\end{align*}
The final inequality follows from convexity and
$\EE_{\mathrm{null}}[S_K]=1$.
If the chi-square divergence is finite, then under
the original sampling scheme,
\begin{equation*}
\EE[S_K]
=
1+\frac{\chi^2(P_{\thetatrue}\Vert P_\theta)}{K+1}.
\end{equation*}
Jensen's inequality gives
$\Delta_K^{(\theta)}\leq\log\EE[S_K]$.
The remaining bound follows from $\log(1+x)\leq x$.
\end{proof}

The classification interpretation explains why the
likelihood ratio is optimal for every $K$: it assigns
the correct probability to each point being the draw
from $P_{\thetatrue}$. The same principle underlies
contrastive ratio estimation \citep{Durkan2020}.
Moreover, $0\leq\Gamma_K^{(\theta)}(r_{\thetatrue,\theta})\leq\log(K+1)$, so a finite reference batch limits the expected
log evidence obtainable from each observation.
Equation~\eqref{eq:finite-k-chi-square-bound} quantifies the resulting gap.

\subsection{Neural ratio scores}
\label{sec:neural-scores}

The oracle score depends on the unknown $\thetatrue$
and on likelihoods that cannot be evaluated.
We approximate it in two steps: replace $\thetatrue$
by an estimate based on past information, and estimate
the resulting likelihood ratio with a neural network.
Specifically, define
\begin{equation}
T_i^{(\theta)}(y)
=
\exp\!\left\{
s_{\phi_{i-1}}(\widehat\theta_{i-1},y)
-
s_{\phi_{i-1}}(\theta,y)
\right\},
\label{eq:neural-score}
\end{equation}
where $s_\phi$ is the likelihood-to-evidence network
of \cref{sec:backgroundNRE}, and both
$\widehat\theta_{i-1}$ and $\phi_{i-1}$ are
$\mathcal F_{i-1}$-measurable.

At the population target \eqref{eq:lte-target},
the term $-\log m(y)$ cancels in the score difference.
Thus the target of \eqref{eq:neural-score} is
$r_{\widehat\theta_{i-1},\theta}(y)$. 

We train the likelihood-to-evidence network
$s_\phi$ using the NRE-C objective \citep{Miller2022}
on independent simulations before monitoring, and then
keep it fixed. At each step, we choose the alternative
$\widehat\theta_{i-1}$ by minimizing the accumulated
evidence over the evaluation grid. This choice is related
to ME-estimation and e-posterior methods
\citep{Wang2026,Gruenwald2023}.
Other predictable estimates and network updates are also permitted.

Table~\ref{tab:evalue-constructions} summarizes the
proposed conditional e-value and the likelihood-ratio based
benchmarks used in our experiments.

\begin{table}[t]
\centering
\small
\setlength{\tabcolsep}{3pt}
\caption{
Conditional e-values for candidate $\theta$.
Write $\widetilde Y_{i0}^{(\theta)}=Y_i$;
the remaining $K$ observations are fresh reference
draws from $P_\theta$.
The neural score $T_i^{(\theta)}$ is defined in
\eqref{eq:neural-score}.
Oracle LR and oracle score are the oracle variants of Direct and Conformal LR, respectively.
}
\label{tab:evalue-constructions}
\begin{tabular}{@{}ll@{}}
\toprule
Construction & Conditional e-value $E_i^{(\theta)}$ \\
\midrule

Oracle LR (optimal)
&
$\displaystyle r_{\thetatrue,\theta}(Y_i)$
\\[1ex]

Direct LR
&
$\displaystyle r_{\widehat\theta_{i-1},\theta}(Y_i)$
\\[1ex]

Oracle score
&
$\displaystyle
\frac{(K+1)r_{\thetatrue,\theta}(Y_i)}
{\sum_{j=0}^{K}
 r_{\thetatrue,\theta}
 (\widetilde Y_{ij}^{(\theta)})}
$
\\[2ex]

Conformal LR
&
$\displaystyle
\frac{(K+1)r_{\widehat\theta_{i-1},\theta}(Y_i)}
{\sum_{j=0}^{K}
 r_{\widehat\theta_{i-1},\theta}
 (\widetilde Y_{ij}^{(\theta)})}
$
\\[2ex]

Proposed (Conformal NRE)
&
$\displaystyle
\frac{(K+1)T_i^{(\theta)}(Y_i)}
{\sum_{j=0}^{K}
 T_i^{(\theta)}
 (\widetilde Y_{ij}^{(\theta)})}
$
\\[1ex]

\bottomrule
\end{tabular}
\end{table}

\subsection{Overall procedure}

\begin{algorithm}[t]
\caption{Likelihood-free confidence sequence with neural ratio scores}
\label{alg:likelihood-free-cs}
\begin{algorithmic}[1]
\Require Simulator for $P_\theta$; fixed finite grid
$\Theta_{\mathrm{eval}}$; pretrained network $s_\phi$;
number of reference draws $K\geq 1$; level $\alpha\in(0,1)$.

\State Initialize $M_0^{(\theta)}\gets 1$ for every
$\theta\in\Theta_{\mathrm{eval}}$, and
$\bar C_{0,\alpha}\gets\Theta_{\mathrm{eval}}$.

\For{$i=1,2,\ldots$}
    \State Using the previous martingale values, choose
    \begin{equation*}
        \widehat\theta_{i-1}
        \in \arg\min_{\vartheta\in\Theta_{\mathrm{eval}}}
        M_{i-1}^{(\vartheta)}.
    \end{equation*}

    \State Before observing $Y_i$ or drawing the current
    reference samples, fix the score functions
    \begin{equation*}
        T_i^{(\theta)}(y)
        \gets
        \exp\!\left\{
            s_\phi(\widehat\theta_{i-1},y)
            -s_\phi(\theta,y)
        \right\},
        ~~ \theta\in\Theta_{\mathrm{eval}}.
    \end{equation*}

    \State Observe $Y_i$.

    \ForAll{$\theta\in\Theta_{\mathrm{eval}}$}
        \State Draw a fresh reference batch
        $\widetilde Y_{i1}^{(\theta)},\ldots,
        \widetilde Y_{iK}^{(\theta)}
        \overset{\mathrm{i.i.d.}}{\sim}P_\theta$.

        \State Compute the one-step e-value \eqref{eq:main-evalue}

        \State Update
        $M_i^{(\theta)}
        \gets M_{i-1}^{(\theta)}E_i^{(\theta)}$.
    \EndFor

    \State Update the running intersection
    \begin{equation*}
        \bar C_{i,\alpha}
        \gets
        \bar C_{i-1,\alpha}
        \cap
        \left\{
            \theta\in\Theta_{\mathrm{eval}}:
            M_i^{(\theta)}<\alpha^{-1}
        \right\}.
    \end{equation*}
    \State Output $\bar C_{i,\alpha}$.
\EndFor
\end{algorithmic}
\end{algorithm}

Algorithm~\ref{alg:likelihood-free-cs} summarizes our
implementation. We update the martingales over the entire grid, including
previously excluded candidates, because the next alternative
is selected by minimizing over this grid.

The output satisfies
$\bar C_{n,\alpha}
=\Theta_{\mathrm{eval}}\cap\bigcap_{i=1}^n C_{i,\alpha}$.
Consequently, by \cref{prop:exact-martingale}, $(\bar C_{n,\alpha})$ satisfies \eqref{eq:cs-definition} for every $\thetatrue\in\Theta_{\mathrm{eval}}$.
Coverage beyond the evaluation grid requires additional
control between grid points; see \cref{app:continuous-parameters}.

\Cref{app:extensions} discusses extensions to non-i.i.d.\
observations, nuisance parameters, and averaging over reference batches to reduce simulation noise.
\section{EXPERIMENTS}
\label{sec:experiments}

We first use synthetic experiments to assess the efficiency costs of conformalization, 
learning likelihood ratios, and estimating the alternative parameter. 
Analytical or numerical likelihood evaluation makes these comparisons possible, 
with the true parameter providing an oracle alternative.
We then turn to applications with intractable likelihoods, moving from a semi-synthetic 
audit of model substitution in an LLM API using real model replies,
to stellar mass-function 
inference from real microlensing observations. Beyond assessing coverage where the truth is known, 
we examine how the confidence sequences narrow as observations accumulate, 
comparing neural and application-specific scores within the SLICE construction.

Throughout, we set $\alpha=0.1$ and follow \cref{alg:likelihood-free-cs},
evaluating running-intersection confidence sets on finite candidate grids and using
the minimum-evidence alternative by default. Learned likelihood-ratio estimators are fitted on independent 
data before monitoring and held fixed thereafter. Where the true parameter is known, 
we report cumulative miscoverage: the fraction of repetitions 
in which the method has excluded it by time $n$.
We assess efficiency through confidence-set widths and candidate retention. 
Reported widths and endpoints are averaged over repetitions with nonempty sets. 
Implementation details, score diagnostics and additional experiments are provided in 
\crefrange{app:implementation}{app:microlensing}.

\paragraph{Synthetic experiments.}

\begin{figure*}[t]
  \centering
  \includegraphics[width=\textwidth]{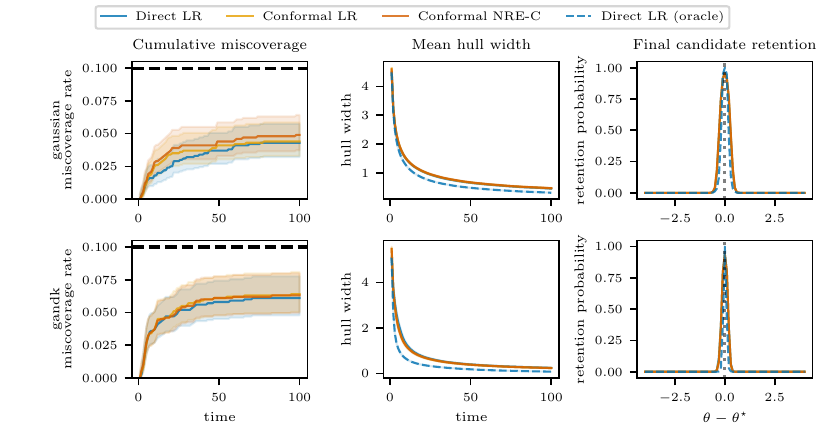}
\caption{Confidence sequences for the Gaussian (top) and
$g$-and-$k$ (bottom) models at $\thetatrue=0$.
Left: cumulative miscoverage (Wilson $95\%$ bands;
horizontal dashed line: $\alpha=0.1$).
Center: mean hull width over repetitions with nonempty sets.
Right: final-time retention probability for each candidate.
Solid curves: min-evidence alternative;
dashed curves: oracle alternative.}
  \label{fig:main-experiment}
\end{figure*}

We consider unit-variance Gaussian and likelihood-intractable
$g$-and-$k$ location models, each with a scalar parameter
$\theta\in\Theta=[-4,4]$, and SLCP, a five-parameter SBI benchmark
with a complex, multimodal posterior on $\Theta=[-3,3]^5$
\citep{papamakarios2019sequential}.
For benchmarking, likelihoods are evaluated in closed form for
Gaussian and SLCP, and by numerical inversion for $g$-and-$k$.
We use an NRE-C score and compare against direct and conformalized
likelihood-ratio confidence sequences, together with infeasible
oracle variants using $\thetatrue$ as the alternative.
All three models use $n=100$ observations over $1{,}000$ repetitions
 and $K=15$ reference draws
per candidate per step. \looseness=-1

For the two location models, we use $81$ equally spaced candidates.
\Cref{fig:main-experiment} shows results at $\thetatrue=0$;
other true parameters and ablations are reported in \cref{app:ablations}.
Empirical cumulative miscoverage remains below $\alpha$,
consistent with the theoretical guarantee.
Direct and conformalized LR confidence sequences have nearly
identical hull widths and candidate retention, indicating little
efficiency loss from conformalization.
Replacing the likelihood ratio with the NRE-C score yields
essentially the same results, while using the true parameter
as the alternative gives narrower sets.

For SLCP, we use $5{,}120$ scrambled Sobol candidates on
$\Theta=[-3,3]^5$, including $\thetatrue=(0,0,2.2,2.2,0)$
and its three sign variants, which give the same distribution because
$\theta_3$ and $\theta_4$ enter only through their squares.
Because the fixed initial alternative is far from $\thetatrue$ in this model,
we use the first two observations only to choose the alternative and start
the e-processes at $n=3$ (\cref{app:running-cs}).
\Cref{fig:slcp-main} shows cumulative miscoverage below $\alpha$
and similar rates of candidate exclusion for conformal NRE-C
and conformal LR, with direct LR shrinking faster.
\begin{figure}[h!]
  \centering
  \includegraphics[width=\columnwidth]{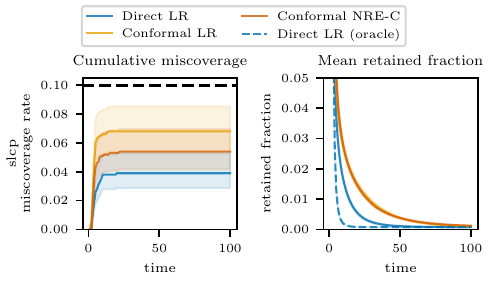}
\caption{Confidence sequences for SLCP.
Left: cumulative miscoverage (Wilson $95\%$ bands;
horizontal dashed line: $\alpha=0.1$; oracle omitted).
Right: mean fraction of candidates retained.
Solid curves: min-evidence alternative;
dashed curves: oracle alternative.}
  \label{fig:slcp-main}
\end{figure}
Here conformalization has a visible cost, the price of obtaining validity
from simulations rather than from the likelihood, whereas learning the ratio
again costs little.
The retained fractions level off as only a few candidates remain,
and the gap to the oracle narrows.
At $n=100$, conformal NRE-C contains $\thetatrue$ in $94.6\%$ of
repetitions, and keeps $\thetatrue$ while dropping at least one of its three
mirror images in only $1.2\%$, preserving the model's multiple modes
(\cref{app:slcp}).

\paragraph{Auditing model routing in an LLM API.}
We illustrate our method in a semi-synthetic audit of an API that may
substitute a cheaper language model $R$ for its advertised model $P$
(\citealp{Gao2025,Cai2025}; IRIS, \citealp{Zhang2026IRIS}). The auditor can sample from both
models but cannot evaluate either likelihood. The aim is to estimate
the substitution fraction $\theta$ while allowing the auditor to monitor
incoming replies and stop adaptively. We take $P=\texttt{qwen3:8b}$, one of IRIS's reference models,
and $R=\texttt{qwen3:1.7b}$, and use two of IRIS's prompts requesting random outputs
whose model-specific reply distributions make substitution detectable. For each
prompt, each query independently yields a reply from $R$ with probability
$\theta$ and $P$ otherwise, giving the reply pmf
$f_\theta=(1-\theta)p+\theta r$, where $p$ and $r$ are the models'
respective reply pmfs. Likelihood ratios between candidate mixtures
depend only on $\rho=r/p$, which we estimate from reply counts or a non-amortized
binary NRE classifier distinguishing $P$ from $R$, and use to construct the
scores $T_i^{(\theta)}$ (\cref{app:routing-scores}). 
We construct $200$ streams of $n=300$ observations per prompt
at each $\thetatrue\in\{0,0.1,0.3\}$ by mixing held-out replies
from $R$ and $P$ with probabilities $\thetatrue$ and $1-\thetatrue$.
We compute confidence sequences using \cref{alg:likelihood-free-cs},
with $K=10$ references per step from separate held-out reply pools.
This semi-synthetic design preserves real model outputs while
allowing coverage assessment at known substitution fractions.
We compare with our implementation of IRIS's fixed-time confidence interval (IRIS), 
which estimates
$\theta$ from the share of replies our count-ratio estimate flags as atypical
for $P$, and with an extension of IRIS's betting test to a confidence sequence
for $\theta$ (Betting IRIS), which is only approximately valid because it uses
estimated indicator rates (\cref{app:routing-baselines}).
\Cref{fig:routing-main} shows results at $\thetatrue=0.1$ for the digit prompt.
The conformal sequences and Betting IRIS keep cumulative miscoverage below $0.1$, whereas
IRIS exceeds this level under continuous
monitoring (panel~a).
\begin{figure}[h!]
  \centering
  \includegraphics[width=\columnwidth]{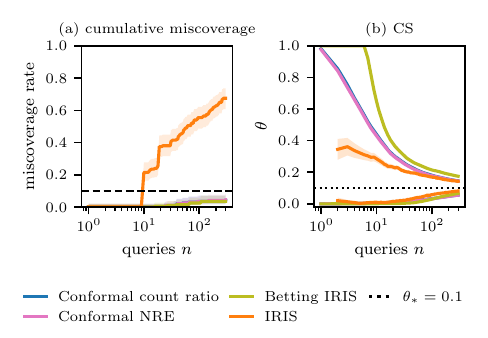}
  \caption{LLM routing audit at $\thetatrue=0.1$, 8 random digits prompt
($200$ repetitions). (a) Cumulative miscoverage (Wilson $95\%$ bands; dashed
$\alpha=0.1$). (b) Mean lower and upper endpoints of the confidence sets over
repetitions with nonempty sets, with $95\%$ intervals for the means. Both
prompts, and the IRIS conventions, are in \cref{fig:routing-theta01}.}
  \label{fig:routing-main}
\end{figure}
\begin{figure*}
  \centering
  \includegraphics[width=\textwidth]{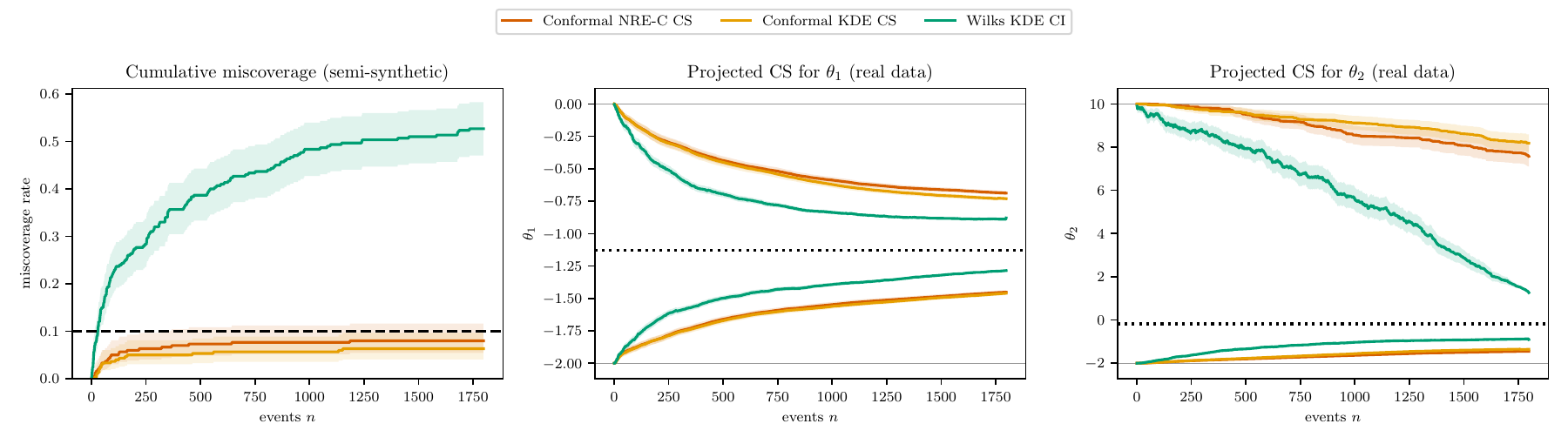}
  \caption{OGLE-IV microlensing. Left: cumulative miscoverage of the
semi-synthetic runs (Wilson $95\%$ bands; dashed $\alpha=0.1$). Center and right: projected
confidence sequences for $\theta_1$ and $\theta_2$ on the real data (mean
endpoints over random orderings, with $95\%$ intervals for the mean); grey lines mark the box and dotted lines the published
$\theta_1=-1.13$ and $\theta_2=-0.18$ \citep{Koshimoto2021}.}
  \label{fig:microlensing}
\end{figure*}
The conformal count-ratio and NRE sequences have
nearly identical CS hulls, both narrower than Betting IRIS (panel~b).
They also detect substitution sooner, as measured by the number of
queries needed to exclude $\theta=0$ (\cref{tab:routing}). The number prompt
behaves alike (\cref{fig:routing-theta01}), and results at
$\thetatrue\in\{0,0.3\}$ show consistent coverage for the conformal
sequences and Betting IRIS, with faster detection also at $0.3$, where Betting IRIS ends with
slightly narrower sets (\cref{fig:routing-theta0,fig:routing-theta03}).

\paragraph{Inferring the initial mass function from microlensing events.}
Gravitational microlensing occurs when a foreground object, the lens,
magnifies a background star, the source; the event's duration carries
information about the lens mass, even when the lens is too faint to
observe directly. We infer the slopes $\theta=(\theta_1,\theta_2)$ of
the initial mass function of low-mass stars ($\theta_1$) and brown dwarfs
($\theta_2$) from the timescales of $1{,}798$ OGLE-IV microlensing events
toward the Galactic bulge \citep{Mroz2019}, conditionally on each event's
field and measurement error (\cref{cor:conditional-validity}). 
The simulator is the Galactic model of \citet{Koshimoto2021}, as implemented in
\texttt{genulens} \citep{KoshimotoRanc2021genulens}; its likelihood integrates
over the distances, velocities, and masses of lenses and sources, and is intractable.
A held-out half of the survey seasons is used to fit the model,
which then passes a goodness-of-fit test on the analysed half (\cref{app:microlensing-fit-check}),
and simulation cost is reduced by reweighting per-field banks
(\cref{app:microlensing-reweighting}).
The reweighting is based on the simulator's exact density ratio,
and reweighted draws agree with direct simulation across the
parameter range.
We use $G=2{,}048$ candidates on the box $[-2,0]\times[-2,10]$
and $K=15$ references.
\Cref{fig:microlensing} compares the conformal NRE-C CS with two methods based
on an accurate kernel density estimate (KDE) of the timescale's likelihood,
possible because the timescale is one-dimensional: the conformal KDE CS, which
uses the KDE likelihood ratio as its score, and the fixed-sample Wilks KDE interval,
recomputed at every $n$.
In $300$ semi-synthetic repetitions, which simulate new timescales for the
real events' covariates at the true value $\thetatrue=(-1.1,1.0)$ (left), both conformal CSs keep the cumulative
miscoverage below $\alpha$, whereas the monitored Wilks interval soon exceeds
it. On the real data, we show the projections of the sets 
(\cref{app:nuisance-projections})
onto $\theta_1$ (center) and $\theta_2$ (right), averaged over
$200$ random orderings of the events. In both settings the NRE-C score is almost
as efficient as the KDE score, and the Wilks interval is narrower,
consistent with the additional cost of anytime-valid inference
(\cref{tab:microlensing-semi}). 
For $\theta_1$, both sequences shrink steadily
with $n$ and contain the published value $-1.13$ \citep{Koshimoto2021} in $98$--$99\%$ of orderings. For
$\theta_2$, they exclude low values but not high ones: a high $\theta_2$ means
few brown dwarfs, whose short events low-cadence fields rarely detect, so the
data say little against it.
\section{DISCUSSION}
\label{sec:discussion}

SLICE decouples validity from efficiency: normalizing any positive
predictable score against fresh simulations yields exact likelihood-free
confidence sequences, while the score only determines how fast they shrink.
Neural ratio scores with a minimum-evidence alternative make a solid
application-agnostic default: they came close to exact likelihood ratios on
synthetic benchmarks and to application-specific scores in two applications
with real data.
Limitations include the reliance on a grid of candidates, the randomness that reference sampling adds to the sequences, and the need for a good predictable alternative to reach the best efficiency, and, as for any model-based method, a correctly specified simulator. Extensions addressing the first two are discussed in \cref{app:extensions}, while addressing the third, for example by incorporating prior information, is a valuable direction for future work.

\subsection*{Acknowledgements}

We would like to thank David T. Frazier and Jack Jewson for the helpful
discussions. SC is supported by the EPSRC Centre for Doctoral Training in Modern
Statistics and Statistical Machine Learning (EP/S023151/1). LB is supported by
the Clarendon Funds Scholarship.

\bibliographystyle{abbrvnat}
\bibliography{neuralratioanytimevalid}

@inproceedings{Balinsky2024,
  author    = {Balinsky, Alexander A. and
               Balinsky, Alexander David},
  title     = {Enhancing Conformal Prediction Using
               {E}-Test Statistics},
  booktitle = {Proceedings of the Thirteenth Symposium on
               Conformal and Probabilistic Prediction
               with Applications},
  series    = {Proceedings of Machine Learning Research},
  volume    = {230},
  pages     = {65--72},
  publisher = {PMLR},
  year      = {2024}
}

@misc{Gauthier2025,
  author        = {Gauthier, Etienne and Bach, Francis and
                   Jordan, Michael I.},
  title         = {{E}-Values Expand the Scope of
                   Conformal Prediction},
  year          = {2025},
  eprint        = {2503.13050},
  archivePrefix = {arXiv},
  primaryClass  = {stat.ML},
  doi           = {10.48550/arXiv.2503.13050}
}

@InProceedings{Masserano2023,
  author    = {Masserano, Luca and Dorigo, Tommaso and
               Izbicki, Rafael and Kuusela, Mikael and Lee, Ann},
  title     = {Simulator-Based Inference with {WALDO}:
               Confidence Regions by Leveraging Prediction
               Algorithms and Posterior Estimators for
               Inverse Problems},
  booktitle = {Proceedings of the 26th International Conference
               on Artificial Intelligence and Statistics},
  series    = {Proceedings of Machine Learning Research},
  volume    = {206},
  pages     = {2960--2974},
  publisher = {PMLR},
  year      = {2023},
  url       = {https://proceedings.mlr.press/v206/masserano23a.html}
}

@Article{Cabezas2026,
  author  = {Cabezas, Luben M. C. and Soares, Guilherme P. and
             Ramos, Thiago R. and Stern, Rafael B. and
             Izbicki, Rafael},
  title   = {Conformal Calibration of Statistical Confidence Sets},
  journal = {Transactions on Machine Learning Research},
  year    = {2026},
  issn    = {2835-8856},
  url     = {https://openreview.net/forum?id=J4lK62PVE6}
}

@InProceedings{PodkopaevRamdas2023,
  author    = {Podkopaev, Aleksandr and Ramdas, Aaditya},
  title     = {Sequential Predictive Two-Sample and
               Independence Testing},
  booktitle = {Advances in Neural Information Processing Systems},
  volume    = {36},
  year      = {2023},
}

@Article{FischerRamdas2025,
  author  = {Fischer, Lasse and Ramdas, Aaditya},
  title   = {Sequential {Monte Carlo} Testing by Betting},
  journal = {Journal of the Royal Statistical Society
             Series B: Statistical Methodology},
  volume  = {87},
  number  = {4},
  pages   = {1200--1220},
  year    = {2025},
  doi     = {10.1093/jrsssb/qkaf014}
}

@inproceedings{SotoudehThakur2019,
  author    = {Sotoudeh, Matthew and Thakur, Aditya V.},
  title     = {Computing Linear Restrictions of Neural Networks},
  booktitle = {Advances in Neural Information Processing Systems},
  volume    = {32},
  year      = {2019},
  url       = {https://arxiv.org/abs/1908.06214}
}

@article{rayner2002numerical,
  author  = {Rayner, G. D. and MacGillivray, H. L.},
  title   = {Numerical maximum likelihood estimation for the
             {g}-and-{k} and generalized {g}-and-{h} distributions},
  journal = {Statistics and Computing},
  year    = {2002},
  volume  = {12},
  number  = {1},
  pages   = {57--75},
  doi     = {10.1023/A:1013120305780}
}

@Article{Vovk2025ConformalE,
author  = {Vovk, Vladimir},
title   = {Conformal E-Prediction},
journal = {Pattern Recognition},
year    = {2025},
volume  = {166},
pages   = {111674},
doi     = {10.1016/j.patcog.2025.111674},
}

@Article{Vovk2025ConformalTesting,
author  = {Vovk, Vladimir and Nouretdinov, Ilia and Gammerman, Alex},
title   = {Conformal E-Testing},
journal = {Pattern Recognition},
year    = {2025},
volume  = {168},
pages   = {111841},
doi     = {10.1016/j.patcog.2025.111841},
}

@Article{Hermans2022,
  author  = {J. Hermans and A. Delaunoy and F. Rozet and A. Wehenkel and V. Begy and G. Louppe},
  journal = {Transactions on Machine Learning Research},
  title   = {A Crisis In Simulation-Based Inference? Beware, Your Posterior Approximations Can Be Unfaithful},
  year    = {2022},
  issn    = {2835-8856},
  url     = {https://openreview.net/forum?id=LHAbHkt6Aq},
}

@Article{Falkiewicz2023,
  author  = {Falkiewicz, M. and Takeishi, N. and Shekhzadeh, I. and Wehenkel, A. and Delaunoy, A. and Louppe, G. and Kalousis, A.},
  journal = {Advances in Neural Information Processing Systems},
  title   = {Calibrating neural simulation-based inference with differentiable coverage probability},
  year    = {2023},
  pages   = {1082--1099},
  volume  = {36},
}

@Article{Delaunoy2022,
  author  = {Delaunoy, A. and Hermans, J. and Rozet, F. and Wehenkel, A. and Louppe, G.},
  journal = {Advances in Neural Information Processing Systems},
  title   = {Towards reliable simulation-based inference with balanced neural ratio estimation},
  year    = {2022},
  pages   = {20025--20037},
  volume  = {35},
}

@InProceedings{Patel2024,
  author       = {Patel, Y. and Mcnamara, D. and Loper, J. and Regier, J. and Tewari, A.},
  booktitle    = {International Conference on Machine Learning},
  title        = {Variational Inference with Coverage Guarantees in Simulation-Based Inference},
  year         = {2024},
  organization = {PMLR},
  pages        = {39861--39888},
}

@InProceedings{Dalmasso2020,
  author       = {Dalmasso, N. and Izbicki, R. and Lee, A.},
  booktitle    = {International Conference on Machine Learning},
  title        = {Confidence sets and hypothesis testing in a likelihood-free inference setting},
  year         = {2020},
  organization = {PMLR},
  pages        = {2323--2334},
}

@Article{Miller2022,
  author  = {Miller, B. K. and Weniger, C. and Forr{\'e}, P.},
  journal = {Advances in Neural Information Processing Systems},
  title   = {Contrastive neural ratio estimation},
  year    = {2022},
  pages   = {3262--3278},
  volume  = {35},
}

@InProceedings{Hermans2020,
  author       = {Hermans, J. and Begy, V. and Louppe, G.},
  booktitle    = {International conference on machine learning},
  title        = {Likelihood-free mcmc with amortized approximate ratio estimators},
  year         = {2020},
  organization = {PMLR},
  pages        = {4239--4248},
}

@Article{Thomas2022,
  author    = {Thomas, O. and Dutta, R. and Corander, J. and Kaski, S. and Gutmann, M. U.},
  journal   = {Bayesian Analysis},
  title     = {Likelihood-free inference by ratio estimation},
  year      = {2022},
  number    = {1},
  pages     = {1--31},
  volume    = {17},
  publisher = {International Society for Bayesian Analysis},
}

@Article{Gutmann2012,
  author    = {Gutmann, M. U. and Hyv{\"a}rinen, A.},
  journal   = {The journal of machine learning research},
  title     = {Noise-contrastive estimation of unnormalized statistical models, with applications to natural image statistics},
  year      = {2012},
  number    = {1},
  pages     = {307--361},
  volume    = {13},
  publisher = {JMLR. org},
}

@InProceedings{Poole2019,
  author       = {Poole, Ben and Ozair, Sherjil and van den Oord, A{\"a}ron and Alemi, Alexander A. and Tucker, George},
  booktitle    = {Proceedings of the 36th International Conference on Machine Learning},
  title        = {On Variational Bounds of Mutual Information},
  year         = {2019},
  organization = {PMLR},
  pages        = {5171--5180},
}

@InProceedings{Gammerman1998,
  author    = {Gammerman, Alex and Vovk, Vladimir and Vapnik, Vladimir},
  booktitle = {Proceedings of the Fourteenth Conference on Uncertainty in Artificial Intelligence},
  title     = {Learning by Transduction},
  year      = {1998},
  pages     = {148--155},
  publisher = {Morgan Kaufmann},
}

@book{Ramdas2025,
  title = {Hypothesis {{Testing}} with {{E-values}}},
  author = {Ramdas, Aaditya and Wang, Ruodu},
  year = 2025,
  publisher = {{Foundations and Trends in Statistics, Vol. 1: No. 1-2, pp 1-390}},
  langid = {english}
}

@Article{Dalmasso2024,
  author    = {Dalmasso, N. and Masserano, L. and Zhao, D. and Izbicki, R. and Lee, A. B.},
  journal   = {Electronic Journal of Statistics},
  title     = {Likelihood-free frequentist inference: Bridging classical statistics and machine learning for reliable simulator-based inference},
  year      = {2024},
  number    = {2},
  pages     = {5045--5090},
  volume    = {18},
  publisher = {The Institute of Mathematical Statistics and the Bernoulli Society},
}

@Article{Oord2018,
  author  = {van den Oord, A. and Li, Y. and Vinyals, O.},
  journal = {arXiv preprint arXiv:1807.03748},
  title   = {Representation Learning with Contrastive Predictive Coding},
  year    = {2018},
}

@Article{Wang2026,
  author  = {Wang, H. and Ramdas, A.},
  journal = {arXiv preprint arXiv:2605.30485},
  title   = {M-estimation with e-statistics},
  year    = {2026},
}

@Article{Gruenwald2023,
  author    = {Gr{\"u}nwald, P. D.},
  journal   = {Philosophical Transactions of the Royal Society A},
  title     = {The e-posterior},
  year      = {2023},
  number    = {2247},
  pages     = {20220146},
  volume    = {381},
  publisher = {The Royal Society},
}

@InProceedings{Gutmann2010,
  author    = {Gutmann, M. and Hyv{\"a}rinen, A.},
  booktitle = {Proceedings of the Thirteenth International Conference on Artificial Intelligence and Statistics},
  title     = {Noise-Contrastive Estimation: A New Estimation Principle for Unnormalized Statistical Models},
  year      = {2010},
  editor    = {Teh, Yee Whye and Titterington, Mike},
  pages     = {297--304},
  publisher = {PMLR},
  series    = {Proceedings of Machine Learning Research},
  volume    = {9},
}

@Article{Cranmer2015,
  author        = {Cranmer, K. and Pavez, J. and Louppe, G.},
  journal       = {arXiv preprint arXiv:1506.02169},
  title         = {Approximating Likelihood Ratios with Calibrated Discriminative Classifiers},
  year          = {2015},
  archiveprefix = {arXiv},
  eprint        = {1506.02169},
}

@Article{Howard2021,
  author  = {Howard, S. R. and Ramdas, A. and McAuliffe, J. and Sekhon, J.},
  journal = {The Annals of Statistics},
  title   = {Time-uniform, nonparametric, nonasymptotic confidence sequences},
  year    = {2021},
  number  = {2},
  volume  = {49},
}

@Book{Ville1939,
  author    = {Ville, J.},
  publisher = {Gauthier-Villars Paris},
  title     = {Etude critique de la notion de collectif},
  year      = {1939},
}

@Article{Lai1976,
  author    = {Lai, Tze Leung},
  journal   = {The Annals of Statistics},
  title     = {On {{Confidence Sequences}}},
  year      = {1976},
  issn      = {0090-5364, 2168-8966},
  number    = {2},
  pages     = {265--280},
  volume    = {4},
  date      = {1976},
  doi       = {10.1214/aos/1176343406},
  publisher = {Institute of Mathematical Statistics},
  urldate   = {2025-03-26},
}

@article{Shafer2011,
  author = {Glenn Shafer and Alexander Shen and Nikolai Vereshchagin and Vladimir Vovk},
  title = {{Test Martingales, Bayes Factors and p-Values}},
  volume = {26},
  journal = {Statistical Science},
  number = {1},
  publisher = {Institute of Mathematical Statistics},
  pages = {84 -- 101},
  year = {2011},
  doi = {10.1214/10-STS347},
  URL = {https://doi.org/10.1214/10-STS347}
}

@Article{Koning2026,
  author    = {Koning, Nick},
  journal   = {Annals of Statistics},
  title     = {Measuring evidence against exchangeability and group invariance with e-values},
  year      = {2026},
  publisher = {Institute of Mathematical Statistics},
}

@Article{Beaumont2002,
author = {Beaumont, Mark A. and Zhang, Wenyang and Balding, David J.},
title = {Approximate Bayesian Computation in Population Genetics},
journal = {Genetics},
year = {2002},
volume = {162},
number = {4},
pages = {2025--2035},
doi = {10.1093/genetics/162.4.2025},
}

@Article{Wood2010,
author = {Wood, Simon N.},
title = {Statistical Inference for Noisy Nonlinear Ecological Dynamic Systems},
journal = {Nature},
year = {2010},
volume = {466},
number = {7310},
pages = {1102--1104},
doi = {10.1038/nature09319},
}

@Article{Cranmer2020,
author = {Cranmer, Kyle and Brehmer, Johann and Louppe, Gilles},
title = {The Frontier of Simulation-Based Inference},
journal = {Proceedings of the National Academy of Sciences},
year = {2020},
volume = {117},
number = {48},
pages = {30055--30062},
doi = {10.1073/pnas.1912789117},
}

@InProceedings{Durkan2020,
author = {Durkan, Conor and Murray, Iain and Papamakarios, George},
title = {On Contrastive Learning for Likelihood-Free Inference},
booktitle = {Proceedings of the 37th International Conference on Machine Learning},
year = {2020},
volume = {119},
series = {Proceedings of Machine Learning Research},
pages = {2771--2781},
publisher = {PMLR},
}

@Book{Sugiyama2012,
author = {Sugiyama, Masashi and Suzuki, Taiji and Kanamori, Takafumi},
title = {Density Ratio Estimation in Machine Learning},
publisher = {Cambridge University Press},
year = {2012},
doi = {10.1017/CBO9781139035613},
}

@Article{VovkWang2021,
author = {Vovk, Vladimir and Wang, Ruodu},
title = {{E}-Values: Calibration, Combination, and Applications},
journal = {The Annals of Statistics},
year = {2021},
volume = {49},
number = {3},
pages = {1736--1754},
doi = {10.1214/20-AOS2020},
}

@Article{Ramdas2023,
  author  = {Ramdas, Aaditya and Gr{\"u}nwald, Peter and Vovk, Vladimir and Shafer, Glenn},
  journal = {Statistical Science},
  title   = {Game-Theoretic Statistics and Safe Anytime-Valid Inference},
  year    = {2023},
  number  = {4},
  pages   = {576--601},
  volume  = {38},
  doi     = {10.1214/23-STS894},
}

@Article{Wasserman2020,
author = {Wasserman, Larry and Ramdas, Aaditya and Balakrishnan, Sivaraman},
title = {Universal Inference},
journal = {Proceedings of the National Academy of Sciences},
year = {2020},
volume = {117},
number = {29},
pages = {16880--16890},
doi = {10.1073/pnas.1922664117},
}

@article{dombowsky2026besag,
  title={Besag-{C}lifford e-values for unnormalized testing},
  author={Dombowsky, Alexander and Engelhardt, Barbara E and Ramdas, Aaditya},
  journal={arXiv preprint arXiv:2603.15845},
  year={2026}
}

@Article{Paczynski1986,
  author  = {Paczy{\'n}ski, B.},
  journal = {The Astrophysical Journal},
  title   = {Gravitational microlensing by the galactic halo},
  year    = {1986},
  pages   = {1--5},
  volume  = {304},
}

@Article{Mroz2019,
  author  = {Mr{\'o}z, P. and Udalski, A. and Skowron, J. and Szyma{\'n}ski, M. K. and Soszy{\'n}ski, I. and Wyrzykowski, {\L}. and Pietrukowicz, P. and Koz{\l}owski, S. and Poleski, R. and Ulaczyk, K. and Rybicki, K. and Iwanek, P.},
  journal = {The Astrophysical Journal Supplement Series},
  title   = {Microlensing optical depth and event rate toward the {G}alactic bulge from 8 yr of {OGLE-IV} observations},
  year    = {2019},
  number  = {2},
  pages   = {29},
  volume  = {244},
}

@Article{Koshimoto2021,
  author  = {Koshimoto, N. and Baba, J. and Bennett, D. P.},
  journal = {The Astrophysical Journal},
  title   = {A parametric {G}alactic model toward the {G}alactic bulge based on {G}aia and microlensing data},
  year    = {2021},
  number  = {2},
  pages   = {78},
  volume  = {917},
}

@Article{Zhang2026IRIS,
  author  = {Zhang, Yuewei and Zhang, Zhi-Hai and Qin, Hanzhang},
  journal = {arXiv preprint arXiv:2607.20860},
  title   = {Which Model Is Actually Serving You? {IRIS}: Budgeted Black-Box Auditing of Model Substitution and Routing Dilution in {LLM} Gateways},
  year    = {2026},
}

@InProceedings{Gao2025,
  author    = {Gao, Irena and Liang, Percy and Guestrin, Carlos},
  booktitle = {International Conference on Learning Representations},
  title     = {Model Equality Testing: Which Model Is This {API} Serving?},
  year      = {2025},
}

@Article{Cai2025,
  author  = {Cai, Will and Shi, Tianneng and Zhao, Xuandong and Song, Dawn},
  journal = {arXiv preprint arXiv:2504.04715},
  title   = {Are You Getting What You Pay For? {A}uditing Model Substitution in {LLM} {API}s},
  year    = {2025},
}

@Article{WaudbySmith2024,
  author  = {Waudby-Smith, Ian and Ramdas, Aaditya},
  journal = {Journal of the Royal Statistical Society Series B: Statistical Methodology},
  title   = {Estimating Means of Bounded Random Variables by Betting},
  year    = {2024},
  number  = {1},
  pages   = {1--27},
  volume  = {86},
}

@Article{Yang2025Qwen3,
  author  = {Yang, An and others},
  journal = {arXiv preprint arXiv:2505.09388},
  title   = {{Qwen3} Technical Report},
  year    = {2025},
}

@Misc{KoshimotoRanc2021genulens,
  author       = {Koshimoto, Naoki and Ranc, Cl{\'e}ment},
  title        = {genulens: A Tool for Gravitational Microlensing Events Simulation},
  year         = {2021},
  howpublished = {Computer software, Zenodo},
  doi          = {10.5281/zenodo.4784948},
  url          = {https://doi.org/10.5281/zenodo.4784948},
}

@inproceedings{papamakarios2019sequential,
  author    = {Papamakarios, George and Sterratt, David and Murray, Iain},
  title     = {Sequential Neural Likelihood: Fast Likelihood-free Inference
               with Autoregressive Flows},
  booktitle = {Proceedings of the Twenty-Second International Conference on
               Artificial Intelligence and Statistics},
  series    = {Proceedings of Machine Learning Research},
  volume    = {89},
  pages     = {837--848},
  year      = {2019},
}

@inproceedings{lueckmann2021benchmarking,
  author    = {Lueckmann, Jan-Matthis and Boelts, Jan and Greenberg, David and
               Goncalves, Pedro and Macke, Jakob},
  title     = {Benchmarking Simulation-Based Inference},
  booktitle = {Proceedings of The 24th International Conference on Artificial
               Intelligence and Statistics},
  series    = {Proceedings of Machine Learning Research},
  volume    = {130},
  pages     = {343--351},
  year      = {2021},
}

@article{sobol1967distribution,
  author  = {Sobol', I. M.},
  title   = {On the distribution of points in a cube and the approximate
             evaluation of integrals},
  journal = {USSR Computational Mathematics and Mathematical Physics},
  volume  = {7},
  number  = {4},
  pages   = {86--112},
  year    = {1967},
  doi     = {10.1016/0041-5553(67)90144-9},
}

@Article{Forre2026,
  author  = {P. Forr\'e and L. Brenner},
  journal = {arXiv:2610.08210},
  title   = {Anytime-valid simulation-based hypothesis testing},
  year    = {2026},
  month   = oct,
}

\clearpage
\appendix
\thispagestyle{empty}

\onecolumn
\aistatstitle{Sequential Likelihood-Free Inference via Conformal E-Values:\\Supplementary Materials}

\section{ADDITIONAL RELATED WORK}
\label{app:related}

\paragraph{Simulation-based inference and posterior calibration.}
Classical likelihood-free methods include approximate Bayesian
computation, which compares observed and simulated summaries,
and synthetic likelihood, which models the distribution of
summaries using repeated simulation
\citep{Beaumont2002,Wood2010}.
Neural simulation-based inference learns posterior,
likelihood, or density-ratio surrogates, often amortizing
inference across observations or parameter values
\citep{Cranmer2020}.
Finite-budget posterior approximations can be overconfident
or otherwise poorly calibrated \citep{Hermans2022}.
Proposed remedies include balanced neural ratio estimation
\citep{Delaunoy2022}, differentiable coverage calibration
\citep{Falkiewicz2023}, and conformalized posterior regions
\citep{Patel2024}.
In particular, CANVI provides marginal coverage under the
simulator-induced joint distribution of parameters and
observations. This differs from coverage at each fixed
data-generating parameter, and from coverage that holds
simultaneously over observation times.

\paragraph{Density-ratio and contrastive estimation.}
Direct density-ratio estimation avoids separately estimating
two densities \citep{Sugiyama2012}.
Noise-contrastive estimation learns unnormalized models by
discriminating data from noise
\citep{Gutmann2010,Gutmann2012}; related classifier
constructions estimate likelihood or posterior ratios from
simulator output \citep{Cranmer2015,Thomas2022,Hermans2020}.
Contrastive formulations connect neural posterior and ratio
estimation and provide multi-sample training objectives
\citep{Durkan2020,Miller2022}.
Our classification interpretation of expected log growth is
also connected to InfoNCE
\citep{Oord2018,Poole2019}.
These connections motivate efficient scores; validity follows
from the subsequent normalization, independently of whether
training recovers the population-optimal ratio.

\paragraph{Calibration of frequentist confidence regions.}
ACORE estimates odds ratios with probabilistic classifiers
and calibrates tests through simulated critical values
\citep{Dalmasso2020}.
LF2I provides a modular framework for confidence-set
construction and coverage diagnostics
\citep{Dalmasso2024}.
WALDO uses prediction algorithms or posterior estimators to
construct a Wald-type statistic, followed by simulated
quantile calibration \citep{Masserano2023}.
For methods using learned critical values, coverage accuracy
depends on the calibration step as well as the statistic.
The TRUST framework distinguishes finite-sample local
coverage from parameter-conditional coverage obtained
asymptotically as the simulation budget increases, under
appropriate assumptions \citep{Cabezas2026}.
These guarantees concern a fixed observed sample size.
SLICE instead constructs conditional e-values at each
observation and combines them into test martingales.

\paragraph{E-values and confidence sequences.}
Test martingales and Ville's inequality provide classical
foundations for sequential inference \citep{Ville1939}.
Modern treatments develop confidence sequences, e-values,
and guarantees under continuous monitoring and optional
stopping \citep{Howard2021,VovkWang2021,Ramdas2023}.
Universal inference constructs finite-sample-valid
likelihood-ratio-type tests and confidence sets through
sample splitting, with extensions to sequential inference
\citep{Wasserman2020}.
SLICE shares the use of nonnegative evidence processes and
test inversion, obtaining the required conditional
expectation property through simulator draws and
exchangeability.

\paragraph{Exchangeability and neighboring sequential methods.}
Normalized scores have precedents in early transductive
inference \citep{Gammerman1998}, conformal e-prediction and
e-testing
\citep{Vovk2025ConformalE,Vovk2025ConformalTesting}, and
testing under group invariance \citep{Koning2026}.
Sequential predictive two-sample tests also use learned
classifiers to construct anytime-valid evidence
\citep{PodkopaevRamdas2023}.
SLICE applies sequential testing separately to candidate
simulator distributions and inverts the tests for parameter
inference.
A different sequential setting is studied by
\citet{FischerRamdas2025}, who develop anytime-valid Monte
Carlo tests as additional resamples are generated for a fixed
observed dataset. In SLICE, the observed data themselves
arrive sequentially, with fresh reference simulations at each
round, as in the simulation-based sequential construction of
\citet{dombowsky2026besag}.

\section{BACKGROUND}

\subsection{Coverage of the confidence sequence}
\label{app:coverage}

The argument is standard \citep{Howard2021,Ramdas2023}; we include it for
completeness. Fix $\thetatrue\in\Theta$. By definition of $C_{n,\alpha}$,
\[
\{\exists n\geq 1:\thetatrue\notin C_{n,\alpha}\}
=\{\exists n\geq 1:M_n^{(\thetatrue)}\geq\alpha^{-1}\}
\subseteq\{\textstyle\sup_{n\geq 0}M_n^{(\thetatrue)}\geq\alpha^{-1}\}.
\]
Since $(M_n^{(\thetatrue)})_{n\geq 0}$ is a test supermartingale under
$P_{\thetatrue}$, Ville's inequality \citep{Ville1939} gives
\[
\PP_{\thetatrue}\!\left(
    \exists n\geq 1:\thetatrue\notin C_{n,\alpha}
\right)
\leq
\PP_{\thetatrue}\!\left(
    \sup_{n\geq 0}M_n^{(\thetatrue)}\geq\alpha^{-1}
\right)
\leq\alpha.
\]
As $\thetatrue$ was arbitrary,
$\PP_{\theta}(\forall n\geq 1:\theta\in C_{n,\alpha})\geq 1-\alpha$ for
every $\theta\in\Theta$, which is the claimed coverage.

\subsection{Neural ratio estimation losses}
\label{app:nre-losses}

We use $q$, $m$, and $s_\phi$ as defined in
\cref{sec:backgroundNRE}. Probabilities and expectations
below refer to the specified auxiliary classification
experiments. Population minimizers are characterized over
unrestricted measurable scores; in practice, the score is
restricted to the network family $s_\phi$.
For these characterizations, assume that
$f_\theta(y)/m(y)$ is positive and finite for
$q(\theta)m(y)\,\dd\theta\,\mu(\dd y)$-almost every
$(\theta,y)$. Identities between population targets are
understood almost everywhere under this measure.

\paragraph{Binary logistic loss (NRE-A).}
Introduce a label $C\in\{0,1\}$ with
$\PP(C=1)=\tfrac12$, and let $(\theta,Y)$ have
conditional density $q(\theta)f_\theta(y)$ given $C=1$
and $q(\theta)m(y)$ given $C=0$
\citep{Hermans2020}.
A classifier with logit $s(\theta,Y)$ models the
conditional probability of class $1$ by
$\sigma(s(\theta,Y))$, where
$\sigma(u)=(1+e^{-u})^{-1}$.
Its population cross-entropy, dropping the overall
factor $\tfrac12$, is
\[
\begin{aligned}
\mathcal L_{\mathrm{bin}}(s)
={}&-\EE_{\mathrm{joint}}
    [\log \sigma(s(\theta,Y))]\\
   &-\EE_{\mathrm{marg}}
    [\log\{1-\sigma(s(\theta,Y))\}],
\end{aligned}
\]
where the expectations correspond to the joint density
$q(\theta)f_\theta(y)$ and the product of marginals
$q(\theta)m(y)$, respectively.
Pointwise in $(\theta,y)$ on the common support,
the integrand is minimized at
\[
\sigma(s^\star(\theta,y))
=\frac{f_\theta(y)}{f_\theta(y)+m(y)},
\quad\text{i.e.}\quad
s^\star(\theta,y)=\log\frac{f_\theta(y)}{m(y)},
\]
which agrees with \eqref{eq:lte-target}.
Under unequal sampling with $\PP(C=1)=\pi\in(0,1)$,
the same computation shows that the minimizer of the
correspondingly weighted cross-entropy picks up the
additive constant $\log\{\pi/(1-\pi)\}$.

\paragraph{InfoNCE/NRE-B loss.}
An InfoNCE-type loss arises from a $(K+1)$-way
classification problem with one positive and $K$
negatives \citep{Oord2018,Poole2019}.
Fix $\theta_0,\theta_1$ such that $P_{\theta_0}$ and
$P_{\theta_1}$ are mutually absolutely continuous.
Draw $J\sim\operatorname{Unif}\{0,\ldots,K\}$;
conditional on $J$, draw $X_J\sim P_{\theta_1}$ and
$X_\ell\sim P_{\theta_0}$ independently for $\ell\ne J$.
Bayes' rule gives
\[
\PP(J=j\mid X_{0:K})
=\frac{r_{\theta_1,\theta_0}(X_j)}
       {\sum_{\ell=0}^K r_{\theta_1,\theta_0}(X_\ell)}.
\]
Consider a conditional class distribution parametrized
by $h$,
\[
q_h(j\mid X_{0:K})
=\frac{e^{h(X_j)}}{\sum_{\ell=0}^K e^{h(X_\ell)}}.
\]
The associated population cross-entropy is
\[
\mathcal L_{\mathrm{NCE},K}(h)
=-\EE\!\left[\log q_h(J\mid X_{0:K})\right],
\]
which is minimized by $q_h$ equal to the true conditional
class distribution. Matching the two softmax forms yields
$h^\star(y)=\log r_{\theta_1,\theta_0}(y)+a$
for some constant $a$, almost everywhere under
$P_{\theta_0}$.
This is the classification problem behind the growth
criterion \eqref{eq:finite-k-growth}, with observations
as the contrasted items.

NRE-B \citep{Durkan2020} instead contrasts parameters
at a fixed observation.
Writing $K_{\mathrm{tr}}\geq2$ for the number of
candidate parameters, draw
$\vartheta_1,\ldots,\vartheta_{K_{\mathrm{tr}}}
\stackrel{\mathrm{i.i.d.}}{\sim}q$
and, independently,
$J\sim\operatorname{Unif}\{1,\ldots,K_{\mathrm{tr}}\}$.
Conditional on $J$ and
$\vartheta_{1:K_{\mathrm{tr}}}$, draw
$Y\sim P_{\vartheta_J}$.
Bayes' rule now gives
\[
\PP(J=j\mid\vartheta_{1:K_{\mathrm{tr}}},Y)
=\frac{f_{\vartheta_j}(Y)/m(Y)}
       {\sum_{\ell=1}^{K_{\mathrm{tr}}}
        f_{\vartheta_\ell}(Y)/m(Y)}.
\]
The corresponding model is
\[
q_s(j\mid\vartheta_{1:K_{\mathrm{tr}}},Y)
=\frac{e^{s(\vartheta_j,Y)}}
       {\sum_{\ell=1}^{K_{\mathrm{tr}}}
        e^{s(\vartheta_\ell,Y)}},
\]
with population loss
\[
\mathcal L_{\mathrm B}(s)
=-\EE\!\left[
    \log q_s(J\mid\vartheta_{1:K_{\mathrm{tr}}},Y)
\right].
\]
The same argument shows that its minimizers are
\[
s^\star(\theta,y)
=\log\frac{f_\theta(y)}{m(y)}+a(y),
\]
where $a$ is an arbitrary measurable real-valued
function: matching the two softmax forms at each fixed
$Y$ determines $s$ only up to a per-observation shift,
which enters numerator and denominator as a common
factor.
Score differences
$s^\star(\theta_1,y)-s^\star(\theta_0,y)$
nonetheless recover $\log r_{\theta_1,\theta_0}(y)$,
as the same function of $y$ cancels in the difference.

\paragraph{Multiclass NRE-C loss.}
NRE-C \citep{Miller2022} removes the indeterminacy in
$y$ by adding a class in which the observation is
independent of all candidate parameters.
Fix $K_{\mathrm{tr}}\geq1$ and $\gamma>0$.
Draw $\vartheta_1,\ldots,\vartheta_{K_{\mathrm{tr}}}$
independently from $q$ and, independently, draw a label
$C\in\{0,\ldots,K_{\mathrm{tr}}\}$ with
\[
\PP(C=0)=\frac{1}{1+\gamma},
\qquad
\PP(C=j)=\frac{\gamma}
{K_{\mathrm{tr}}(1+\gamma)},
\quad j=1,\ldots,K_{\mathrm{tr}}.
\]
Conditional on the candidates and $C=j\geq1$, draw
$Y\sim P_{\vartheta_j}$.
For $C=0$, draw $Y$ independently of the candidates
with density $m$.
Thus $\gamma$ is the prior odds of a nonzero class
versus class zero.

The model class probabilities are
\[
\begin{aligned}
\pi_{s,\gamma}(0\mid\vartheta_{1:K_{\mathrm{tr}}},y)
&=
\frac{K_{\mathrm{tr}}}
{K_{\mathrm{tr}}+
 \gamma\sum_{\ell=1}^{K_{\mathrm{tr}}}
 e^{s(\vartheta_\ell,y)}},
\\
\pi_{s,\gamma}(j\mid\vartheta_{1:K_{\mathrm{tr}}},y)
&=
\frac{\gamma e^{s(\vartheta_j,y)}}
{K_{\mathrm{tr}}+
 \gamma\sum_{\ell=1}^{K_{\mathrm{tr}}}
 e^{s(\vartheta_\ell,y)}},
\qquad j=1,\ldots,K_{\mathrm{tr}},
\end{aligned}
\]
and the population loss is
\[
\mathcal L_{\mathrm C}(s)
=-\EE\!\left[
    \log\pi_{s,\gamma}
    (C\mid\vartheta_{1:K_{\mathrm{tr}}},Y)
\right].
\]
Bayes' rule gives the true conditional class
probabilities by replacing
$e^{s(\vartheta_j,y)}$ with $f_{\vartheta_j}(y)/m(y)$.
Matching the odds of each nonzero class against class
zero therefore identifies
$s^\star(\theta,y)=\log\{f_\theta(y)/m(y)\}$
uniquely, as in \eqref{eq:lte-target}, without an
additive ambiguity.
NRE-A is recovered at $K_{\mathrm{tr}}=1$ and
$\gamma=1$.
As $\gamma\to\infty$, class zero disappears and the
remaining class probabilities reduce to those of
NRE-B.

\section{PRACTICAL VARIANTS AND EXTENSIONS}
\label{app:extensions}

\subsection{Variants used in our experiments}

\subsubsection{Beyond i.i.d. observations}

The presentation assumes independent observations. Extensions are possible whenever, under $H_{\theta}$, the observed next datum and the pseudo-observations are conditionally exchangeable given the past. For Markov or state-space simulators, this requires simulating pseudo-observations conditionally on the same observed history or latent-filtering information.

The simplest such case, used in \cref{app:microlensing}, is inference conditional on observed covariates. Suppose the pairs $(X_i,Y_i)$ are i.i.d., and that under $H_\theta$ the conditional law of $Y_i$ given $X_i=x$ is $P_\theta(\cdot\mid x)$, while the law of $X_i$ is left unspecified. Let $\mathcal G_{i-1}=\sigma(\mathcal F_{i-1},X_i)$ add the current covariate to the past.

\begin{corollary}[Conditional validity]
\label{cor:conditional-validity}
Fix $\theta$. At each $i$, let the score $T_i^{(\theta)}$ be positive and $\mathcal G_{i-1}$-measurable, and draw $\widetilde Y_{i1}^{(\theta)},\ldots,\widetilde Y_{iK}^{(\theta)}$ i.i.d.\ from $P_\theta(\cdot\mid X_i)$, independently of $Y_i$ given $\mathcal G_{i-1}$. Then $(M_n^{(\theta)})_{n\ge0}$ with $E_i^{(\theta)}$ as in \eqref{eq:main-evalue} is a test martingale with respect to $(\mathcal G_n)$ under $H_\theta$, whatever the distribution of the covariates. Consequently, \eqref{eq:likelihood-free-cs} is a $(1-\alpha)$ confidence sequence.
\end{corollary}

\begin{proof}
Under $H_\theta$, conditional on $\mathcal G_{i-1}$, the variables $Y_i,\widetilde Y_{i1}^{(\theta)},\ldots,\widetilde Y_{iK}^{(\theta)}$ are i.i.d.\ from $P_\theta(\cdot\mid X_i)$ and $T_i^{(\theta)}$ is fixed. The proof of \cref{prop:exact-martingale} then applies verbatim with $\mathcal G_{i-1}$ in place of $\mathcal F_{i-1}$, giving $\EE_\theta[E_i^{(\theta)}\mid\mathcal G_{i-1}]=1$; Ville's inequality gives the confidence sequence.
\end{proof}

Conditioning discards only the information about $\theta$ carried by the covariates themselves, and removes the need to model their distribution.

\subsubsection{Nuisance parameters and projections}
\label{app:nuisance-projections}

If the target is a lower-dimensional parameter $\psi=\psi(\theta)$, a valid projected confidence sequence is
\[
    C_{n,\alpha}^\psi=\{\psi(\theta):\theta\in C_{n,\alpha}\}.
\]
This projection can be conservative but preserves coverage. For testing a composite null $H_0:\theta\in\Theta_0$, the simple-null e-processes do not automatically combine by averaging, since $M_n^{(\eta)}$ is generally not an e-process under $P_\theta$ when $\eta\ne\theta$. A generic valid but often conservative is to consider the process
\[
    M_n^{\mathrm{comp}}(\Theta_0)=\inf_{\eta\in\Theta_0}M_n^{(\eta)}. 
\]
This process need not be a test martingale. Nevertheless, for every $\theta\in\Theta_0$,
$M_n^{\mathrm{comp}}(\Theta_0)\leq M_n^{(\theta)}$,
so Ville's inequality gives
\[
\begin{aligned}
    &\PP_\theta\!\left(
        \exists n\geq1:
        M_n^{\mathrm{comp}}(\Theta_0)\geq\alpha^{-1}
    \right)\\
    &\qquad\leq
    \PP_\theta\!\left(
        \sup_{n\geq0}M_n^{(\theta)}\geq\alpha^{-1}
    \right)
    \leq\alpha.
\end{aligned}
\]
Thus rejecting when $M_n^{\mathrm{comp}}(\Theta_0)$
crosses $\alpha^{-1}$ gives an anytime-valid
level-$\alpha$ test of the composite null.

\subsection{Further extensions}

\subsubsection{Averaging over independent reference batches}

To reduce variability due to the reference simulations, one may average conditional e-values obtained with independent reference batches. At each time $i$, generate
$B\geq 1$ mutually independent reference batches, each containing
$K$ independent draws from $P_\theta$, independently of
$(\calF_{i-1},Y_i)$.
Let $E_{i,b}^{(\theta)}$ denote the factor obtained from
batch $b$ using \eqref{eq:main-evalue}, with the same
observation $Y_i$ and predictable score $T_i^{(\theta)}$
for every batch. Define
\[
    \bar E_i^{(\theta)}
    =\frac1B\sum_{b=1}^B E_{i,b}^{(\theta)},
    \qquad
    \bar M_n^{(\theta)}
    =\prod_{i=1}^n\bar E_i^{(\theta)},
    \qquad
    \bar M_0^{(\theta)}=1.
\]
By \cref{prop:exact-martingale} and linearity of
conditional expectation,
\[
    \EE_\theta\!\left[
        \bar E_i^{(\theta)}\mid\calF_{i-1}
    \right]=1.
\]
Hence $(\bar M_n^{(\theta)})_{n\geq0}$ is a test
martingale under $\PP_\theta$ and can be inverted
to obtain a confidence sequence.
The factors share $Y_i$ but are independent conditional
on $(\calF_{i-1},Y_i)$; averaging therefore divides their
conditional variance by $B$, at the cost of $BK$
reference draws per observation and candidate parameter.

\subsubsection{Grid evaluation and continuous parameter spaces}
\label{app:continuous-parameters}

The theoretical confidence sequence in
\eqref{eq:likelihood-free-cs} is defined over all
$\theta\in\Theta$. Evaluating it on a finite grid gives
the stated coverage guarantee only for true parameters
on that grid. Continuous-parameter coverage can instead
be obtained by computing a \emph{certified outer
approximation}: an interval is excluded only when every
parameter it contains is certified to have crossed the
rejection threshold. We describe this extension for a
compact interval $\Theta\subset\mathbb R$.

\paragraph{A common random field over parameters.}
Suppose the simulator admits a jointly measurable
representation $S(\theta,U)$, where the distribution of
$U$ does not depend on $\theta$ and
$S(\theta,U)\sim P_\theta$. At each round $i$, after
fixing the positive predictable score family
$T_i^{(\theta)}$, generate independent base draws
$U_{i1},\ldots,U_{iK}$, independently of
$(\mathcal F_{i-1},Y_i)$, and set
\[
    \widetilde Y_{ij}^{(\theta)}
    =S(\theta,U_{ij}).
\]
Assume that the score family is jointly measurable in
$(\theta,y)$. These common random numbers define
$M_n^{(\theta)}$ simultaneously over $\Theta$, whether
or not a particular parameter has been evaluated
numerically. For each fixed $\theta$, the conditional
exchangeability argument of
\cref{prop:exact-martingale} remains unchanged.

Write
\[
    \ell_m(\theta)=\log M_m^{(\theta)},
    \qquad
    h_\alpha=\log(1/\alpha).
\]
The ideal continuous running intersection is
\[
    \mathcal C_{n,\alpha}
    =
    \bigcap_{m=1}^n C_{m,\alpha}
    =
    \left\{
        \theta\in\Theta:
        \max_{1\le m\le n}\ell_m(\theta)<h_\alpha
    \right\}.
\]

\paragraph{Certified outer approximation.}
At round $n$, let $\mathcal R_n$ be a finite,
measurably selected collection of intervals
$I\subseteq\Theta$. An interval may be included in
$\mathcal R_n$ only if, for some $m\le n$, a certified
lower bound satisfies
\begin{equation}
    h_\alpha
    \le \underline\ell_m(I)
    \le \inf_{\theta\in I}\ell_m(\theta).
    \label{eq:continuous-interval-certificate}
\end{equation}
Here, certification means that the bound holds
throughout the selected interval for every realized
history, or on a common probability-one event.
Define
\begin{equation}
    \widehat C_{n,\alpha}
    =
    \Theta\setminus
    \bigcup_{s=1}^n\ \bigcup_{I\in\mathcal R_s} I.
    \label{eq:continuous-outer-set}
\end{equation}

\begin{proposition}[Certified continuous-parameter inversion]
\label{prop:continuous-certified-inversion}
Suppose every excluded interval satisfies
\eqref{eq:continuous-interval-certificate}, and all
selection rules are measurable. Then
$(\widehat C_{n,\alpha})_{n\ge1}$ is nested and satisfies
\[
    \mathcal C_{n,\alpha}
    \subseteq \widehat C_{n,\alpha}
    \qquad\text{for every }n\ge1.
\]
Consequently, for every $\thetatrue\in\Theta$,
\[
    \PP_{\thetatrue}\!\left(
        \thetatrue\in\widehat C_{n,\alpha}
        \text{ for all }n\ge1
    \right)
    \ge 1-\alpha.
\]
\end{proposition}

\begin{proof}
If $\theta$ belongs to an interval excluded by round
$n$, then \eqref{eq:continuous-interval-certificate}
implies $\ell_m(\theta)\ge h_\alpha$ for some $m\le n$.
Thus $\theta\notin\mathcal C_{n,\alpha}$, establishing
the containment. Nesting follows from
\eqref{eq:continuous-outer-set}, and the coverage
guarantee follows from Ville's inequality applied to
$M^{(\thetatrue)}$.
\end{proof}

Intervals may be selected and subdivided adaptively
after observing the current data. No multiplicity
adjustment is needed for the number of intervals or
refinement steps: the numerical procedure only
enlarges the ideal confidence sequence. In particular,
a computation may stop at any prescribed or
data-dependent budget, provided all unresolved
regions are retained.

\paragraph{Constructing interval bounds.}
Define
\begin{equation*}
    d_{ij}(\theta)
    =
    \log T_i^{(\theta)}\!\left(S(\theta,U_{ij})\right)
    -
    \log T_i^{(\theta)}(Y_i).
\end{equation*}
Dividing the numerator and denominator of
\eqref{eq:main-evalue} by the observed score gives
\begin{equation}
\begin{aligned}
    \ell_m(\theta)
    ={}&m\log(K+1)\\
    &-\sum_{i=1}^m
    \log\!\left(
        1+\sum_{j=1}^K e^{d_{ij}(\theta)}
    \right).
\end{aligned}
\label{eq:continuous-log-evidence}
\end{equation}
Suppose verified calculations provide upper bounds
\[
    d_{ij}(\theta)\le\overline d_{ij}(I)
    \qquad\text{for all }\theta\in I.
\]
Monotonicity of log-sum-exp yields the valid lower bound
\begin{equation}
\begin{aligned}
    \underline\ell_m(I)
    ={}&m\log(K+1)\\
    &-\sum_{i=1}^m
    \log\!\left(
        1+\sum_{j=1}^K e^{\overline d_{ij}(I)}
    \right).
\end{aligned}
\label{eq:continuous-score-envelope}
\end{equation}
Hence $I$ may be excluded whenever
$\max_{m\le n}\underline\ell_m(I)\ge h_\alpha$.
Lower bounds on $d_{ij}$ analogously give upper bounds
on $\ell_m$. If these upper bounds remain below
$h_\alpha$ for every $m\le n$, the entire interval is
certified to belong to $\mathcal C_{n,\alpha}$.

The required bounds may be obtained by propagating
interval enclosures through the simulator and score
network. An alternative is a verified local
Lipschitz bound. If $I=[c_I-r_I,c_I+r_I]$ and
$\ell_m$ is $L_m(I)$-Lipschitz throughout $I$, then
\begin{equation}
    \underline\ell_m(I)
    =
    \ell_m(c_I)-L_m(I)r_I
    \label{eq:continuous-lipschitz-envelope}
\end{equation}
is valid. The constant must control the full
theta-dependent calculation, including the reference
simulations; gradients evaluated only at finitely
many points do not establish such a bound.

\paragraph{Location simulators and ReLU scores.}
The one-dimensional experimental models in \cref{app:simulators}
(Gaussian and $g$-and-$k$) have the form
\[
    S(\theta,U)=\theta+q(U).
\]
A sharper certificate is available when the score
network is composed of affine maps and ReLU
activations, with fixed input standardization and
frozen batch-normalization statistics during
evaluation. Residual additions preserve this
structure.

For the neural score in \eqref{eq:neural-score},
write $q_{ij}=q(U_{ij})$. Conditional on the realized
history and reference draws,
\begin{align*}
    d_{ij}(\theta)
    ={}&
    s_{\phi_{i-1}}(
        \widehat\theta_{i-1},\theta+q_{ij})\\
    &-s_{\phi_{i-1}}(\theta,\theta+q_{ij})\\
    &-s_{\phi_{i-1}}(\widehat\theta_{i-1},Y_i)
      +s_{\phi_{i-1}}(\theta,Y_i).
\end{align*}
Each nonconstant network input follows an affine
line as $\theta$ varies. The network restricted to
such a line is continuous and piecewise affine.
Consequently, there is a finite partition of
$\Theta$ into intervals on which every
$d_{ij}$, for $i\le m$ and $j\le K$, is affine.

Consider one such interval $I=[a,b]$, and write
$d_{ij}(\theta)=u_{ij}\theta+v_{ij}$ on $I$.
Equation~\eqref{eq:continuous-log-evidence} becomes
\[
\begin{aligned}
    \ell_m(\theta)
    ={}&m\log(K+1)\\
    &-\sum_{i=1}^m
    \log\!\left(
        1+\sum_{j=1}^K
        e^{u_{ij}\theta+v_{ij}}
    \right).
\end{aligned}
\]
Since log-sum-exp of affine functions is convex,
$\ell_m$ is concave on $I$. Therefore,
\begin{equation}
    \inf_{\theta\in[a,b]}\ell_m(\theta)
    =
    \min\{\ell_m(a),\ell_m(b)\}.
    \label{eq:continuous-relu-certificate}
\end{equation}
In particular, rejection of both endpoints at the
same time $m$ certifies rejection of the entire
interval. More generally, the chord inequality
\[
\begin{aligned}
    \ell_m((1-t)a+tb)
    \ge{}&(1-t)\ell_m(a)+t\ell_m(b),\\
    &\hspace{3em}0\le t\le1,
\end{aligned}
\]
provides a lower bound throughout the interval.

The partition must account for activation changes
in all relevant historical network evaluations.
Such partitions can be constructed by propagating
affine representations through the layers and
splitting intervals at zeros of ReLU preactivations;
see \citet{SotoudehThakur2019}.
Within each resulting interval, the fixed-time
rejection region $\{\theta:\ell_m(\theta)\ge h_\alpha\}$
is an interval, possibly empty. Its boundaries can
be enclosed using certified one-dimensional
optimization and root finding. Unresolved boundary
regions are retained.

Two qualifications are necessary.
First, arbitrary grid intervals may contain
activation breakpoints, so
\eqref{eq:continuous-relu-certificate} cannot be
applied without verifying the affine partition.
Second, $\max_{m\le n}\ell_m$ need not be concave:
endpoint crossings at different times do not
establish rejection throughout the interval.
The certificate is applied separately at each time,
with certified exclusions retained thereafter.

\paragraph{Adaptive evaluation and numerical rigor.}
A practical implementation can first apply coarse
interval bounds, subdivide unresolved intervals,
and exploit \eqref{eq:continuous-relu-certificate}
where activation patterns are certified.
Full activation-partition enumeration is therefore
optional, and may be computationally expensive.

When a new parameter or subinterval is evaluated,
its historical factors must use the original
reference draws and the original score states,
including the historical alternatives and network
parameters. Refinement changes the numerical
representation of the confidence set, not the
underlying martingale. The next alternative may use
all information available from the completed round;
an exact minimizer over $\Theta$ is unnecessary
for validity.

Formal numerical certification requires enclosures
that account for rounding error in network
evaluations, activation boundaries, and threshold
comparisons. In particular, ordinary floating-point
breakpoint calculations alone do not establish
\eqref{eq:continuous-interval-certificate}.
With sound deterministic bounds, no additional
statistical error allowance is required.
The accuracy of the learned ratio remains irrelevant
to coverage: the certificates concern the actual
chosen score, rather than its unknown population
target.

\section{IMPLEMENTATION AND COMPUTE}
\label{app:implementation}

All experiments are implemented in \texttt{PyTorch} and \texttt{NumPy}. Randomness is controlled throughout by a single base seed expanded via NumPy \texttt{SeedSequence} spawning into independent per-component and per-repetition streams, so every training run, reference simulation, and repetition is individually reproducible. All experiments were run on an Apple Silicon M4 Pro with 24 GB of memory, using PyTorch's Metal (MPS) backend for GPU acceleration.

\paragraph{Cost.} Let $G=|\Theta_{\mathrm{eval}}|$ be the number of
candidates and $K$ the number of references per candidate. At each step the
procedure draws $K$ references for every candidate, so $GK$ simulator calls,
and evaluates the score $(K+1)G$ times; updating the wealth and choosing the
minimum-evidence alternative costs $O(G)$. Over $n$ steps this gives
$O(nGK)$ simulator calls and score evaluations, while memory stays $O(GK)$,
since only the current wealth of each candidate and the current step's
references must be kept. The network is trained once before monitoring, a
fixed cost that does not grow with $n$. The procedure needs no sample size
fixed in advance: by \cref{prop:exact-martingale} the confidence sequence
is valid at any stopping time.

\section{EXPERIMENTAL DETAILS FOR THE SYNTHETIC EXPERIMENTS}
\label{app:experimental-details}

\subsection{Simulators}
\label{app:simulators}

\paragraph{Gaussian.} The Gaussian simulator returns $Y = \theta + Z$ with
$Z \sim \mathcal{N}(0, 1)$ and $\theta \in \Theta = [-4, 4]$. Its likelihood
is available in closed form, so it serves as a tractable benchmark where every
oracle quantity can be computed exactly.

\paragraph{$g$-and-$k$.} The $g$-and-$k$ simulator
\citep{rayner2002numerical} returns $Y = Q_\theta(Z)$ with
$Z \sim \mathcal{N}(0,1)$ and
\begin{equation*}
    Q_\theta(z) = \theta + b\,\bigl(1 + c \tanh(gz/2)\bigr)\, z\, \bigl(1 + z^2\bigr)^k,
\end{equation*}
where $(b, g, k, c) = (1, 2, 0.5, 0.8)$ as is standard. Sampling is a single
standard-normal draw pushed through $Q_\theta$, but $Q_\theta$ admits no
closed-form inverse, so the density can only be evaluated by numerical root
solving, making this a canonical likelihood-free benchmark.

\paragraph{SLCP.} The SLCP (simple likelihood, complex posterior) simulator
\citep{papamakarios2019sequential}, as specified in \texttt{sbibm}
\citep{lueckmann2021benchmarking}, has $\theta\in\Theta=[-3,3]^5$ and returns
$Y=(Y^{(1)},\ldots,Y^{(4)})\in\mathbb{R}^8$ with
$Y^{(j)}\overset{\mathrm{i.i.d.}}{\sim}\mathcal{N}(m_\theta,S_\theta)$,
\begin{equation*}
    m_\theta=(\theta_1,\theta_2),\qquad
    S_\theta=\begin{pmatrix} s_1^2+\epsilon & \rho s_1 s_2\\
    \rho s_1 s_2 & s_2^2+\epsilon \end{pmatrix},
\end{equation*}
where $s_1=\theta_3^2$, $s_2=\theta_4^2$, $\rho=\tanh(\theta_5)$, and
$\epsilon=10^{-6}$ ensures positive definiteness. The likelihood is available
in closed form, but since $\theta_3$ and $\theta_4$ enter only through their
squares, $(\theta_1,\theta_2,\pm\theta_3,\pm\theta_4,\theta_5)$ define the
same distribution and the posterior has four symmetric modes.

\subsection{Training the neural ratio score}
\label{app:training}

\paragraph{Objective and architecture.} We train the likelihood-to-evidence
network $s_\phi(\theta,y)$ of \cref{sec:backgroundNRE} with the NRE-C
objective (\cref{app:nre-losses}), with cardinality $K_{\mathrm{tr}} = 20$ and
$\gamma = 1$. The network is a residual multilayer perceptron
with three residual blocks of width 128, ReLU activations,
with batch normalization. Inputs $(\theta, y)$ are
standardized using statistics estimated from an initial batch of
$5\times10^4$ simulator draws.

\paragraph{Training data and optimization.} Training is simulation-based and
streaming: at every optimizer step we draw a fresh minibatch of
1024 pairs $\theta\sim q$, $Y\sim P_\theta$, with $q$ uniform on
$\Theta$, for a total of $1.024\times10^7$ simulated
pairs per model. We use Adam with learning rate
$5\times10^{-4}$ and gradient-norm clipping at 5. The network is
trained once per model before monitoring begins and kept fixed throughout;
updating $\phi$ during monitoring is permitted provided every update at time
$i$ uses only $\mathcal F_{i-1}$, and affects efficiency but not validity.

\paragraph{Monitoring training quality via expected log growth.} The natural
figure of merit for a score used in the test martingale is its expected
one-step log growth \eqref{eq:finite-k-growth}. For a score
$T(y)=\exp\{h(y)\}$ with $X_0\sim P_\vartheta$ and
$X_1,\ldots,X_K\overset{\mathrm{i.i.d.}}{\sim}P_\theta$, define
\begin{equation}
\mathcal L_K(h;\theta,\vartheta)
=
-\EE\!\left[
\log
\frac{\exp\{h(X_0)\}}
{\sum_{j=0}^K\exp\{h(X_j)\}}
\right],
\label{eq:finite-k-contrastive-loss}
\end{equation}
so that
\begin{equation}
\Gamma_K(T;\theta,\vartheta)
=
\log(K+1)-\mathcal L_K(h;\theta,\vartheta),
\label{eq:growth-contrastive-equivalence}
\end{equation}
where $\Gamma_K(T;\theta,\vartheta)$ denotes \eqref{eq:finite-k-growth} with
$P_{\thetatrue}$ replaced by $P_\vartheta$. We therefore monitor training on
a fixed held-out bank of 4{,}096 tuples
$(\theta,\vartheta,X_{0:K})$ with $K=15$, sampled half with
independent prior draws $(\theta,\vartheta)$ and half with local
perturbations $\vartheta=\theta+\varepsilon$,
$\varepsilon\sim\mathcal N(0,0.25^2)$ clipped to $\Theta$, and report the estimated negative expected log growth
$-\Gamma_K = \mathcal L_K - \log(K+1)$ of the trained score against that of
the exact (Gaussian) or numerical ($g$-and-$k$)
likelihood ratio on the same bank.
Since \eqref{eq:growth-contrastive-equivalence} shows that minimizing
$\mathcal L_K$ maximizes the same finite-$K$ growth used by the test
martingale, $\mathcal L_K$ could also serve directly as a training
objective; we do not pursue this here.

\subsection{Running the confidence sequence}
\label{app:running-cs}

For the Gaussian and $g$-and-$k$ models, the martingale is evaluated over a
fixed grid $\Theta_{\mathrm{eval}}$ of 81 equispaced candidates on
$\Theta = [-4,4]$ (see \cref{app:extensions} for coverage at off-grid
values). For SLCP, $\Theta_{\mathrm{eval}}$ consists of 5{,}120 points of a
scrambled Sobol sequence \citep{sobol1967distribution} on $\Theta=[-3,3]^5$, in
which $\thetatrue=(0,0,2.2,2.2,0)$ and its three mirror images replace their
nearest points; together these are the four parameters
$(0,0,\pm2.2,\pm2.2,0)$, which define the same distribution because
$\theta_3$ and $\theta_4$ enter the model only through their squares
(\cref{app:simulators}). In all models we monitor $n = 100$ observations over
$1{,}000$ repetitions at level $\alpha = 0.1$ with $K = 15$ null simulations
per candidate per step. Throughout, we report the nested (running-intersection) version of the
confidence sequence, $C_{i,\alpha} \gets C_{i-1,\alpha} \cap
\{\theta \in \Theta_{\mathrm{eval}} : M_i^{(\theta)} < \alpha^{-1}\}$: a
candidate whose martingale has ever crossed $\alpha^{-1}$ stays excluded.
Since Ville's inequality bounds the probability that $M^{(\thetatrue)}$
\emph{ever} exceeds $\alpha^{-1}$, this intersection preserves
anytime-validity while making the sets monotone.

\paragraph{Predictable alternative.} Any estimator based only on past
information can be used for $\widehat\theta_{i-1}$. For example, it may be a
point estimator based on $Y_{1:i-1}$, a simulator-based posterior summary,
or the minimum-evidence estimator
\begin{equation}
\widehat\theta_{i-1}
\in
\argmin_{\vartheta\in\Theta_{\mathrm{eval}}}
M_{i-1}^{(\vartheta)},
\label{eq:minimum-e-estimator}
\end{equation}
related to ME-estimation and e-posterior methods
\citep{Wang2026,Gruenwald2023}. These choices affect only efficiency, not
validity. We use \eqref{eq:minimum-e-estimator}; \cref{app:ablations}
compares it with model-based running point estimates.

\paragraph{Discarding the first observations (SLCP).} The predictable
alternative must be initialized before any data are observed, and we use a
fixed default, $\widehat\theta_0=(0,0,1,1,0)$ for SLCP. In this problem the
default is far from $\thetatrue$, so the first one-step e-values sharply
decrease the wealth of many candidates, which then takes many observations to
recover. For SLCP we therefore discard the first two e-values and build the
sets from $M_i^{(\theta)}/M_2^{(\theta)}$ for $i\ge3$; the first two
observations are still used to choose $\widehat\theta_2$, so the guarantee
holds from $i=3$ onwards. More generally, discarding a few initial e-values
reduces the dependence on the starting value at the cost of a few
observations.

\paragraph{Reference simulations.} At each step, $K$ fresh pseudo-observations
are generated for every candidate, from random numbers shared across candidates.

\emph{Common random numbers.}
Draws are coupled across candidates by common random numbers: if the simulator is
written as $Y=S(\theta,U)$ with $U$ a vector of standard normals, we use the same
$U_{ij}$ across candidates and set $\widetilde Y_{ij}^{(\theta)}=S(\theta,U_{ij})$
(a single batch of standard normals per step, shifted by $\theta$ for the
Gaussian, pushed through $Q_\theta$ for $g$-and-$k$, or mapped to
$m_\theta+L_\theta z$ for SLCP, with $L_\theta$ the Cholesky factor of
$S_\theta$). This smooths the random field
$\theta\mapsto M_n(\theta)$ over the grid without affecting coverage: for the true
$\thetatrue$ the pseudo-observations still have distribution $P_{\thetatrue}$
and are independent across time, so the coverage proof is unaffected, and
dependence of the auxiliary simulations across candidates does not compromise
pointwise confidence-sequence validity.

\subsection{Confidence-sequence diagnostics}
\label{app:eval}

\emph{Cumulative miscoverage.} Cumulative miscoverage at time $n$ is
the fraction of repetitions in which the confidence sequence has excluded
the true parameter $\thetatrue$ at some time $i \le n$; it is therefore
nondecreasing in $n$, and anytime-validity requires it to stay below
$\alpha$ uniformly. Figures show it with pointwise Wilson $95\%$ intervals;
the table columns headed ``mis.''\ report it at the final time and should be
read against the nominal $\alpha = 0.1$.

\emph{Mean hull width.} For the Gaussian and $g$-and-$k$ models, hull width
measures how tightly the procedure localises $\thetatrue$. At each time the
confidence set is the collection of grid candidates not yet rejected, and we
record the width of its convex hull, i.e.\ the largest retained candidate
minus the smallest, averaged over the repetitions with a nonempty set.
Figures trace it over time; tables report it early, at $n = 10$ ($W_{10}$),
and at the final time ($W_{100}$).

\emph{Final candidate retention.} For the Gaussian and $g$-and-$k$ models,
retention of a grid candidate $\theta$ is the fraction of repetitions in which
$\theta$ remains in the confidence set at the final time, plotted against
$\theta - \thetatrue$; concentration around zero indicates that false
candidates are excluded reliably.

\emph{Mean retained fraction.} For SLCP, whose confidence set is a subset of
a five-dimensional candidate set and so has no single width, we instead
report the fraction of the candidates in $\Theta_{\mathrm{eval}}$ that remain
in the confidence set, averaged over all repetitions.

\section{ADDITIONAL RESULTS FOR THE SYNTHETIC EXPERIMENTS}
\label{app:synthetic-results}

\subsection{Neural ratio quality (Gaussian and $g$-and-$k$)}
\label{app:nre-eval}

For the Gaussian and $g$-and-$k$ models, we assess the baseline network of
Section~\ref{app:training}, separately for the contrastive (NRE-C) and binary
(NRE-A) objectives. The score enters the martingale only through the log-ratio
$s_\phi(\theta_1,y)-s_\phi(\theta_0,y)$, so we measure how well this recovers
the true log-ratio and the log-growth it delivers.
On all ordered pairs $(\theta_0,\theta_1)$ from a $21$-point equispaced
subgrid of $\Theta$ (every fourth candidate of $\Theta_{\mathrm{eval}}$) we
report three static, pair-level measures. \emph{Excess BCE} is the network's
binary cross-entropy on discriminating $f_{\theta_1}$ from $f_{\theta_0}$, minus
that of the optimal classifier (whose logit is the true log-ratio); this
difference is zero exactly when the score recovers the log-ratio, and positive
otherwise. \emph{Log-ratio RMSE} and \emph{bias} are the root-mean-square and
mean signed error of $s_\phi(\theta_1,\cdot)-s_\phi(\theta_0,\cdot)$ against the
exact (Gaussian) or numerical ($g$-and-$k$) log-ratio. We complement these with
the dynamic \emph{log-growth validation} of Section~\ref{app:training}: along training we track the negative expected one-step log-growth
$-\Gamma_K$ of the current score against that of the oracle likelihood
ratio on the held-out bank, reported
on three subsets: \emph{broad} pairs, with $\theta$ and $\vartheta$ drawn
independently from the prior; \emph{local} pairs, with $\vartheta=\theta+\varepsilon$
a small perturbation of $\theta$; and \emph{overall}, the equal ($50/50$) mixture
of the two that makes up the full bank. A score that recovers the oracle drives
every panel down to the oracle line.

These diagnostics explain why the neural confidence sequences in
\cref{fig:main-experiment,fig:cs-baseline} are essentially indistinguishable from the
exact-likelihood-ratio sequences: at baseline capacity both scores achieve
near-oracle log-growth, even where they do
not recover the log-ratio itself.

\paragraph{Pair-level accuracy.}
Excess BCE is small everywhere for both scores and both models, confirming the
ratio classifiers have converged. The log-ratio RMSE and bias are near zero on and around the diagonal
$\theta_0\approx\theta_1$ and grow toward the off-diagonal corners, where NRE-C is
generally the weaker of the two: its off-diagonal error exceeds NRE-A's on
$g$-and-$k$ in both RMSE and bias (Fig.~\ref{fig:heat-gandk}) and on the Gaussian
in bias (Fig.~\ref{fig:heat-gauss}). These off-diagonal errors do not,
however, reflect in the confidence sequences, which are near-identical across
scores (\cref{fig:cs-baseline}).  This is because efficiency is
governed not by log-ratio accuracy but by the expected log-growth, and for
well-separated laws the growth saturates: once a candidate is clearly distinct
from the truth, any roughly correct score already achieves near-maximal growth,
so the large off-diagonal errors leave efficiency essentially untouched. Pair-level accuracy is
therefore the wrong lens for efficiency; the relevant object is the log-growth
itself, which we examine next.

\paragraph{Log-growth validation.}
Tracking expected one-step log-growth along training
(Figs.~\ref{fig:val-gauss-nrea}--\ref{fig:val-gandk-nrec}), both scores descend
onto the oracle likelihood-ratio line in all three subsets (broad, local, and
overall) for both models. Since
expected log-growth is exactly the efficiency criterion of
Proposition~\ref{prop:optimal-score}, this convergence directly explains why the neural
sequences in \cref{fig:cs-baseline} match the exact-likelihood-ratio sequences
in width and retention.

\begin{figure}[!htbp]
  \centering
  \includegraphics[width=0.85\linewidth]{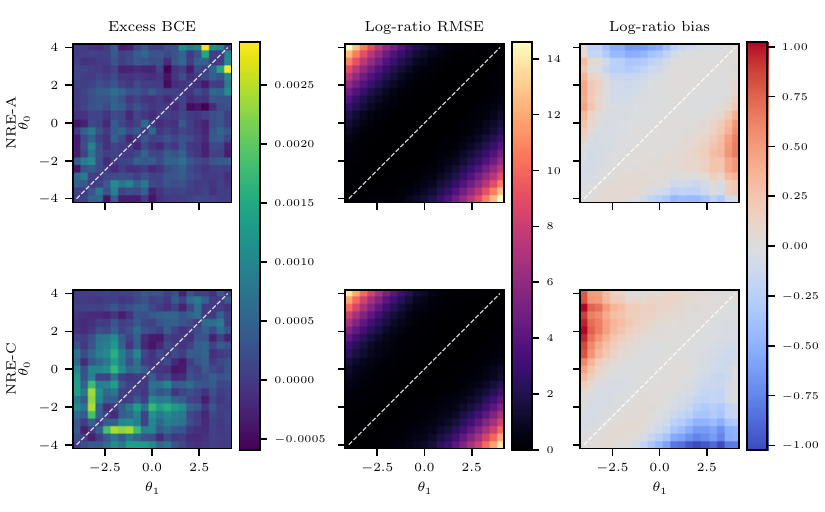}
  \caption{Gaussian model, baseline network. Per parameter-pair
    $(\theta_0,\theta_1)$: excess BCE (left), log-ratio RMSE (center), and
    log-ratio bias (right), for NRE-A (top) and NRE-C (bottom).}
  \label{fig:heat-gauss}
\end{figure}

\begin{figure}[!htbp]
  \centering
  \includegraphics[width=0.85\linewidth]{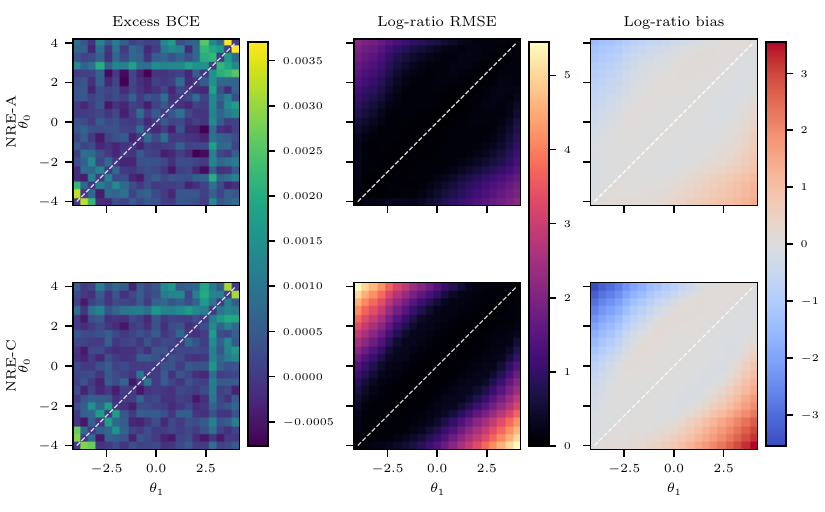}
  \caption{As Fig.~\ref{fig:heat-gauss}, for the $g$-and-$k$ model.}
  \label{fig:heat-gandk}
\end{figure}

\begin{figure}[!htbp]
  \centering
  \begin{subfigure}{0.49\linewidth}
    \centering
    \includegraphics[width=\linewidth]{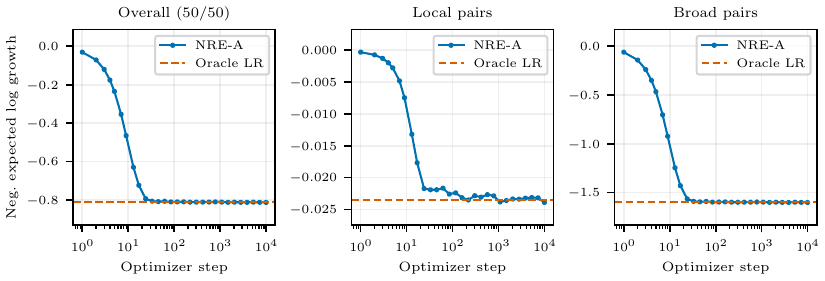}
    \caption{Gaussian}
    \label{fig:val-gauss-nrea}
  \end{subfigure}\hfill
  \begin{subfigure}{0.49\linewidth}
    \centering
    \includegraphics[width=\linewidth]{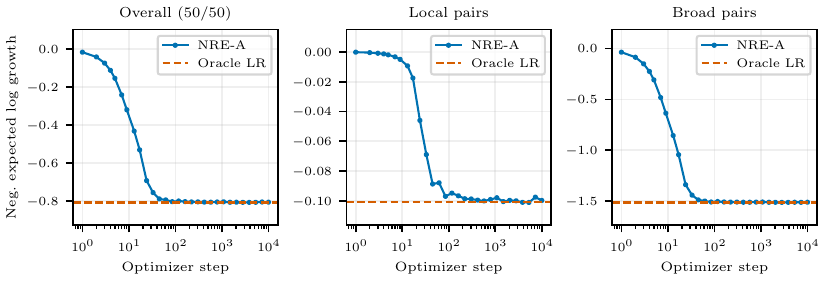}
    \caption{$g$-and-$k$}
    \label{fig:val-gandk-nrea}
  \end{subfigure}
  \caption{NRE-A log-growth validation along training, baseline network.
    Negative expected one-step log-growth of the trained score (blue) against the
    oracle likelihood ratio (dashed), on broad, local, and overall pairs.}
  \label{fig:val-nrea}
\end{figure}

\begin{figure}[!htbp]
  \centering
  \begin{subfigure}{0.49\linewidth}
    \centering
    \includegraphics[width=\linewidth]{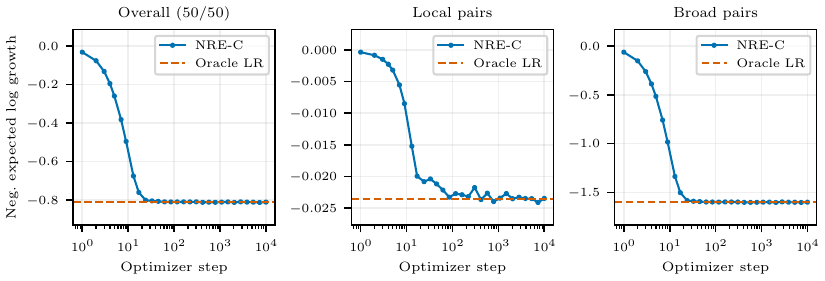}
    \caption{Gaussian}
    \label{fig:val-gauss-nrec}
  \end{subfigure}\hfill
  \begin{subfigure}{0.49\linewidth}
    \centering
    \includegraphics[width=\linewidth]{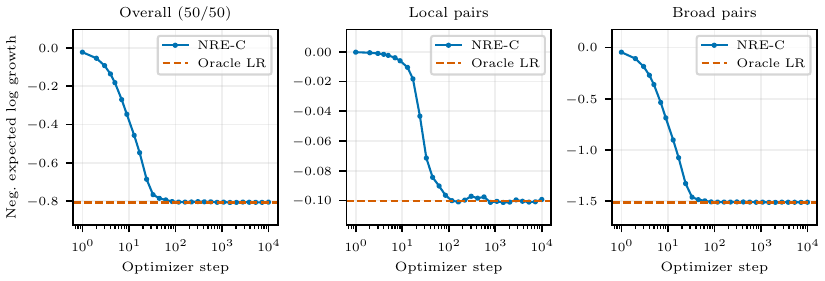}
    \caption{$g$-and-$k$}
    \label{fig:val-gandk-nrec}
  \end{subfigure}
  \caption{As Fig.~\ref{fig:val-nrea}, for NRE-C.}
  \label{fig:val-nrec}
\end{figure}

\subsubsection{Small architecture}

We repeat these diagnostics for the small architecture of the ablations
(\cref{app:ablations}), a single hidden layer of two units.

\paragraph{Pair-level accuracy.}
Under the deliberately underparameterized network the two objectives separate
sharply, reversing their baseline ordering. NRE-C's excess BCE stays
near zero and its bias small on both models; its log-ratio RMSE stays low on
$g$-and-$k$ (Fig.~\ref{fig:heat-gandk-small}), but on the Gaussian it grows in the
far-apart corners about as much as NRE-A's (Fig.~\ref{fig:heat-gauss-small}).
NRE-A degrades markedly on both: its excess BCE is now visibly elevated, and its
log-ratio RMSE and bias grow several-fold. The degradation is far
worse on $g$-and-$k$, where NRE-A's RMSE is large across almost the entire grid.
Crucially, this error is no longer confined to the far-apart, saturated corners
that were harmless at baseline: on $g$-and-$k$ it reaches the near-diagonal band,
which drives localization and where growth is \emph{not} saturated. This is what
makes the small-architecture failure of NRE-A costly rather than cosmetic, as the
log-growth confirms next.

\paragraph{Log-growth validation.}
The one-step log-growth (Figs.~\ref{fig:val-gauss-nrea-small}--\ref{fig:val-gandk-nrec-small})
makes the separation explicit. NRE-C moves toward the oracle likelihood-ratio
line in all three subsets on both models: it reaches it on $g$-and-$k$, while
on the Gaussian it is still approaching it when training stops, most visibly on
the local pairs. NRE-A does not. On the Gaussian
it plateaus visibly above the oracle in every subset, a persistent but bounded
log-growth deficit. On $g$-and-$k$ it collapses, remaining almost flat and
capturing a small fraction of the available growth (its
overall curve settles near $-0.11$ against an oracle near $-0.81$). Because this
shortfall reaches the contested near-diagonal region rather than the saturated
tails, it translates directly into wider, less concentrated confidence sequences
under NRE-A. This is severe on $g$-and-$k$, where its sets fill much of the grid
and barely localize $\theta_*$ (Table~\ref{tab:ablations}). Coverage
nonetheless remains valid throughout: a misspecified score forfeits power,
never validity, and the contrastive objective is markedly more robust to the
misspecification.

\begin{figure}[!htbp]
  \centering
  \includegraphics[width=0.85\linewidth]{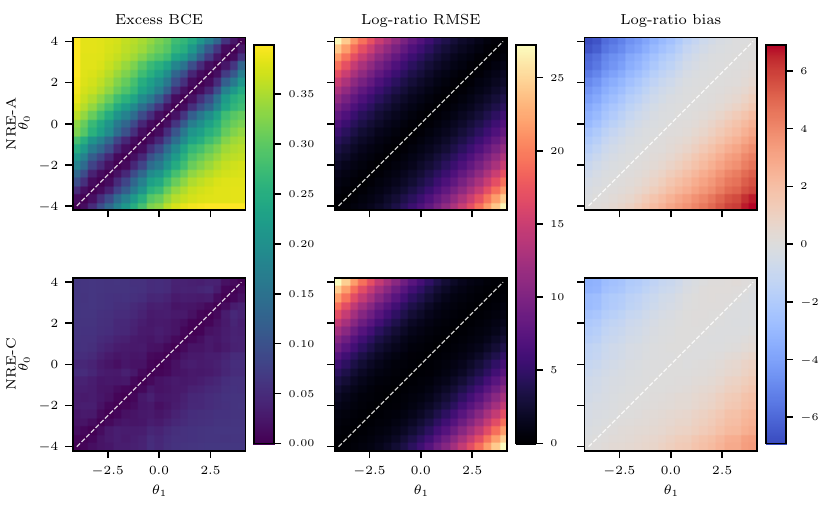}
  \caption{Gaussian model, small (underparameterized) network. Per parameter-pair
    $(\theta_0,\theta_1)$: excess BCE (left), log-ratio RMSE (center), and
    log-ratio bias (right), for NRE-A (top) and NRE-C (bottom).}
  \label{fig:heat-gauss-small}
\end{figure}

\begin{figure}[!htbp]
  \centering
  \includegraphics[width=0.85\linewidth]{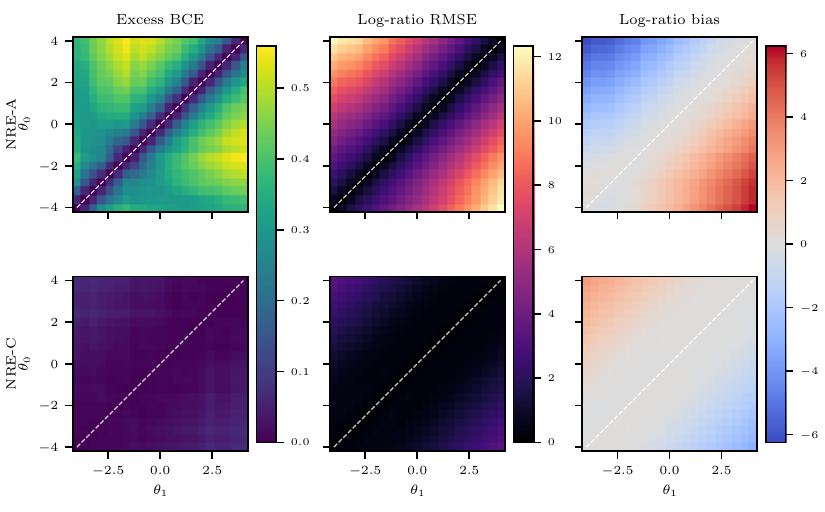}
  \caption{As Fig.~\ref{fig:heat-gauss-small}, for the $g$-and-$k$ model.}
  \label{fig:heat-gandk-small}
\end{figure}

\begin{figure}[!htbp]
  \centering
  \begin{subfigure}{0.49\linewidth}
    \centering
    \includegraphics[width=\linewidth]{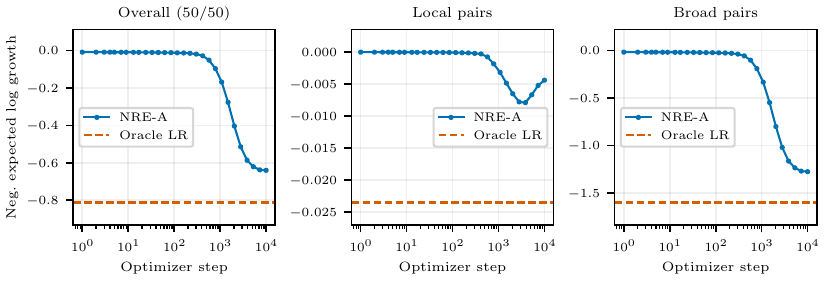}
    \caption{Gaussian}
    \label{fig:val-gauss-nrea-small}
  \end{subfigure}\hfill
  \begin{subfigure}{0.49\linewidth}
    \centering
    \includegraphics[width=\linewidth]{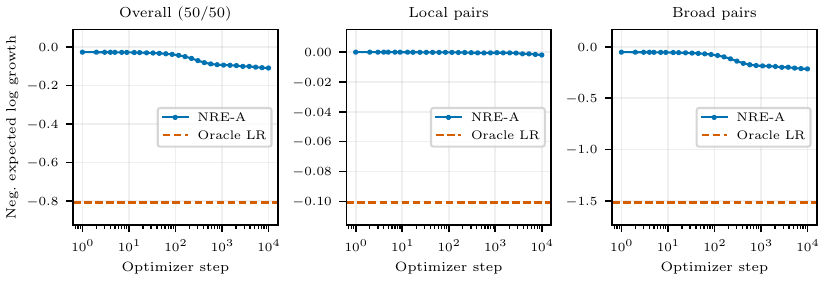}
    \caption{$g$-and-$k$}
    \label{fig:val-gandk-nrea-small}
  \end{subfigure}
  \caption{NRE-A log-growth validation along training, small (underparameterized)
    network. Negative expected one-step log-growth of the trained score (blue)
    against the oracle likelihood ratio (dashed), on broad, local, and overall
    pairs.}
  \label{fig:val-nrea-small}
\end{figure}

\begin{figure}[!htbp]
  \centering
  \begin{subfigure}{0.49\linewidth}
    \centering
    \includegraphics[width=\linewidth]{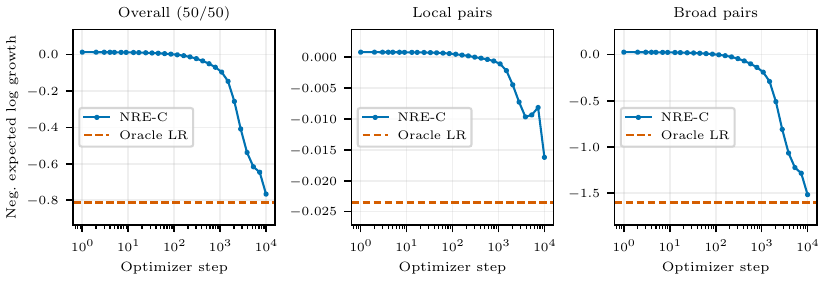}
    \caption{Gaussian}
    \label{fig:val-gauss-nrec-small}
  \end{subfigure}\hfill
  \begin{subfigure}{0.49\linewidth}
    \centering
    \includegraphics[width=\linewidth]{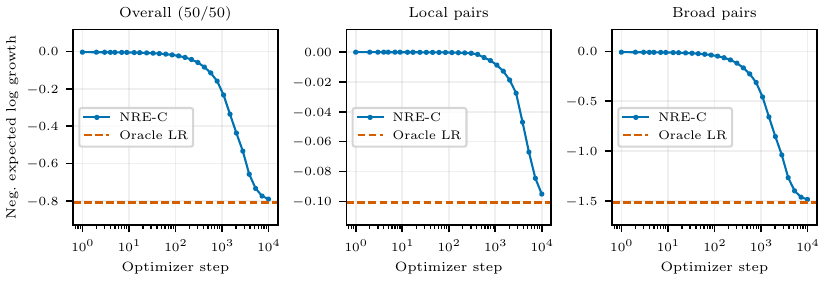}
    \caption{$g$-and-$k$}
    \label{fig:val-gandk-nrec-small}
  \end{subfigure}
  \caption{As Fig.~\ref{fig:val-nrea-small}, for NRE-C.}
  \label{fig:val-nrec-small}
\end{figure}

\subsection{Ablations (Gaussian and $g$-and-$k$)}
\label{app:ablations}

\subsubsection{Set-up}

Because coverage is guaranteed pointwise at the true parameter, we
repeat every ablation at three truths $\thetatrue\in\{-2,0,2\}$, one
central and two off-center in $\Theta=[-4,4]$, to check that
behavior does not depend on where the truth sits. Over $1{,}000$
repetitions each, we vary one component of the baseline procedure at a
time (minimum-evidence alternative, $K=15$, the score network of
\cref{app:training}):

\paragraph{Score.}
The score $T_i^{(\theta)}$ determines efficiency but not validity, so we
compare four choices of increasing practicality. The \emph{direct
exact-LR} martingale uses the true likelihood ratio without
conformalization; its \emph{conformalized} counterpart applies the
soft-rank normalization of \eqref{eq:main-evalue} to the same exact
ratio, isolating the cost of conformalization alone. The two neural
scores replace the likelihood ratio with a learned estimate:
NRE-C, the contrastive estimator of the main text, and
NRE-A, the binary-classification ratio estimator
(\citealt{Hermans2020}) (see \cref{app:nre-losses}). Comparing all four separates the price of
conformalization from the price of learning the ratio.

\paragraph{$K$.}
$K$ is the number of fresh null simulations drawn per candidate at each
step, forming the exchangeable reference batch against which the
observation is scored. For the oracle score, the finite-$K$ gap in expected log-growth vanishes as $K\to\infty$ (\cref{prop:optimal-score}). We contrast the
baseline $K=15$ with the minimal $K=1$ to examine how strongly the
simulation budget affects the sequence.

\paragraph{$\widehat\theta_{i-1}$.}
The predictable alternative $\widehat\theta_{i-1}$ is the reference
parameter against which each candidate $\theta$ is scored, and must
depend only on $\mathcal F_{i-1}$. The baseline chooses it by minimum
evidence, $\widehat\theta_{i-1}\in\argmin_{\vartheta}M_{i-1}^{(\vartheta)}$, the candidate the accumulated evidence currently favors. We compare
this against a simpler running point estimate of the observations so far
that does not require evaluating the martingale over the grid: the
running mean for the Gaussian model and the running median for
$g$-and-$k$.

\paragraph{Architecture.}
The neural scores are trained networks, so their quality depends on
capacity. To probe robustness to a poorly estimated ratio, we replace
the baseline network with a deliberately underparameterized ``small''
architecture i.e., a single hidden layer of two units, far too small to
represent the true log-ratio, and ask whether coverage survives when
the score is badly misspecified.

\subsubsection{Aggregated results}

\Cref{tab:ablations} reports, for every ablation block at all three truths,
cumulative miscoverage at $n=100$ and the mean hull width at an early and the
final time, $W_{10}$ and $W_{100}$. Since the widths are conditioned on
nonempty sets, \cref{tab:nonempty} gives the corresponding rates: they never
fall below $96\%$, so the conditioning is a negligible correction almost
everywhere. The one exception is the underparameterized NRE-A on $g$-and-$k$ at
$\thetatrue=0$ and $2$, whose sets never empty ($100\%$) because it rejects almost nothing: with
$W_{100} = 5.5$ on a width-$8$ grid, it retains most candidates and has
essentially failed to localize $\thetatrue$.

\begin{table}[!tp]
  \centering
  \small
  \setlength{\tabcolsep}{4pt}
  \caption{Ablations at $\alpha=0.1$ over $1{,}000$ repetitions, $n=100$ monitored observations. Each ablation block lists only the estimators it affects. Columns report cumulative miscoverage at $n=100$, and mean hull width at $n=10$ and $n=100$.}
  \label{tab:ablations}
  \begin{tabular}{l ccc ccc ccc}
    \toprule
    & \multicolumn{3}{c}{$\thetatrue=-2$} & \multicolumn{3}{c}{$\thetatrue=0$} & \multicolumn{3}{c}{$\thetatrue=2$} \\
    \cmidrule(lr){2-4}\cmidrule(lr){5-7}\cmidrule(lr){8-10}
    Estimator & mis.\ & $W_{10}$ & $W_{100}$ & mis.\ & $W_{10}$ & $W_{100}$ & mis.\ & $W_{10}$ & $W_{100}$ \\
    \midrule
    \multicolumn{10}{l}{\textbf{\boldmath \emph{Gaussian}}} \\
    \multicolumn{10}{l}{\quad \textbf{\boldmath Baseline (min.-evidence, $K=15$)}} \\
    \qquad Direct LR & 0.049 & 1.885 & 0.584 & 0.043 & 1.486 & 0.458 & 0.034 & 1.922 & 0.593 \\
    \qquad Conformal LR & 0.054 & 1.757 & 0.556 & 0.044 & 1.493 & 0.467 & 0.050 & 1.797 & 0.568 \\
    \qquad Conformal NRE-A & 0.051 & 1.755 & 0.556 & 0.053 & 1.491 & 0.464 & 0.051 & 1.796 & 0.562 \\
    \qquad Conformal NRE-C & 0.055 & 1.773 & 0.574 & 0.049 & 1.496 & 0.470 & 0.052 & 1.799 & 0.571 \\
    \addlinespace
    \multicolumn{10}{l}{\quad \textbf{\boldmath $K=1$ null simulation}} \\
    \qquad Conformal LR & 0.060 & 2.407 & 0.718 & 0.065 & 2.055 & 0.612 & 0.076 & 2.392 & 0.710 \\
    \qquad Conformal NRE-A & 0.062 & 2.398 & 0.717 & 0.066 & 2.055 & 0.612 & 0.078 & 2.395 & 0.710 \\
    \qquad Conformal NRE-C & 0.062 & 2.410 & 0.733 & 0.068 & 2.044 & 0.614 & 0.076 & 2.388 & 0.713 \\
    \addlinespace
    \multicolumn{10}{l}{\quad \textbf{\boldmath Model-based plug-in alternative}} \\
    \qquad Direct LR & 0.048 & 1.886 & 0.581 & 0.045 & 1.484 & 0.454 & 0.033 & 1.922 & 0.594 \\
    \qquad Conformal LR & 0.051 & 1.768 & 0.557 & 0.055 & 1.492 & 0.458 & 0.054 & 1.807 & 0.569 \\
    \qquad Conformal NRE-A & 0.051 & 1.767 & 0.558 & 0.056 & 1.492 & 0.457 & 0.052 & 1.805 & 0.568 \\
    \qquad Conformal NRE-C & 0.049 & 1.790 & 0.576 & 0.055 & 1.492 & 0.462 & 0.051 & 1.808 & 0.575 \\
    \addlinespace
    \multicolumn{10}{l}{\quad \textbf{\boldmath Small architecture}} \\
    \qquad Conformal NRE-A & 0.078 & 2.161 & 0.752 & 0.043 & 2.229 & 0.843 & 0.064 & 2.470 & 1.038 \\
    \qquad Conformal NRE-C & 0.058 & 1.782 & 0.637 & 0.058 & 1.625 & 0.584 & 0.055 & 1.872 & 0.677 \\
    \addlinespace
    \midrule
    \multicolumn{10}{l}{\textbf{\boldmath \emph{$g$-and-$k$}}} \\
    \multicolumn{10}{l}{\quad \textbf{\boldmath Baseline (min.-evidence, $K=15$)}} \\
    \qquad Direct LR & 0.040 & 1.968 & 0.384 & 0.061 & 1.303 & 0.252 & 0.027 & 1.400 & 0.284 \\
    \qquad Conformal LR & 0.045 & 1.949 & 0.380 & 0.064 & 1.219 & 0.245 & 0.066 & 1.180 & 0.240 \\
    \qquad Conformal NRE-A & 0.047 & 1.960 & 0.384 & 0.061 & 1.223 & 0.248 & 0.067 & 1.185 & 0.249 \\
    \qquad Conformal NRE-C & 0.045 & 1.967 & 0.387 & 0.063 & 1.212 & 0.244 & 0.067 & 1.189 & 0.247 \\
    \addlinespace
    \multicolumn{10}{l}{\quad \textbf{\boldmath $K=1$ null simulation}} \\
    \qquad Conformal LR & 0.058 & 2.523 & 0.471 & 0.065 & 1.910 & 0.333 & 0.075 & 2.380 & 0.335 \\
    \qquad Conformal NRE-A & 0.054 & 2.439 & 0.460 & 0.071 & 1.830 & 0.324 & 0.074 & 2.510 & 0.327 \\
    \qquad Conformal NRE-C & 0.056 & 2.419 & 0.460 & 0.071 & 1.835 & 0.322 & 0.077 & 2.460 & 0.329 \\
    \addlinespace
    \multicolumn{10}{l}{\quad \textbf{\boldmath Model-based plug-in alternative}} \\
    \qquad Direct LR & 0.042 & 2.300 & 0.564 & 0.041 & 1.825 & 0.456 & 0.023 & 1.932 & 0.484 \\
    \qquad Conformal LR & 0.045 & 2.348 & 0.558 & 0.056 & 1.775 & 0.451 & 0.059 & 1.747 & 0.450 \\
    \qquad Conformal NRE-A & 0.048 & 2.347 & 0.558 & 0.056 & 1.766 & 0.455 & 0.061 & 1.750 & 0.462 \\
    \qquad Conformal NRE-C & 0.049 & 2.334 & 0.554 & 0.056 & 1.736 & 0.449 & 0.060 & 1.764 & 0.458 \\
    \addlinespace
    \multicolumn{10}{l}{\quad \textbf{\boldmath Small architecture}} \\
    \qquad Conformal NRE-A & 0.001 & 6.739 & 1.067 & 0.000 & 6.157 & 5.501 & 0.000 & 6.147 & 5.500 \\
    \qquad Conformal NRE-C & 0.052 & 2.001 & 0.389 & 0.055 & 1.290 & 0.248 & 0.055 & 1.314 & 0.250 \\
    \bottomrule
  \end{tabular}
\end{table}

Across the baseline block, the three conformal estimators agree to within
Monte Carlo error in both miscoverage and width, at all three truths and on
both models: neither neural score is distinguishable from the conformalized
exact likelihood ratio. The direct exact-LR martingale departs from them in
some cases, mainly at $\thetatrue=\pm2$, where it is more conservative and its
sets are wider. Baseline miscoverage stays below $0.07$, comfortably below the
nominal $0.1$, the conservativeness expected of Ville's inequality. Replacing
minimum evidence with a model-based running estimate leaves the Gaussian
results essentially unchanged, but widens the $g$-and-$k$ sets substantially.

Reducing the null budget to $K=1$ preserves validity but widens the sets, by
roughly $20$--$40\%$ at $n=100$ on both models, while miscoverage rises toward
the nominal level. Both follow from the
added Monte Carlo noise in the single-draw e-value, false candidates
are excluded more slowly, widening the sets, while the truth's own
martingale fluctuates more and approaches the Ville bound. Neither
effect breaches validity.

The small-architecture block is the clearest separation of validity from
efficiency. Deliberately underparameterizing the score leaves NRE-C
close to its well-specified baseline but degrades NRE-A, and the
degradation is model-dependent: mild on the Gaussian ($W_{100}$ up to
$1.04$ against a baseline near $0.56$), but a collapse on $g$-and-$k$ at
$\thetatrue=0$ and $2$, where NRE-A's miscoverage falls to essentially zero
and $W_{100}$ inflates to $5.5$ on a grid of width $8$. In every case coverage
remains valid; a misspecified score forfeits power,
never validity, and NRE-C's contrastive objective is markedly more
robust to this misspecification.

\begin{table}[t]
  \centering
  \small
  \setlength{\tabcolsep}{4pt}
  \caption{Percentage of nonempty confidence sets at $n=100$, over $1{,}000$ repetitions. Mean hull widths in Table~\ref{tab:ablations} are averaged over these repetitions only. Rows follow the same blocks as Table~\ref{tab:ablations}.}
  \label{tab:nonempty}
  \begin{tabular}{l ccc ccc}
    \toprule
    & \multicolumn{3}{c}{Gaussian} & \multicolumn{3}{c}{$g$-and-$k$} \\
    \cmidrule(lr){2-4}\cmidrule(lr){5-7}
    Estimator & $\thetatrue=-2$ & $\thetatrue=0$ & $\thetatrue=2$ & $\thetatrue=-2$ & $\thetatrue=0$ & $\thetatrue=2$ \\
    \midrule
    \multicolumn{7}{l}{\textbf{\boldmath Baseline (min.-evidence, $K=15$)}} \\
    \quad Direct LR & 98.1 & 99.9 & 98.8 & 96.6 & 98.0 & 98.8 \\
    \quad Conformal LR & 98.6 & 99.8 & 98.4 & 96.4 & 97.8 & 97.6 \\
    \quad Conformal NRE-A & 98.6 & 99.8 & 98.4 & 96.3 & 97.9 & 97.6 \\
    \quad Conformal NRE-C & 98.6 & 99.9 & 98.4 & 96.4 & 98.2 & 97.7 \\
    \addlinespace
    \multicolumn{7}{l}{\textbf{\boldmath $K=1$ null simulation}} \\
    \quad Conformal LR & 99.3 & 100.0 & 99.1 & 97.5 & 99.1 & 98.5 \\
    \quad Conformal NRE-A & 99.3 & 100.0 & 99.0 & 97.7 & 99.2 & 98.3 \\
    \quad Conformal NRE-C & 99.3 & 100.0 & 98.9 & 97.6 & 99.1 & 98.2 \\
    \addlinespace
    \multicolumn{7}{l}{\textbf{\boldmath Model-based plug-in alternative}} \\
    \quad Direct LR & 98.3 & 99.9 & 98.6 & 96.7 & 98.5 & 99.2 \\
    \quad Conformal LR & 98.7 & 99.9 & 98.0 & 97.0 & 98.0 & 97.5 \\
    \quad Conformal NRE-A & 98.7 & 99.9 & 98.0 & 97.1 & 97.9 & 97.8 \\
    \quad Conformal NRE-C & 98.5 & 99.8 & 97.9 & 97.1 & 98.0 & 97.8 \\
    \addlinespace
    \multicolumn{7}{l}{\textbf{\boldmath Small architecture}} \\
    \quad Conformal NRE-A & 99.7 & 100.0 & 99.8 & 100.0 & 100.0 & 100.0 \\
    \quad Conformal NRE-C & 98.8 & 100.0 & 98.8 & 96.5 & 97.9 & 97.9 \\
    \bottomrule
  \end{tabular}
\end{table}

\subsubsection{Per-ablation confidence sequences ($\thetatrue = 0$)}

\label{app:cs-theta0}

\begin{figure}[t]
  \centering
  \includegraphics[width=\linewidth]{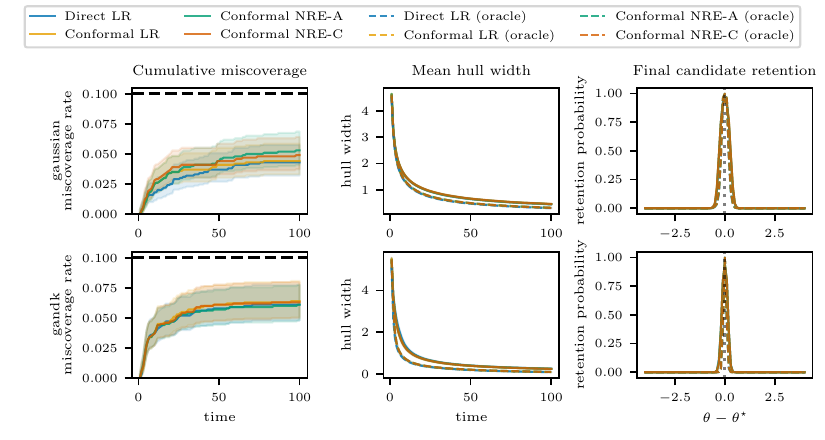}
  \caption{Sequential confidence sequences at $\thetatrue=0$ for the
  Gaussian (top) and $g$-and-$k$ (bottom) models under the baseline
  configuration (minimum-evidence alternative, $K=15$). Left: cumulative
  miscoverage (dashed $\alpha=0.1$). Center: mean hull width of nonempty
  sets. Right: final-time retention per candidate. Solid: feasible
  procedures; dashed: infeasible oracle using $\thetatrue$ as the
  alternative.}
  \label{fig:cs-baseline}
\end{figure}

\begin{figure}[t]
  \centering
  \includegraphics[width=\linewidth]{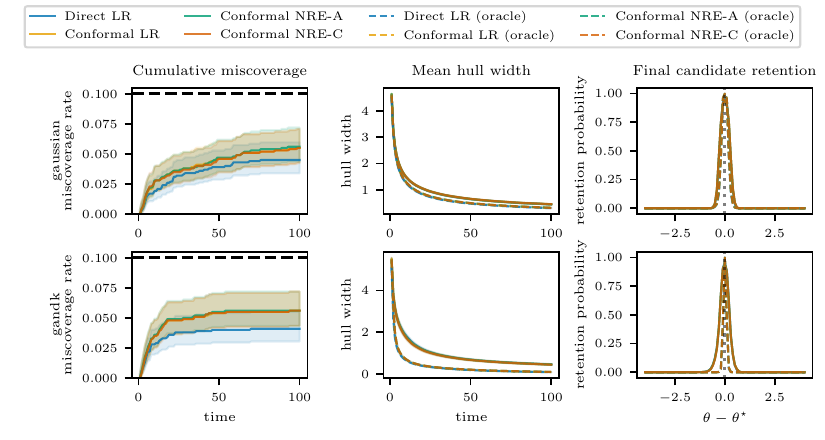}
  \caption{As \cref{fig:cs-baseline}, with the model-based plug-in
  (running mean / running median) in place of minimum evidence.}
  \label{fig:cs-plugin}
\end{figure}

\begin{figure}[t]
  \centering
  \includegraphics[width=\linewidth]{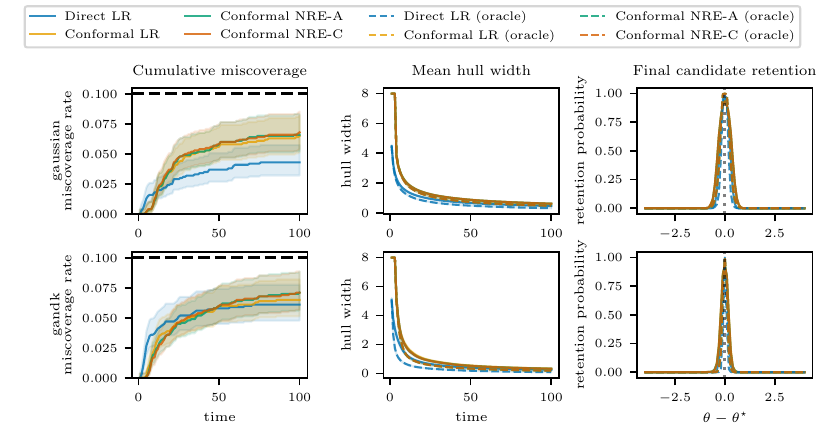}
  \caption{As \cref{fig:cs-baseline}, with $K=1$ null simulation per
  step.}
  \label{fig:cs-k1}
\end{figure}

\begin{figure}[t]
  \centering
  \includegraphics[width=\linewidth]{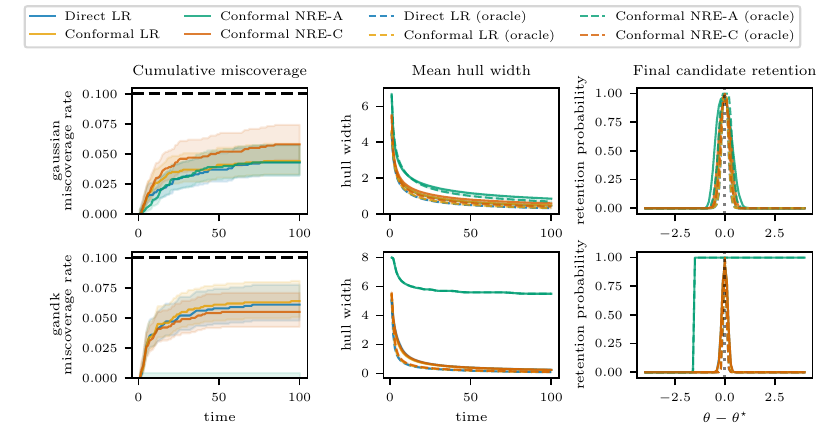}
  \caption{As \cref{fig:cs-baseline}, with a deliberately
  underparameterized score network (single hidden layer of two units).}
  \label{fig:cs-small}
\end{figure}

\Cref{tab:ablations} records each procedure at two time points; here we
show the full sequential trajectories at $\thetatrue=0$, for both models,
with the infeasible oracle variant of every estimator (which uses
$\thetatrue$ as the alternative) shown dashed. \Cref{fig:cs-baseline}
shows the baseline: the four feasible curves (exact-LR and neural) are
essentially superimposed in every panel, with miscoverage plateauing well
below $\alpha$, coincident hull widths, and final-time retention
concentrating similarly at $\thetatrue$, so the choice of score is invisible
on these models. The oracle variants never miscover and give narrower sets at
every time, so the extra width of the feasible sets comes from estimating the
alternative rather than from the score.

Replacing minimum evidence with the model-based running estimate
(\cref{fig:cs-plugin}) separates the two models: the Gaussian trajectories
are indistinguishable from the baseline, while on $g$-and-$k$ the feasible
width curves sit visibly above both the baseline and the oracle, with
final-time widths roughly twice the baseline (\cref{tab:ablations});
retention concentrates correspondingly less sharply. Miscoverage is
unaffected in both cases, as validity does not depend on the alternative.

Reducing the null budget to $K=1$ (\cref{fig:cs-k1}) changes only the
conformal sequences, which use the $K$ reference draws; the direct sequences
are unchanged. The conformal widths inflate (note the rescaled axis), the three
scores degrading together, and the miscoverage of the feasible conformal
sequences rises toward the nominal line. Validity is retained throughout.

Underparameterizing the score network (\cref{fig:cs-small}) degrades
NRE-A on both models while leaving NRE-C close to its baseline. The
degradation is mild on the Gaussian (NRE-A's sets widen but still
concentrate) and severe on $g$-and-$k$, where its width stays close to
the full grid and its retention never concentrates. Miscoverage
nonetheless stays below $\alpha$ for every curve: the underparameterized
score forfeits power, sharply so on $g$-and-$k$, without ever
compromising coverage. 
\subsection{SLCP}
\label{app:slcp}

\Cref{fig:slcp-all} adds conformal NRE-A and the oracle versions of the
conformal exact-LR, NRE-A and NRE-C sequences to \cref{fig:slcp-main}; the
set-up is described in \cref{app:experimental-details}.

Cumulative miscoverage stays below $\alpha$ for every procedure. The three
conformal sequences shrink at similar rates, NRE-A slightly more slowly, and
direct LR shrinks fastest. At $n=100$ they retain between $4.1$ and $8.2$
candidates on average, against about $4$ for the oracles: $\thetatrue$ and its
three mirror images, which give the same distribution of the data.
\Cref{tab:slcp-modes} shows which candidates remain. The exact-LR methods
retain or exclude the four equivalent values together, since their likelihood
ratios and shared reference draws yield identical evidence processes; with
$\thetatrue$ as the alternative, they therefore never exclude any of the four.
Neural scores need not preserve this symmetry: conformal NRE-C retains all
four in $93.4\%$ of repetitions, close to its coverage of $\thetatrue$
($94.6\%$), but its final set consists of exactly these four in only $28.2\%$,
with $2.56$ other candidates retained on average, and the neural oracle
variants exclude at least one of the three other values in $13$--$17\%$ of
repetitions while always retaining $\thetatrue$. The coverage guarantee
applies to each parameter value taken as the truth, not to all four jointly.

\begin{table}[t]
  \centering
  \small
  \caption{SLCP candidate retention at $n=100$ over $1{,}000$ repetitions,
  discarding the first two e-values (\cref{app:running-cs}). All four:
  percentage of repetitions retaining all four equivalent values
  $(0,0,\pm2.2,\pm2.2,0)$. Other candidates: mean number of retained
  candidates outside these four, over all repetitions. Exactly four:
  percentage of repetitions whose final set consists of precisely these
  four values.}
  \label{tab:slcp-modes}
  \begin{tabular}{lrrr}
    \toprule
    Method & All four (\%) & Other candidates & Exactly four (\%) \\
    \midrule
    Direct LR                  & 96.1  & 0.23 & 93.3 \\
    Conformal LR               & 93.2  & 1.78 & 53.5 \\
    Conformal NRE-A            & 92.9  & 4.38 & 20.8 \\
    Conformal NRE-C            & 93.4  & 2.56 & 28.2 \\
    Direct LR (oracle)         & 100.0 & 0.00 & 100.0 \\
    Conformal LR (oracle)      & 100.0 & 0.00 & 100.0 \\
    Conformal NRE-A (oracle)   & 83.5  & 0.28 & 67.0 \\
    Conformal NRE-C (oracle)   & 87.1  & 0.03 & 84.4 \\
    \bottomrule
  \end{tabular}
\end{table}

\begin{figure}[t]
  \centering
  \includegraphics[width=\linewidth]{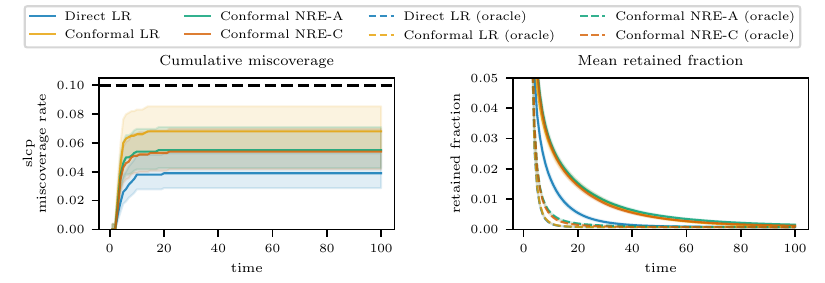}
  \caption{Sequential confidence sequences for the SLCP model, all procedures.
  Left: cumulative miscoverage (Wilson $95\%$ bands; dashed $\alpha=0.1$;
  oracles omitted). Right: mean fraction of candidates retained. Solid:
  feasible procedures; dashed: infeasible oracle using $\thetatrue$ as the
  alternative.}
  \label{fig:slcp-all}
\end{figure}

\section{APPLICATION: AUDITING MODEL ROUTING IN AN LLM API}
\label{app:routing}

\subsection{Setting and data}
\label{app:routing-data}

\paragraph{The audit problem.} An API that advertises a language model $P$ may
serve a different one without saying so \citep{Gao2025,Cai2025}, for example by
answering a fraction $\theta$ of the queries with a cheaper model $R$
\citep{Cai2025,Zhang2026IRIS}. An auditor who can only send
queries and read the replies wants a confidence sequence for $\theta$.

Our setup follows IRIS \citep{Zhang2026IRIS}, a recent black-box auditor for
this problem. It sends prompts that ask for random output, marks the replies
that look atypical for the advertised model, and detects routing and estimates
$\theta$ from how often such replies occur.
The approach hinges on the fact that language models are poor random number generators, and each model
is poor in its own reproducible way.
We use IRIS's prompts, its parser and its request settings, and in
\cref{app:routing-results} we compare with its fixed-time confidence interval for
$\theta$ and with a confidence-sequence version of its sequential test.

\paragraph{Probes.} We use two of IRIS's prompts as probes to detect routing
anomalies. The prompt \texttt{short\_digits\_8} asks ``Generate exactly 8
independent uniformly random digits from 0 to 9. Output only the compact digit
string, no separators.'', and \texttt{single\_number\_100} asks ``Choose one
uniformly random integer between 1 and 100 inclusive. Output only the
integer.'' IRIS's parser reduces each reply to a symbol $y$: the string of
digits for the first prompt, and the first integer in the reply for the second.

\paragraph{Data.} We take $P=\texttt{qwen3:8b}$ and $R=\texttt{qwen3:1.7b}$
\citep{Yang2025Qwen3}, served locally by Ollama\footnote{\url{https://ollama.com}, version 0.34.4.} with IRIS's request settings:
reasoning off, temperature~$1$, top-$p$~$1$, and IRIS's limit on reply length.
We collected $28{,}000$ replies per model and prompt. Each query is a separate
call with the same prompt and no conversation history, so the replies of one
model are i.i.d.\ draws from its output distribution given the prompt.

On both prompts the two models' replies barely overlap: each model's most
frequent reply is one that the other model almost never gives
(\cref{tab:routing-fingerprints}), so a single reply is strong evidence of which
model answered. On the number prompt the two models are separated by the
wording of the reply rather than by the number chosen: \texttt{qwen3:1.7b}
restates the task before answering, so the parser records $1$. The prompts also
differ in how many distinct replies occur in our collections: $1{,}008$ for the
digit prompt and $20$ for the number prompt.

\begin{table}[h]
\centering
\caption{The most frequent parsed reply of each model, and how often each model
gives it, out of $28{,}000$ calls per model and prompt.}
\label{tab:routing-fingerprints}
\small
\begin{tabular}{llcc}
\toprule
Prompt & Parsed reply & \texttt{qwen3:8b} ($P$) & \texttt{qwen3:1.7b} ($R$)\\
\midrule
8 random digits & \texttt{73951842} & $29\%$ & never\\
8 random digits & \texttt{0123456789}\footnotemark & never & $95\%$\\
a number from 1 to 100 & $42$ & $92\%$ & $9$ replies\\
a number from 1 to 100 & $1$\footnotemark & never & $99.9\%$\\
\bottomrule
\end{tabular}

\end{table}

\addtocounter{footnote}{-1}\footnotetext{Ten digits, although the prompt asks for eight.}
\stepcounter{footnote}\footnotetext{\texttt{qwen3:1.7b} opens its reply with ``Let's choose a random integer between 1 and 100 \ldots'', and the parser keeps the first integer in
the reply.}

\paragraph{Semi-synthetic design.} As in IRIS's routing experiments, we build every audited stream ourselves by simulating the routing while using real LLM responses.
Given $\thetatrue \in [0,1]$, each query is answered with a reply from $R$ or $P$ by sampling without replacement from a held-out set of the models' collected replies with probabilities $\thetatrue$ and $(1 - \thetatrue)$, respectively.
The replies are real outputs of the two models; only the routing is simulated.
Because $\thetatrue$ is known, this lets us measure the coverage of every method.

\subsection{Simulator}
\label{app:routing-model}

Let $p$ and $r$ be the pmf of the parsed reply $y$ under $P$ and $R$, respectively.
At routing fraction $\theta$, the simulator sends a query to $R$ with
probability $\theta$ and to $P$ otherwise, so its replies are i.i.d.\ with
pmf
\[
    f_\theta(y)=(1-\theta)\,p(y)+\theta\, r(y),
    \qquad
    \theta\in[0,1].
\]
In practical applications, the likelihood $f_\theta$ has no closed form: $p$ and $r$ are the models'
output pmfs over parsed replies, which are not exposed by APIs in the audit setting. The
auditor can only sample from them by querying the models.

For any $\theta_0,\theta_1\in[0,1]$, the likelihood ratio depends on
the reply only through $\rho(y)=r(y)/p(y)$:
\begin{equation}
\label{eq:routing-ratio}
    \frac{f_{\theta_1}(y)}{f_{\theta_0}(y)}
    =\frac{(1-\theta_1)+\theta_1\rho(y)}{(1-\theta_0)+\theta_0\rho(y)}.
\end{equation}
Hence an estimate of $\rho$, learned from replies of the two models, gives an estimated likelihood ratio for every pair of candidates $\theta_0, \theta_1$.
Our audited streams are built from held-out replies of the two models, with
routing that we apply ourselves.

\subsection{Computing the scores}
\label{app:routing-scores}

Of the $28{,}000$ replies collected per model and prompt, we set aside
$10{,}000$ of $P$'s and $10{,}000$ of $R$'s to estimate the scores and use
them for nothing else; we call these the score pools. The remaining replies are
used to build the audited streams and their references
(\cref{app:routing-experiments}). From the score pools we estimate the ratio
$\rho(y)=r(y)/p(y)$ in one of two ways.

\paragraph{Count ratio.} Since the number of unique responses given by the LLMs are quite low for this experiment, it is feasible to estimate $p(y)$ and $r(y)$ by the share of each
score pool's replies equal to $y$, with half a count added to every possible
reply. The latter smoothing keeps $\hat\rho$ finite and positive, and a reply seen in
neither pool gets $\hat\rho=1$, so it carries no evidence.

\paragraph{NRE.} A binary NRE is trained to
tell $P$'s replies from $R$'s. The network has two hidden layers of $64$ units.
Its inputs are features of the reply string: its length, the frequency of each
character, the first and last characters, how often consecutive characters
repeat or step by one, and indicators for the $64$ most frequent replies. It is
trained with a class-balanced logistic loss, so its output estimates
$\log\rho(y)$. An amortised NRE, as in the other experiments, is not needed
here, because $\theta$ enters only through \cref{eq:routing-ratio}.

\paragraph{Score.} With either estimate $\hat\rho$, plugging it into
\cref{eq:routing-ratio} gives an estimate of the likelihood ratio between any
two candidates. The score of candidate $\theta$ at step $i$ is
\[
    T_i^{(\theta)}(y)=\frac{\hat f_{\hat\theta}(y)}{\hat f_{\theta}(y)},
\]
where $\hat\theta=\hat\theta_{i-1}$ is the minimum-evidence
alternative: the candidate with the least accumulated evidence against it after
the first $i-1$ observations.

\paragraph{Defensive score.} We use $(1-\lambda)\,T_i^{(\theta)}(y)+\lambda$
with $\lambda=0.3$ in place of $T_i^{(\theta)}(y)$. Mixing every score with
the same constant brings the scores of different replies closer together, so
each step changes a candidate's wealth less. This matters early in the
sequence. The alternative $\hat\theta$ starts at $0$ and moves towards the
true $\thetatrue$ only as evidence accumulates, and while it is still far
away, every candidate between it and the truth loses wealth, and those that
should eventually be excluded must later win it back. The
mixed score is still positive and depends only on past observations, so
validity holds for any fixed $\lambda$; only efficiency depends
on it. We chose $\lambda$ on held-out development streams, by a
rule fixed before running them: with the count-ratio score, both prompts,
$\lambda\in\{0,0.05,0.1,0.2,0.3,0.5\}$ and
$\thetatrue\in\{0,0.05,0.1,0.2,0.3,0.5\}$, take the value with the smallest
mean hull width at $n=300$ among those whose cumulative miscoverage is at most
$\alpha$ up to its Wilson band and whose median detection time is at most one query
slower than with $\lambda=0$, at every $\thetatrue$. The chosen
$\lambda=0.3$ is then used unchanged for every score and prompt, and every
result we report comes from fresh test streams, drawn independently of the
development streams. On the test streams (\cref{tab:routing-lambda}), $\lambda=0.2$ and $0.3$ narrow the final sets by about $30\%$ relative to
$\lambda=0$ without slowing detection; $\lambda=0.5$ can slow detection when
routing is rare. Discarding the first few observations instead
narrows the sets only at large $\thetatrue$ and delays detection, and the
average of the $\lambda=0,0.2,0.5$ e-processes, which involves no tuning, comes
within a few percent of the chosen value.

\begin{table}[t]
  \centering
  \small
  \setlength{\tabcolsep}{3pt}
  \caption{Sensitivity of the routing confidence sequence (count-ratio score) to the defensive floor $\lambda$, on fresh test streams disjoint from those used to choose it ($200$ repetitions per $\thetatrue\in\{0,0.05,0.1,0.2,0.3,0.5\}$, $n=300$, $\alpha=0.1$). Miscov.\ is the largest time-uniform miscoverage over $\thetatrue$; width is the mean hull width at $n=300$ over nonempty sets; detect is the median number of queries until $\theta=0$ leaves the set. ``Skip $k$'' discards the first $k$ observations, so its guarantee holds from step $k+1$; the mixture averages the $\lambda=0,0.2,0.5$ e-processes and involves no tuning. Every miscoverage is at most $\alpha$ up to its Wilson $95\%$ band; $\lambda=0.3$ is the value used elsewhere, chosen on separate development streams.}
  \label{tab:routing-lambda}
  \begin{tabular}{l c ccccc ccccc}
    \toprule
    & & \multicolumn{5}{c}{width at $\thetatrue=$} & \multicolumn{5}{c}{detect at $\thetatrue=$} \\
    \cmidrule(lr){3-7}\cmidrule(lr){8-12}
    Variant & miscov. & $0.05$ & $0.1$ & $0.2$ & $0.3$ & $0.5$ & $0.05$ & $0.1$ & $0.2$ & $0.3$ & $0.5$ \\
    \midrule
    \multicolumn{12}{l}{\textbf{\emph{8 random digits}}} \\
    \quad $\lambda=0$ & 0.105 & 0.087 & 0.129 & 0.197 & 0.230 & 0.278 & 32 & 15 & 8 & 6 & 3 \\
    \quad $\lambda=0$, skip 2 & 0.105 & 0.083 & 0.116 & 0.169 & 0.179 & 0.188 & 33 & 15 & 8 & 6 & 4 \\
    \quad $\lambda=0$, skip 5 & 0.085 & 0.076 & 0.104 & 0.142 & 0.149 & 0.163 & 32 & 17 & 10 & 8 & 7 \\
    \quad $\lambda=0.2$ & 0.080 & 0.061 & 0.091 & 0.138 & 0.163 & 0.194 & 30 & 15 & 8 & 5 & 3 \\
    \quad $\lambda=0.3$ & 0.080 & 0.059 & 0.089 & 0.133 & 0.156 & 0.186 & 30 & 15 & 8 & 5 & 3 \\
    \quad $\lambda=0.5$ & 0.055 & 0.062 & 0.092 & 0.135 & 0.160 & 0.184 & 42 & 18 & 10 & 6 & 3 \\
    \quad mixture of $\lambda\in\{0,0.2,0.5\}$ & 0.080 & 0.064 & 0.093 & 0.140 & 0.164 & 0.193 & 32 & 15 & 8 & 6 & 3 \\
    \addlinespace
    \multicolumn{12}{l}{\textbf{\emph{a number from 1 to 100}}} \\
    \quad $\lambda=0$ & 0.100 & 0.087 & 0.134 & 0.205 & 0.243 & 0.289 & 30 & 15 & 8 & 6 & 3 \\
    \quad $\lambda=0$, skip 2 & 0.090 & 0.083 & 0.123 & 0.174 & 0.185 & 0.193 & 30 & 15 & 8 & 6 & 4 \\
    \quad $\lambda=0$, skip 5 & 0.085 & 0.077 & 0.107 & 0.142 & 0.150 & 0.163 & 30 & 16 & 10 & 8 & 7 \\
    \quad $\lambda=0.2$ & 0.095 & 0.059 & 0.091 & 0.140 & 0.163 & 0.197 & 30 & 15 & 8 & 6 & 3 \\
    \quad $\lambda=0.3$ & 0.080 & 0.058 & 0.089 & 0.134 & 0.160 & 0.187 & 30 & 15 & 8 & 6 & 3 \\
    \quad $\lambda=0.5$ & 0.050 & 0.060 & 0.090 & 0.134 & 0.161 & 0.184 & 30 & 15 & 8 & 6 & 3 \\
    \quad mixture of $\lambda\in\{0,0.2,0.5\}$ & 0.080 & 0.061 & 0.094 & 0.141 & 0.165 & 0.195 & 30 & 15 & 8 & 6 & 3 \\
    \bottomrule
  \end{tabular}
\end{table}

\subsection{Experiments}
\label{app:routing-experiments}

The $28{,}000$ replies of each model per prompt are split into disjoint pools:
$10{,}000$ form the score pool (\cref{app:routing-scores}), $3{,}000$ an
observation pool from which the streams are built, and the remaining $15{,}000$
five reference pools of $3{,}000$ from which the references are drawn.

At step $i$, the confidence sequence takes $K$ fresh replies
$R_{i1},\ldots,R_{iK}$ of $R$ and $P_{i1},\ldots,P_{iK}$ of $P$ from a reference
pool, together with uniform draws $U_{i1},\ldots,U_{iK}$, and sets, for every
candidate $\theta$,
\begin{equation}
\label{eq:routing-references}
    \widetilde Y_{ij}^{(\theta)}=
    \begin{cases}R_{ij}&U_{ij}<\theta,\\ P_{ij}&\text{otherwise.}\end{cases}
\end{equation}
Since $\theta$ only decides which model answers, these are draws from
$f_\theta$ for every candidate at once. All candidates share the same $2K$
replies and $K$ uniform draws (common random numbers), so the number of model
calls does not grow with the number of candidates.

A stream of $n=300$ steps uses $K=10$ replies of each model per step, so exactly
one reference pool, and no reply is used twice within a repetition.
As is usual in
semi-synthetic designs, replies are reused across repetitions: each reference
pool serves $40$ of the $200$ repetitions, reshuffled each time, and all streams
draw from the same observation pool.

For each prompt and each true routing fraction
$\thetatrue\in\{0,0.1,0.3\}$, we run $200$ repetitions. We compare three confidence sequences:
\begin{itemize}
\item \emph{Conformal count ratio.} Count-ratio score, minimum-evidence alternative.
\item \emph{Conformal NRE.} NRE score, minimum-evidence alternative.
\item \emph{Conformal count ratio (oracle).} Count-ratio score, with the true $\thetatrue$ as the alternative.
\end{itemize}

For each method we report the following quantities:
\begin{itemize}
\item \emph{Cumulative miscoverage.} The fraction of repetitions in which the
confidence sequence has excluded the true $\thetatrue$ at some step
$i\le n$. Figures show it with Wilson $95\%$ intervals.
\item \emph{Mean hull width.} The largest retained candidate minus the smallest,
averaged over the repetitions with a nonempty set.
\item \emph{Detection time.} The number of steps until $\theta=0$ leaves the set;
we report its median, because it is infinite in repetitions where
$\theta=0$ is never excluded.
\item \emph{Confidence sets over time.} For every $\thetatrue$ we also show the
sets themselves. At each step $n$ we take the lower and upper ends of each
repetition's accepted candidate set, over the repetitions with a nonempty set, and plot the mean
of the lower ends and the mean of the upper ends, with $95\%$ confidence
intervals for the means.
\end{itemize}

\subsection{Baselines}
\label{app:routing-baselines}

All baselines are computed on the same streams as the conformal confidence
sequences, over the same $101$ candidates.

\paragraph{Fixed-$n$ baseline: IRIS interval.} IRIS \citep{Zhang2026IRIS}
reduces each reply to a single bit, a \emph{tell}, which marks the reply as
atypical for the advertised model $P$. In IRIS, a reply $y$ is a tell when
$-\log\hat{\mathbb{P}}(P\mid y)$, computed from a classifier that predicts
which model produced $y$, exceeds a high quantile of the same score on $P$'s
replies. We apply the same rule to our count-ratio estimate: a reply is a tell
if $\log\hat\rho(y)$ exceeds a threshold chosen so that at most $5\%$ of $P$'s
replies in the score pool are tells. Let $a_1$ and $q$ be the shares of tells
among $P$'s and $R$'s replies in the score pools (with half a count added).
As in IRIS, the threshold, $a_1$ and $q$ are fixed before the audit starts;
during the audit, each reply is only checked against the threshold.
Under routing fraction $\theta$, a reply is a tell with probability
$(1-\theta)a_1+\theta q$. If $f_n$ is the share of tells among the first $n$
replies, IRIS's estimate is $\hat\theta_n=(f_n-a_1)/(q-a_1)$, clipped to
$[0,1]$, and its delta-method interval accounts for the randomness in $f_n$
and for the estimation of $a_1$ and $q$ on the score pools. The interval is
built for a single $n$ fixed in advance, and even there its coverage is
nominal only asymptotically, because the delta method relies on a normal
approximation. With very few replies the interval is degenerate. We therefore
count it as covering $\thetatrue$ for $n<10$, which favours it, and in the
figures we plot its lower and upper ends only from $n=2$.

\paragraph{Anytime-valid baseline: Betting IRIS.} IRIS \citep{Zhang2026IRIS}
gives a confidence interval for $\theta$ only at a fixed $n$, and uses a
betting test on the tells only to decide whether the endpoint is honest. We
extend this test to a confidence sequence for $\theta$ \citep{WaudbySmith2024}.
Let $X_i$ be the tell indicator of reply $i$, defined as in the IRIS interval
above. Under routing fraction $\theta$, the tell rate is
$p_\theta=(1-\theta)a_1+\theta q$, so
\[
    Z_i^{(\theta)} = X_i - p_\theta \in [-1, 1]
\]
has mean zero when $\theta$ is the true routing fraction. For each candidate
$\theta$ we track the wealth
\[
    W_n^{(\theta)} = \tfrac12\prod_{i=1}^n\bigl(1+\lambda_i^{+}Z_i^{(\theta)}\bigr)
    + \tfrac12\prod_{i=1}^n\bigl(1-\lambda_i^{-}Z_i^{(\theta)}\bigr)
\]
of two bettors, one on a positive and one on a negative mean of
$Z_i^{(\theta)}$, who start with half the wealth each. The bets
$\lambda_i^{\pm}\in[0,0.9]$ are chosen from $Z_1^{(\theta)},\ldots,
Z_{i-1}^{(\theta)}$ by the aGRAPA rule of \citet{WaudbySmith2024}, and $\theta$
is excluded from the set once $W_n^{(\theta)}$ reaches $1/\alpha$. The
sequence is only approximately valid, because $a_1$ and $q$ are estimated on
the score pools, so $p_\theta$ is not exactly the tell rate under $\theta$.

\subsection{Results}
\label{app:routing-results}

\paragraph{Validity.} The conformal sequences and Betting IRIS keep cumulative
miscoverage below $\alpha=0.1$ throughout, at every $\thetatrue$
(\cref{tab:routing}; \cref{fig:routing-theta0,fig:routing-theta01,fig:routing-theta03}, panels a and c).
The IRIS interval, checked at every step, does not: its cumulative miscoverage
exceeds $\alpha$ at $\thetatrue=0.1$ and $0.3$ on both prompts, and at
$\thetatrue=0$ on the digits prompt.

\begin{table}[t]
  \centering
  \small
  \setlength{\tabcolsep}{3pt}
  \caption{LLM routing audit on semi-synthetic test streams, drawn independently of the development streams used to choose $\lambda=0.3$: $200$ repetitions for each $\thetatrue\in\{0,0.1,0.3\}$, $n=300$ steps, $\alpha=0.1$. Miscoverage is the cumulative miscoverage at $n=300$, taking the largest of the three values of $\thetatrue$; for the IRIS interval it counts misses from $n=10$ onwards. Width is the mean hull width at $n=300$ over repetitions with a nonempty set, at $\thetatrue=0$, $0.1$ and $0.3$. Detect is the median number of steps until $\theta=0$ leaves the set, at $\thetatrue=0.1$ and $0.3$.}
  \label{tab:routing}
  \begin{tabular}{l cccccc@{\hspace{8pt}}cccccc}
    \toprule
    & \multicolumn{6}{c}{8 random digits} & \multicolumn{6}{c}{a number from 1 to 100} \\
    \cmidrule(lr){2-7}\cmidrule(lr){8-13}
    & & \multicolumn{3}{c}{width} & \multicolumn{2}{c}{detect} & & \multicolumn{3}{c}{width} & \multicolumn{2}{c}{detect} \\
    Method & miscov. & $0$ & $0.1$ & $0.3$ & $0.1$ & $0.3$ & miscov. & $0$ & $0.1$ & $0.3$ & $0.1$ & $0.3$ \\
    \midrule
    \multicolumn{13}{l}{\textbf{\emph{Conformal confidence sequences}}} \\
    \quad Conformal count ratio & 0.08 & 0.01 & 0.09 & 0.16 & 16 & 5 & 0.09 & 0.01 & 0.09 & 0.16 & 16 & 5 \\
    \quad Conformal NRE & 0.08 & 0.01 & 0.09 & 0.16 & 17 & 5 & 0.09 & 0.01 & 0.09 & 0.16 & 17 & 5 \\
    \quad Conformal count ratio (oracle) & 0.00 & 0.01 & 0.07 & 0.11 & 7 & 2 & 0.00 & 0.01 & 0.07 & 0.11 & 7 & 2 \\
    \addlinespace
    \multicolumn{13}{l}{\textbf{\emph{Anytime-valid baseline}}} \\
    \quad Betting IRIS & 0.06 & 0.05 & 0.11 & 0.14 & 73 & 21 & 0.07 & 0.03 & 0.10 & 0.13 & 81 & 23 \\
    \addlinespace
    \multicolumn{13}{l}{\textbf{\emph{Fixed-$n$ baseline}}} \\
    \quad IRIS interval & 0.70 & 0.03 & 0.06 & 0.08 & 21 & 4 & 0.66 & 0.02 & 0.06 & 0.08 & 24 & 4 \\
    \bottomrule
  \end{tabular}
\end{table}

\begin{figure*}[t]
  \centering
  \includegraphics[width=\textwidth]{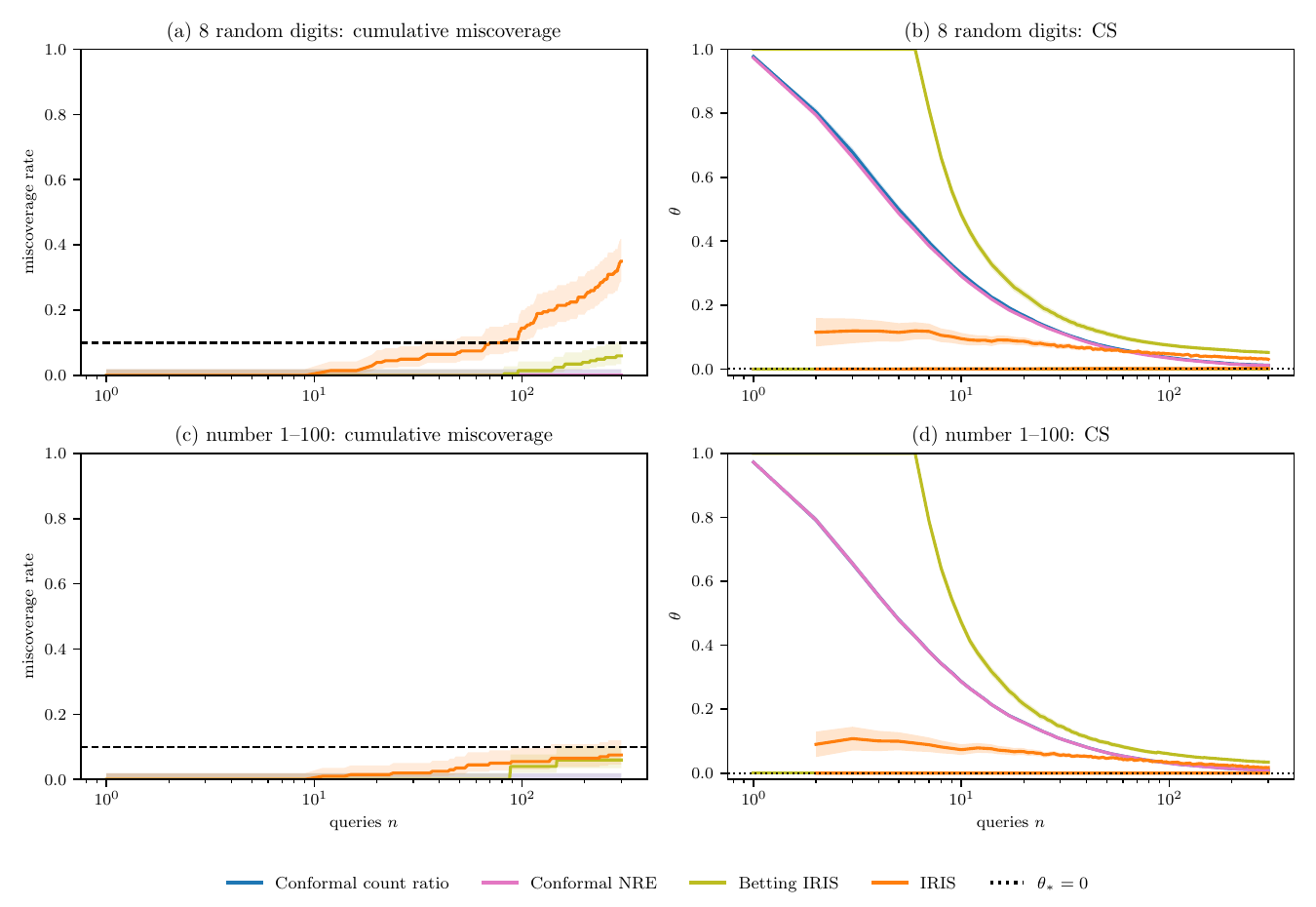}
  \caption{LLM routing audit at $\thetatrue=0$ (an honest API), one row per prompt.
(a, c) Cumulative miscoverage (Wilson $95\%$ bands; dashed $\alpha=0.1$); the
fixed-time IRIS interval is counted as covering $\thetatrue$ for $n<10$. (b, d) Confidence
sequences: mean lower and upper ends over the $200$ repetitions with a
nonempty set, with $95\%$ intervals for the means; IRIS is shown from $n=2$.}
  \label{fig:routing-theta0}
\end{figure*}

\begin{figure*}[t]
  \centering
  \includegraphics[width=\textwidth]{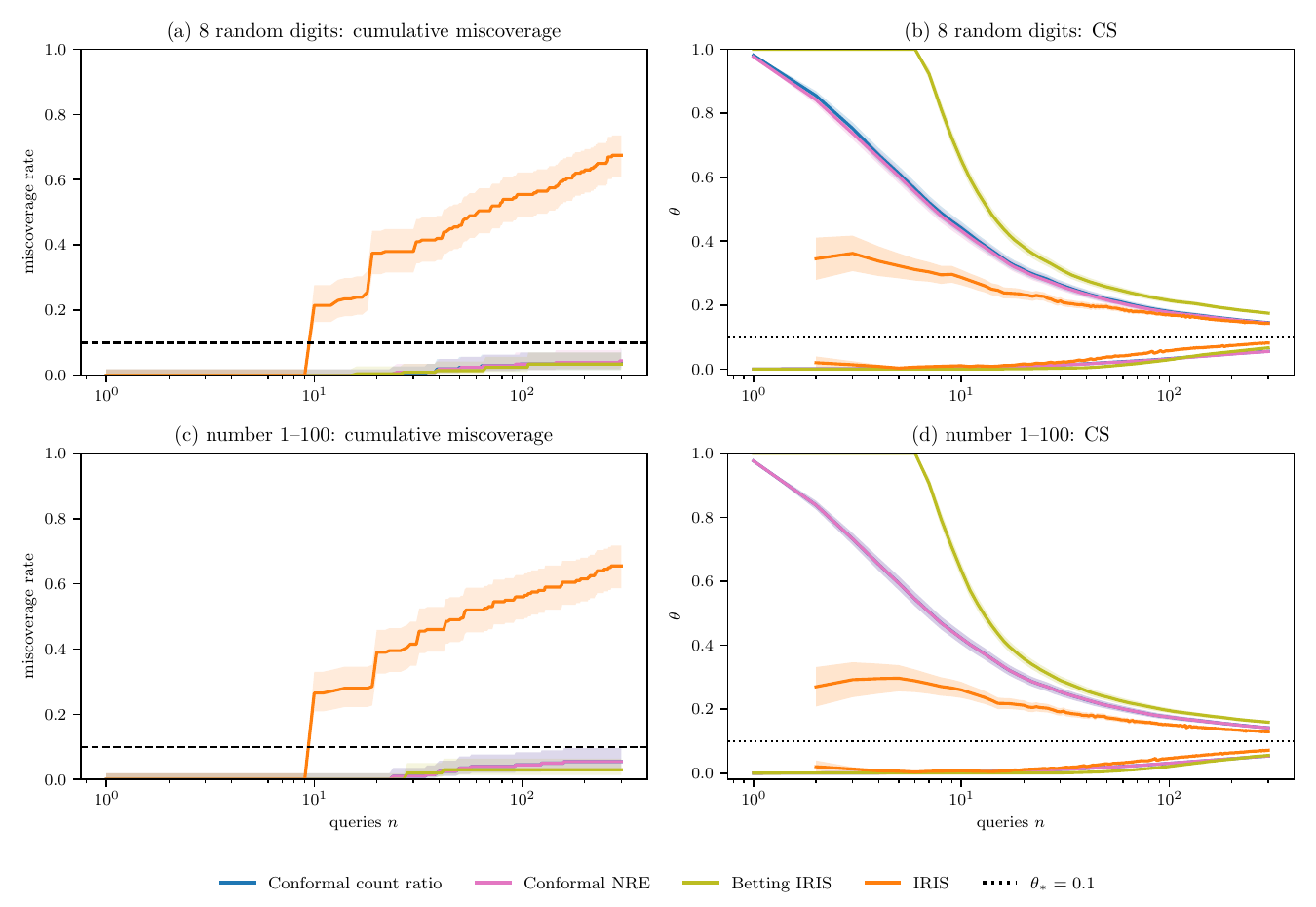}
  \caption{LLM routing audit at $\thetatrue=0.1$, one row per prompt.
(a, c) Cumulative miscoverage (Wilson $95\%$ bands; dashed $\alpha=0.1$); the
fixed-time IRIS interval is counted as covering $\thetatrue$ for $n<10$. (b, d) Confidence
sequences: mean lower and upper ends over the $200$ repetitions with a
nonempty set, with $95\%$ intervals for the means; IRIS is shown from $n=2$.}
  \label{fig:routing-theta01}
\end{figure*}

\begin{figure*}[t]
  \centering
  \includegraphics[width=\textwidth]{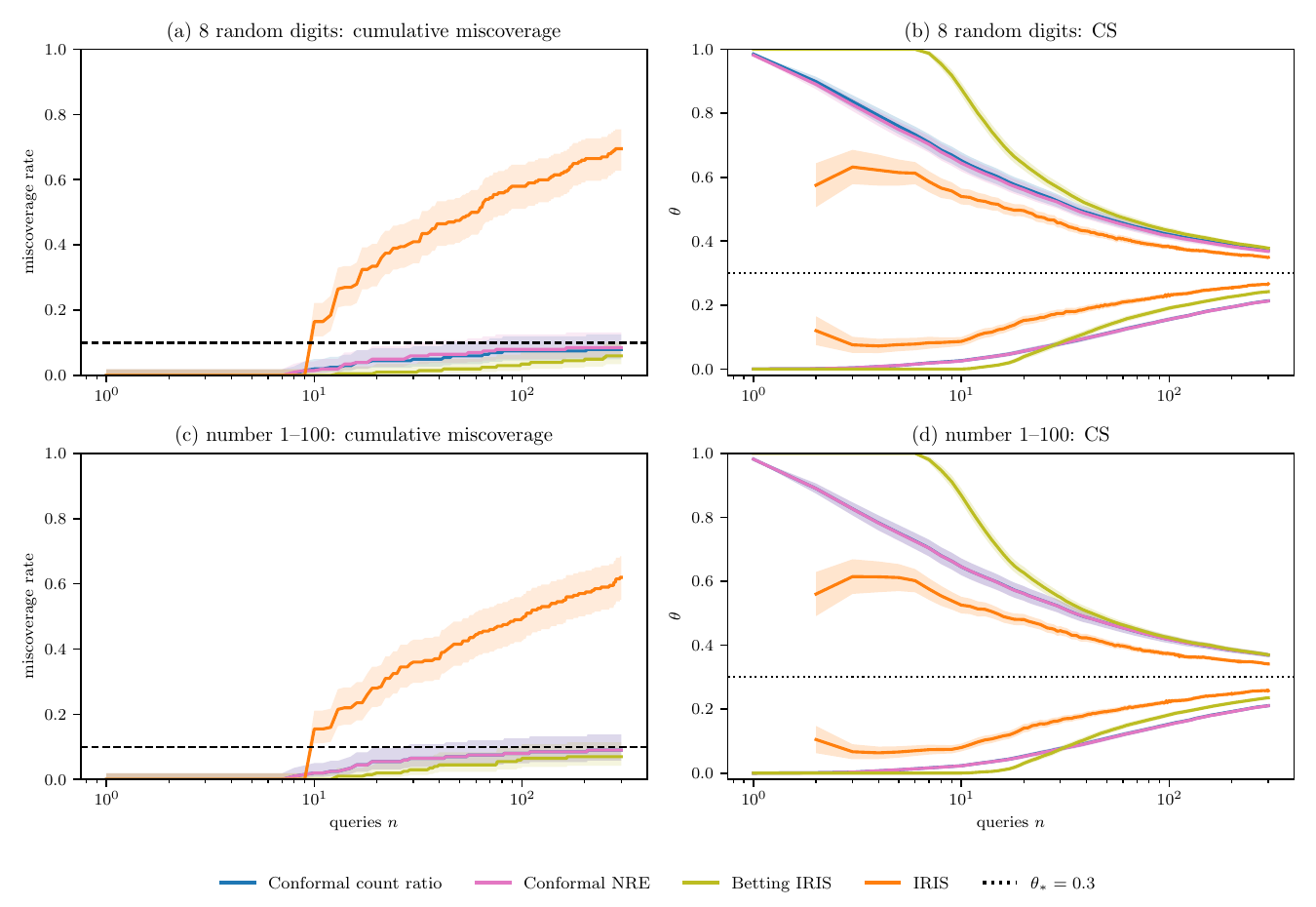}
  \caption{LLM routing audit at $\thetatrue=0.3$, one row per prompt.
(a, c) Cumulative miscoverage (Wilson $95\%$ bands; dashed $\alpha=0.1$); the
fixed-time IRIS interval is counted as covering $\thetatrue$ for $n<10$. (b, d) Confidence
sequences: mean lower and upper ends over the $200$ repetitions with a
nonempty set, with $95\%$ intervals for the means; IRIS is shown from $n=2$.}
  \label{fig:routing-theta03}
\end{figure*}

\paragraph{Width.} The conformal sequences are narrower than Betting IRIS
throughout at $\thetatrue=0$ and $0.1$, and early on at $\thetatrue=0.3$,
where Betting IRIS ends slightly narrower, with the NRE matching the count
ratio (\cref{tab:routing}; \cref{fig:routing-theta0,fig:routing-theta01,fig:routing-theta03}, panels b and d).
The sets are nonempty in almost every repetition (\cref{tab:routing-nonempty}).
The conformal sequence with the oracle alternative $\thetatrue$ achieves the
smallest width among the confidence sequences, suggesting that further gains
may be obtained by tuning $\hat\theta$.

\begin{table}[t]
  \centering
  \small
  \setlength{\tabcolsep}{4pt}
  \caption{Percentage of nonempty confidence sets at $n=300$, over $200$ repetitions. Mean hull widths in \cref{tab:routing} are averaged over these repetitions only.}
  \label{tab:routing-nonempty}
  \begin{tabular}{l ccc ccc}
    \toprule
    & \multicolumn{3}{c}{8 random digits} & \multicolumn{3}{c}{a number from 1 to 100} \\
    \cmidrule(lr){2-4}\cmidrule(lr){5-7}
    Method & $\thetatrue=0$ & $\thetatrue=0.1$ & $\thetatrue=0.3$ & $\thetatrue=0$ & $\thetatrue=0.1$ & $\thetatrue=0.3$ \\
    \midrule
    \multicolumn{7}{l}{\textbf{\emph{Conformal confidence sequences}}} \\
    \quad Conformal count ratio & 100.0 & 100.0 & 100.0 & 100.0 & 100.0 & 100.0 \\
    \quad Conformal NRE & 100.0 & 100.0 & 100.0 & 100.0 & 100.0 & 100.0 \\
    \quad Conformal count ratio (oracle) & 100.0 & 100.0 & 100.0 & 100.0 & 100.0 & 100.0 \\
    \addlinespace
    \multicolumn{7}{l}{\textbf{\emph{Anytime-valid baseline}}} \\
    \quad Betting IRIS & 99.0 & 100.0 & 100.0 & 94.0 & 100.0 & 100.0 \\
    \addlinespace
    \multicolumn{7}{l}{\textbf{\emph{Fixed-$n$ baseline}}} \\
    \quad IRIS interval & 100.0 & 100.0 & 100.0 & 100.0 & 100.0 & 100.0 \\
    \bottomrule
  \end{tabular}
\end{table}

\paragraph{Detection.} The conformal sequences with the count ratio and the
NRE rule out honest service, $\theta=0$, four to five times sooner than
Betting IRIS, at both $\thetatrue=0.1$ and $0.3$ (\cref{tab:routing}). With
the oracle alternative, the count-ratio sequence is more than twice as fast
again, so here too a
better choice of $\hat\theta$ could help.

\subsection{Limitations}
\label{app:routing-limitations}

The simulator assumes that every reply comes from $P$ or $R$, run with the
settings used to collect our pools, and that each query is routed
independently. If the API departs from this, for example by also using a third
model, our scores have little power against it: they are likelihood ratios
between members of the model $\{f_\theta\}$, so they only compare values of
$\theta$. Detecting such departures would need a score that targets them.

\section{APPLICATION: THE LENS MASS FUNCTION FROM MICROLENSING TIMESCALES}
\label{app:microlensing}

\subsection{Setting and data}
\label{app:microlensing-data}

\paragraph{Scientific target.} We infer the mass function of stars and brown
dwarfs toward the centre of the Milky Way, that is, how common objects of each
mass are. It is modelled as a power law $\xi(M)\propto M^{\alpha}$ whose slope
changes at a few fixed masses. The parameter of interest is
$\theta=(\theta_1,\theta_2)$: the slope $\theta_1$ for low-mass stars and the slope $\theta_2$ for brown dwarfs. 
These objects are faint, and brown dwarfs emit almost
no light, so they are hard to count directly. Microlensing detects them through
their gravity instead, whether or not they shine.

\paragraph{Microlensing events.} When such an object (the \emph{lens}) passes
almost exactly in front of a more distant star (the \emph{source}), its gravity
bends the source's light, and the source brightens and then fades back over a
few days to a few months. This brightening is a microlensing event
\citep{Paczynski1986}. OGLE (the Optical Gravitational Lensing Experiment), a
long-running survey with a dedicated telescope in Chile, detects such events by
repeatedly photographing the same patches of sky toward the centre of the Milky
Way and recording the brightness of millions of stars.

The length of an event, its Einstein timescale $t_E$, is longer for heavier
lenses. Because $t_E$ also depends on the lens's and source's distances and
speeds, one event says little about the lens's mass. The distribution of $t_E$
over many events, however, depends on how common lenses of each mass are, and
hence on $\theta$. \Cref{app:microlensing-simulator} gives this distribution for each
$\theta$, given the event's field and measurement error.

\paragraph{Data.} We use the catalogue of \citet{Mroz2019}: $5{,}790$ events
detected in the low-cadence fields of the fourth phase of OGLE (OGLE-IV) over
eight yearly seasons (2010--2017). The survey divides the sky into fixed patches called \emph{fields}.
For each event the catalogue gives its field, its estimated timescale
$\hat t_{E,i}$, and error bars $e_i^+$ and $e_i^-$ on that estimate. For each
field $f$ it also gives the detection efficiency $\varepsilon_f(t_E)$, the
probability that an event of timescale $t_E$ in that field is detected; short
events are easily missed between photographs, so $\varepsilon_f$ is small at
short $t_E$. We keep the $3{,}514$ events in the $70$ fields with Galactic
coordinates $|\ell|\le10^\circ$ and $3^\circ\le|b|\le10^\circ$, roughly the
region covered by the simulator.

\paragraph{Observations and covariates.} Each event $i$ gives an observation
$Y_i=\log_{10}\hat t_{E,i}$ and covariates $X_i=(\ell_{f_i},b_{f_i},\sigma_i)$:
the coordinates of the centre of its field and its measurement error on the log
scale, $\sigma_i=\tfrac12[\log(1+e_i^+/\bar t_{E,i})-\log(1-e_i^-/\bar t_{E,i})]$,
where $e_i^\pm\ge0$ are the catalogue's error bars about its median timescale
$\bar t_{E,i}$. We model $Y_i$ given $X_i$: for a given field, the simulator
(\cref{app:microlensing-simulator}) produces the event's true timescale $t$, and
the measured timescale is $\hat t_{E,i}=t\,e^{\sigma_i\eta}$ with
$\eta\sim\mathcal N(0,1)$. This symmetric log-normal error, with $\sigma_i$
treated as known, approximates the catalogue's asymmetric error bars. The references for event $i$ are
therefore simulated in the same field and with the same error
(\cref{app:microlensing-simulator}).

\paragraph{Independence.} Microlensing events occur independently of one
another, at random times and places, and each involves a different pair of lens
and source. The survey detects each event with a probability that depends only
on its field and timescale. Given their number, the detected events are
therefore independent draws from one distribution, so the pairs $(X_i,Y_i)$ are
i.i.d.; in particular, given the covariates, the timescales are independent and
each $Y_i$ depends only on its own $X_i$, which is what the conditional
construction requires (\cref{cor:conditional-validity}). Each repetition of the
confidence sequence processes the events in a fresh random order.

\paragraph{Held-out and analysed events.} We split the events by season. The
$1{,}716$ events from even seasons (2010, 2012, 2014 and 2016) are held out and
used only to check the model (\cref{app:microlensing-fit-check}) and to choose the
truth of the semi-synthetic runs. The $1{,}798$ events from odd seasons (2011,
2013, 2015 and 2017) are the data on which the confidence sequence is run,
and on which the fitted model is tested.

\subsection{Simulator}
\label{app:microlensing-simulator}

\paragraph{Galactic model.} We use the Galactic model of \citet{Koshimoto2021},
fitted to Gaia, VVV, BRAVA, and OGLE data, as implemented in \texttt{genulens} \citep{KoshimotoRanc2021genulens}. 
It describes the Galaxy as different groups of
stars, or \emph{components}, each with its own spatial density, velocities, and
age; lenses and sources are drawn from these components.

\paragraph{Parameter.} The parameter $\theta=(\theta_1,\theta_2)$ enters 
the simulator through the initial mass
function $\xi_\theta(M)$, the number of objects per unit mass with mass $M$ at
birth. It is a broken power law. The mass range from $0.001$ to $120\,M_\odot$
is divided at $0.01$, $0.08$, and $0.86\,M_\odot$ into four intervals, 
and on the $k$-th interval $\xi_\theta(M)\propto M^{\alpha_k}$, where
$\alpha_2=\theta_1$ and $\alpha_3=\theta_2$ are the free slopes and $\alpha_1$
and $\alpha_4$ are fixed; on
a log--log plot, $\xi_\theta$ is a straight line within each interval, and the
pieces join at the break masses. From heaviest to lightest, the intervals and
their slopes are:
\begin{center}
\begin{tabular}{lll}
\toprule
Mass range ($M_\odot$) & Objects & Slope \\
\midrule
$0.86$--$120$ & heavier stars & $\alpha_1=-2.32$ (fixed) \\
$0.08$--$0.86$ & low-mass stars & $\theta_1$ \\
$0.01$--$0.08$ & brown dwarfs & $\theta_2$ \\
$0.001$--$0.01$ & planetary-mass objects & $\alpha_4=-0.18$ (fixed) \\
\bottomrule
\end{tabular}
\end{center}
The parameter $\theta$ is restricted to the box $[-2,0]\times[-2,10]$. The break masses and the fixed slopes are the defaults of
\texttt{genulens}, from the fit of \citet{Koshimoto2021}. As $\theta_2$ grows,
brown dwarfs disappear and the model converges to a limit with no lenses below
$0.08\,M_\odot$. Before running the confidence sequences we fixed the upper end
of $\theta_2$ as the value at which the distribution of detected timescales is
within about $1\%$ total variation of this limit at every $\theta_1$ in the box, which gives $10$; the box therefore
contains every model that the data could distinguish from the limit.

\paragraph{Generative model.} For a field $f$ and parameter $\theta$,
\texttt{genulens} (version 2.0.0a4) generates a detected event in five steps:
\begin{enumerate}
    \item a source: a Galactic component $c_S$ (an age bin of the thin disk, the
    thick disk, the bar, or another component of the model) and a distance
    $D_S$, with probability proportional to $s_{c_S}(\theta)\,\sigma_{c_S}(D_S)$;
    \item a lens: a component $c_L$ and a distance $D_L<D_S$, with probability
    proportional to $n_{c_L}(\theta)\,\rho_{c_L}(D_L)$;
    \item the lens's initial mass $M\sim\xi_\theta$;
    \item given these, with no further dependence on $\theta$: the lens's
    evolution over its component's age (into a white dwarf, neutron star, or
    black hole, or not at all), velocities, and the timescale $t$;
    \item an event-rate weight $\Gamma(z)$ and detection with probability
    $\varepsilon_f(t)$.
\end{enumerate}
Here $\sigma_c$ and $\rho_c$ are fixed spatial profiles of the fitted Galactic
model. The factors $s_c(\theta)$ and $n_c(\theta)$ turn each component's fitted
\emph{mass} density into the \emph{number} of source stars and of lenses; they
depend on $\theta$ because the mean mass of a component's objects, and its split
into stars and remnants, are integrals over $\xi_\theta$ at the component's
age. The parameter therefore enters in three places, not only through the lens
mass. Writing $z$ for all latent variables, a detected event has density
\[
    p_\theta(z\mid f)=\frac{1}{Z_\theta}\,
    s_{c_S}(\theta)\,\sigma_{c_S}(D_S)\,
    n_{c_L}(\theta)\,\rho_{c_L}(D_L)\,
    \xi_\theta(M)\,
    k(z_{\rm rest}\mid c_S,c_L,D_S,D_L,M)\,
    \Gamma(z)\,\varepsilon_f(t),
\]
with an unknown normalizing constant $Z_\theta$ and a $\theta$-free kernel $k$
for step 4. We write $q_\theta(\cdot\mid f)$ for the induced distribution of
$t$ \citep{Koshimoto2021}; the measured timescale adds the error of
\cref{app:microlensing-data}. The density of $q_\theta$ integrates over
distances, masses, and velocities and has no closed form, so the model can be
simulated but its likelihood is intractable.

\subsection{Exact simulation by reweighting}
\label{app:microlensing-reweighting}

\paragraph{Why reweight.} \texttt{genulens} costs a few microseconds per event
but about $0.2$\,s per call, because each call rebuilds the Galaxy for the given
$\theta$. The confidence sequence needs fresh references at every candidate and
every step, far too many calls. We therefore simulate, once per field, a
\emph{bank} of detected events at a proposal $\theta_{\rm prop}=(-1.2,-1.6)$,
and turn bank entries into exact draws at any $\theta$ by rejection sampling.

\paragraph{Banks.} \texttt{genulens} does not draw events in proportion to how
often they occur: it returns raw events $z$ with importance weights $\Gamma(z)$.
We keep each raw event with probability $\min\{\Gamma(z)/c,1\}$, $c=50$, and
then with probability $\varepsilon_f(t)$. If $\Gamma\le c$, the kept events are
i.i.d.\ draws from $p_{\theta_{\rm prop}}(\cdot\mid f)$; otherwise their total
variation distance from it is at most $\mathbb E[(\Gamma-c)_+]/\mathbb E[\Gamma]$.
The weights have mean $0.71$, and none of the $1.4\times10^{10}$ raw events
behind our banks exceeds $c$, so the kept events are
exact draws from the model with weights truncated at $c$, and the estimate of
this bound, computed for every bank, is zero. An
entry stores $t$, the lens's current mass and remnant type, and the components
$c_L$ and $c_S$.

\paragraph{Density ratio.} Every factor of $p_\theta$ that does not involve
$\theta$ cancels:
\[
    \frac{p_\theta(z\mid f)}{p_{\theta_{\rm prop}}(z\mid f)}
    =\frac{Z_{\theta_{\rm prop}}}{Z_\theta}\,w_\theta(z),
    \qquad
    w_\theta(z)=
    \frac{\xi_\theta(M)}{\xi_{\theta_{\rm prop}}(M)}\cdot
    \frac{n_{c_L}(\theta)}{n_{c_L}(\theta_{\rm prop})}\cdot
    \frac{s_{c_S}(\theta)}{s_{c_S}(\theta_{\rm prop})}.
\]
The weight depends on an entry only through $(M,c_L,c_S)$, and each factor is
computed exactly. The mass term uses the gridded density from which
\texttt{genulens} samples ($1{,}000$ log-mass bins, linear within a bin). $M$ is
the lens's mass for stars and brown dwarfs and follows from \texttt{genulens}'s
deterministic initial--final mass relation for white dwarfs; neutron stars and
black holes descend from $M>9\,M_\odot$, where the mass term is constant. The
component terms come from \texttt{genulens} itself. Because the profiles do not
depend on $\theta$, each ratio is the same at every distance and on every
sightline (to within $2\times10^{-6}$), so one call per candidate tabulates
them. A component
that is switched off in the model has zero density at every $\theta$, is never
drawn, and is excluded.

\paragraph{Rejection.} Let $B_\theta\ge\sup_z w_\theta(z)$ be the product of the
three factors' maxima (the mass term peaks at a grid point, the component terms
over the components present). Accepting a bank entry with probability
$w_\theta/B_\theta$ yields an entry with density proportional to
$p_{\theta_{\rm prop}}w_\theta\propto p_\theta$, that is, an exact draw from
the model at $\theta$; the constant $Z_\theta$ is never needed.

\paragraph{References.} All candidates draw their references from the same
banks. For each field, the bank is put in one random order, shared by all
candidates together with the uniform draws of the rejection step, and each
candidate keeps the entries it accepts. Since candidates accept entries at
different rates, we align them by step: at each event of the field, the shared
order is scanned until every candidate has accepted $K$ entries, each candidate
takes its first $K$, and the next event starts after the last entry scanned.
Every candidate therefore draws its references from entries that no candidate
has used at an earlier step, and never uses an entry twice. Its references are
i.i.d.\ draws from its own model, independent of the observations and of all
earlier references, as \cref{prop:exact-martingale} requires for validity.

\paragraph{Bank sets.} We use eleven pairs of independently simulated bank
sets. In each pair, one set supplies references ($1.2\times10^7$ raw events per
field) and the other supplies simulated observations ($6\times10^6$), so
observations never share entries with references. One pair serves the real
data; its simulation set also supplies the NRE training data and the KDE score
(\cref{app:microlensing-cs}), and is independent of the other ten pairs, which
are rotated over the semi-synthetic runs. The rotation matters because a bank
is itself a random sample: references drawn from it are exact draws from the
model only on average over the bank, and the coverage guarantee includes this
randomness.

\paragraph{Validation.} At nine values of $\theta$, including all six corners of
$\{-2,0\}\times\{-2,1.2,10\}$ and the published fit, and in four fields spanning
the survey, we compared reweighted entries of a production bank set with
independent direct \texttt{genulens} simulations: Kolmogorov--Smirnov tests on
$\log t$ and chi-square tests on the lens component, the source component, and
the lens class. None of the $144$ tests has $p<0.01$ (the smallest is $0.019$),
and median timescales agree to $0.9\%$. With the mass term alone, $35$ of $80$
comparable tests have $p<10^{-4}$.

\subsection{Model checking}
\label{app:microlensing-fit-check}

To assess whether the model is well specified, we fit $\theta$ on the held-out
even seasons and test the fitted model on the odd seasons, which are the ones
analysed. The test compares the observed distribution of $\log_{10}\hat t_E$
with the one the model predicts: we bin the observed values at their quantiles,
simulate the expected count in each bin conditionally on every event's
covariates, and compute Pearson's $\chi^2$ statistic. On synthetic data from the model, this
statistic follows its nominal $\chi^2$ distribution. Fitting
$(\theta_1,\theta_2)$ on the even seasons gives $(-1.1,1.0)$, and testing it on
the odd seasons gives $p=0.09$ with 20 bins and $p=0.19$ with 40 bins, so the
data show no evidence against the model. The fitted $\theta_1\approx-1.1$ agrees
with the published $-1.13$; $\theta_2$ is weakly identified, possibly due to
low-cadence fields detecting few short events.

\subsection{Confidence sequences and baselines}
\label{app:microlensing-cs}

We use $G=2{,}048$ scrambled Sobol candidates on the box, $K=15$ references,
$\alpha=0.1$, the minimum-evidence alternative, and an NRE-C score trained as in
\cref{app:training} on $(\theta,X,Y)$ with covariates drawn from the empirical
pool of all $3{,}514$ events independently of $\theta$, so that the network
estimates the conditional ratio (\emph{conformal NRE-C CS}).

\paragraph{Conformal KDE.} Because $Y$ is one-dimensional given $X$, its
conditional density $\hat p_\theta(Y_i\mid X_i)$ can be estimated accurately by
a kernel density estimate (KDE) from the simulation set of the real-data pair,
reweighted to $\theta$, with a kernel width that includes the event's
measurement error $\sigma_i$. Its log-ratio is a
second, near-exact conformal score, which gives the \emph{conformal KDE CS}.
Computed from a bank independent of the references
and of any bank that generated the observations, it is positive and predictable, so the sequence stays valid. The NRE and kernel scores are run together, sharing
observations, orderings, and references, so the comparison is paired. 

\paragraph{Baselines.} Recomputed at every $n$ on the same candidates and
orderings, the baseline is the Wilks set
$\{\theta:2[\max_\vartheta\hat\ell_n(\vartheta)-\hat\ell_n(\theta)]\le\chi^2_{2,0.9}\}$
with $\hat\ell_n=\sum_{i\le n}\log\hat p_\theta(Y_i\mid X_i)$ from the kernel
likelihood (\emph{Wilks KDE CI}); we also compare with the published values of \citet{Koshimoto2021}.

\paragraph{Semi-synthetic runs.} These use the even-season fit
$\thetatrue=(-1.1,1.0)$ as the truth. It replaces the nearest of the $2{,}048$
Sobol candidates, so that it is itself a candidate and coverage can be easily assessed. Each repetition takes the
covariates of the $1{,}798$ analysed events in a fresh random order and draws
each response from $q_{\thetatrue}(\cdot\mid X_i)$, with measurement error,
using distinct entries of a simulation bank, so the observations are exactly
i.i.d.\ given the covariates. 
Validity holds conditionally on the covariate
sequence (\cref{cor:conditional-validity}), so fixing the covariates loses
nothing. 
The $300$ repetitions are split into ten groups of $30$, each with its
own independent pair of banks. A repetition miscovers at time $n$ if $\thetatrue$ has been
excluded from the confidence set by then.

\paragraph{Real-data runs.} These run the confidence sequences on $200$ random
orderings of the analysed events. With no known truth, we report $\theta_1$ and
$\theta_2$ separately, through the projections of the confidence set onto each
(\cref{app:nuisance-projections}): for each endpoint of each projection we give its mean
over orderings with a normal $95\%$ interval for that mean.

\subsection{Results}
\label{app:microlensing-results}

\paragraph{Semi-synthetic runs.} Both the conformal NRE-C CS and the conformal
KDE CS keep their empirical cumulative miscoverage below $\alpha=0.1$ at every
$n$, consistent with the theory (\cref{tab:microlensing-semi});
\begin{table}[h]
\centering
\caption{OGLE-IV microlensing, semi-synthetic runs ($300$ repetitions at
$\thetatrue=(-1.1,1.0)$). \emph{Miscoverage}: cumulative miscoverage over
$n\le1{,}798$, with a Wilson $95\%$ interval. \emph{Retained}: mean fraction of the
candidates still in the set at $n=1{,}798$. \emph{$W_{1798}$ ($\theta_1$)}: mean width of
the projection onto $\theta_1$ at $n=1{,}798$, over repetitions with a nonempty
set (the $\theta_1$ range of the box has width $2$). Brackets give Wilson
$95\%$ intervals for miscoverage and normal $95\%$ intervals for the means. All
conformal sequences share observations and references.}
\label{tab:microlensing-semi}
\small
\setlength{\tabcolsep}{4pt}
\begin{tabular}{lccc}
\toprule
Set & Miscoverage & Retained & $W_{1798}$ ($\theta_1$)\\
\midrule
Conformal NRE-C CS (minimum evidence) & $0.080$ $[0.054,0.116]$ & $0.181$ $[0.174,0.188]$ & $0.69$ $[0.67,0.71]$ \\
Conformal KDE CS (minimum evidence) & $0.063$ $[0.041,0.097]$ & $0.178$ $[0.172,0.184]$ & $0.67$ $[0.65,0.68]$ \\
Conformal NRE-C CS (oracle alternative) & --- & $0.108$ $[0.105,0.111]$ & $0.49$ $[0.48,0.50]$ \\
Conformal KDE CS (oracle alternative) & --- & $0.104$ $[0.101,0.108]$ & $0.46$ $[0.45,0.46]$ \\
Wilks KDE CI, monitored & $0.53$ $[0.47,0.58]$ & $0.103$ $[0.100,0.105]$ & $0.43$ $[0.42,0.43]$ \\
\bottomrule
\end{tabular}
\end{table}
empty final sets, in $0.7\%$ of
repetitions with each score, count as miscoverage. The oracle versions use
$\thetatrue$ as the alternative, so they can never exclude it, and we report no
miscoverage for them. The Wilks KDE CI is valid at a fixed $n$ (miscoverage $0.09$ at
$n=1{,}798$) but not under continuous monitoring.

For $\theta_2$, whose true value is $1.0$, the confidence sequences show that it
is weakly identified: after all $1{,}798$ events they exclude only values below
about $-1$ (median lower end $-0.95$ with the NRE score), that is, mass
functions with many brown dwarfs, but still reach near the top of the box, where
brown dwarfs disappear: their upper end exceeds $9$ in $96\%$ of repetitions
with the NRE score and $99\%$ with the KDE score.
We therefore compare the
methods' efficiency through their width in $\theta_1$.

For $\theta_1$,  NRE-C score is almost as efficient as the near-exact KDE score: the two
paired sequences end up with nearly the same width, so learning the
likelihood ratio with a network has a low efficiency cost in this experiment. Replacing the minimum-evidence
alternative by the oracle shrinks the sequence substantially, showing that most
of the extra width comes from estimating the alternative; even with the oracle,
the sequence remains slightly wider than the fixed-$n$ Wilks KDE CI, reflecting the cost of anytime validity.

\paragraph{Real-data runs.} After all $1{,}798$ events, the conformal KDE CS
is nonempty in $99.5\%$ of the $200$ orderings and the conformal NRE-C CS in all of
them. For $\theta_1$, both sequences shrink steadily with $n$ and contain the
published value $\theta_1=-1.13$ \citep{Koshimoto2021} in $98\%$ (NRE-C) and $99\%$
(KDE) of orderings (and $\theta_2=-0.18$ in $98.5\%$ and $99\%$).
As in the semi-synthetic runs, the NRE-C score is almost as efficient as the KDE
score, and both sequences are wider than the Wilks KDE CI, as expected. As in the semi-synthetic runs, the confidence sequences learn little
about $\theta_2$: they exclude only low values, whereas the Wilks KDE CI also
excludes high ones, including the limit with no brown dwarfs.

\end{document}